\documentclass[12pt]{article}

\usepackage{amsmath,amsthm,amssymb} %
\usepackage[all]{xy} %
\usepackage{graphicx} %
\usepackage{verbatim} %

\numberwithin{equation}{section} %

\theoremstyle{plain}
\newtheorem{thm}{\bf Theorem}[section] %
\newtheorem{lem}[thm]{\bf Lemma} %
\newtheorem{cor}[thm]{\bf Corollary} %
\newtheorem{prop}[thm]{\bf Proposition} %
\newtheorem{obs}[thm]{\bf Observation} %

\theoremstyle{definition}
\newtheorem{defn}{\bf Definition}[section] %
\newtheorem{nota}[defn]{Notation} %
\newtheorem{exmp}{\bf Example}[section] %

\theoremstyle{remark}
\newtheorem*{note}{\bf Note} %

\begin{document}

\title{Contributions to Persistence Theory}

\author{Dong Du}

\maketitle

\begin{abstract}

Persistence theory discussed in this paper \footnote{This paper is the Ph.D thesis written under the direction of D. Burghelea at OSU.} is an application of algebraic topology (Morse Theory)  to Data Analysis, precisely to qualitative understanding of point cloud data.
Mathematically a point cloud data is a finite metric space of a very large cardinality. It can be geometrized as a filtration of simplicial
complexes and the homology changes of these complexes provide qualitative information about the data.
There are new invariants which permit to describe the changes in homology (with coefficients in a fixed field) and these invariants are the ``bar codes''.

In Section \ref{sec 3}  we develop additional methods for the calculation of bar codes and their refinements.
When the coefficient field is $\mathbb{Z}_2$, the calculation of bar codes is done by ELZ algorithm (named after H. Edelsbrunner,
D. Letscher, and A. Zomorodian). When the coefficient field is $\mathbb{R}$, we propose an algorithm based on the Hodge decomposition.

The original persistence theory  involves the  ``sub-level sets'' of a nice continuous map (tame map). With Dan Burghelea
and Tamal Dey we developed a persistence theory which involves level sets discussed in Section 4. This is a refinement of the original persistence.
The level persistence we propose is an alternative to Zigzag persistence considered by G. Carlsson and V. D. Silva. We introduce and
 discuss new computable invariants,  the ``relevant level persistence numbers''  and the ``positive and negative bar codes'',
  and explain how they are related to the known ones.

We provide enhancements and modifications of ELZ algorithm to calculate such invariants and illustrate them by examples.

Sections 3 is preceded by background materials (Section \ref{sec 2}) where the concepts of algebraic topology used in this paper are defined.

\end{abstract}

\section{Introduction} \label{sec 1}

We view ``Persistence Theory''  as the Computer Science friendly  application of Morse Theory  to Data Analysis.

      \begin{itemize}

      \item

      \textit{Data} considered in this thesis is  called PCD (point cloud data) which is, mathematically, a \textit{finite metric space} $(\Sigma, d)$
       of very large cardinality. A PCD often appears as a collection of points in some Euclidean space $(\mathbb R^D,d)$ with $d$
        being the Euclidean distance. It is hard to visualize a PCD or to study its structure directly, since the cardinality is large and dimension
        of the Euclidean space in which it embeds can be more than three. However, we can understand qualitative features of a PCD by
        analyzing the family of simplicial complexes associated to it. In this thesis we use Vietoris-Rips complexes which were introduced
        first by Vietoris then by Rips in connection with the group theory (Vietoris \cite{V},   Hausmann \cite{H95}).

       \item

       Morse Theory \cite{M63} is part of geometric and algebraic topology, whose purpose is  to describe the topology of reasonably nice
       spaces  $X$, for example smooth manifolds,   with the help of a generic real valued functions $f:X\to \mathbb R$, for example proper smooth
       functions with all critical points nondegenerate. It uses the local changes in the homology of  the sub-levels $X_{-\infty,t} = f^{-1}((-\infty,t])$
        to  describe the global topology (homology) of the space  $X$. The topology  of the sub-levels for such generic functions changes for a discrete
        collection of $t$'s, called critical values. If $X$ is compact then the topology changes for a finite collection of $t$'s, say $t_0<t_1<\cdots<t_N,$
         and a finite filtration $X_{-\infty,t_0}\subseteq \cdots\subseteq X_{-\infty,t_N} = X$ is obtained.

         \end{itemize}

Persistence theory, as considered by  H. Edelsbrunner, D. Letscher, A. Zomorodian \cite {ELZ},  P. Frosini and C. Landi  \cite{FL},
  also anticipated in Rene Deheuvels \cite {D55},  provides a  slight change of perspective. The object to start with is a space $X$ equipped with a finite filtration rather than a function $f:X \to \mathbb R$. This leads to the concepts  of \textit{persistence vector spaces}, and then to \textit{linear algebra of persistence vector spaces}, and to additional type of invariants,  \textit{bar codes} \cite{CCGZ}, \cite{ZC}. (As recognized by  Carlsson and Zamorodian \cite{ZC} the concept of bar codes corresponds to ``torsion and rank'' of  a finitely generated graded module over the ring of polynomials of one variable. This concept was not previously used in the work of computer scientists.)

When the underlying space $X$ is a simplicial or more general polytopal complex,
algorithms of reasonably low complexity \cite{CEM}, \cite{Mo05}, \cite{KB79} as proposed in \cite{ELZ}, \cite{ZC}  can be used to calculate efficiently the \textit{bar codes} of the  persistence vector space associated with a filtration of the  complex. From the bar codes, one can derive the homology of the sub-levels and then of the underlying space.

\textit{Relationship with Data analysis}

The Persistence Theory becomes relevant to Data Analysis due  to  remarkable new ideas: 

One can understand the qualitative features of a PCD by analyzing geometric shapes, in this case simplicial complexes, associated to it.
One associates to the finite metric space $(\Sigma,d)$ and $\epsilon\geq 0$ a simplicial complex $R_{\epsilon}(\Sigma)$ (so called ``Vietoris-Rips complex''). Since a single Vietoris-Rips complex $R_{\epsilon}(\Sigma)$ does not contain enough information of $(\Sigma,d)$, one considers all of these complexes for all $\epsilon\geq 0$. Fortunately only finitely many of them are different and they form a finite filtration of
 the standard simplex of dimension the cardinality of $\Sigma$ less one unit.  The result of this analysis is encoded in bar codes, which carry information about the qualitative features of the data and  signal the missing parts and  ``accidental noises'' in the data. 
 
 This theory had already nice applications in many areas of science cf. Carlsson \cite {C} and one expects much more to follow.

\vskip .2in

Section \ref{sec 2}, Mathematical Preliminaries, provides a concise presentation of the mathematical concepts behind
the Persistence Theory. In Section \ref{sec 2} we recall the definitions of simplicial complexes, simplicial homology, polytopal complexes,  cellular homology, singular homology and Betti numbers  as well as the methods to calculate them when the coefficient field is $\mathbb{Z}_2$ or $\mathbb{R}$.

When the field is $\mathbb{Z}_2$, the calculation is done by the  \textit{ELZ algorithm} (H. Edelsbrunner, D. Letscher, and A. Zomorodian) \cite{ELZ}.

When the field is $\mathbb{R}$, the calculation can be done by the \textit{Hodge decomposition}.
This is a new method, at least as long as the calculation of bar codes is concerned,  based on the elementary Hodge theory \cite{E}.

\vskip .2in

In Section \ref{sec 3},  Persistent Homology of a PCD, one reviews the basic algebra  of (tame) persistence vector spaces and introduces the
\textit{bar codes} of a tame persistence vector space. These are by now standard concepts but our presentation is slightly different from the existing literatures.

From a point cloud data (PCD), $(\Sigma, d)$, one derives a filtration of simplicial complexes (Vietoris-Rips complexes mentioned before)
$$
R_{\epsilon_0}(\Sigma) \subseteq R_{\epsilon_1}(\Sigma) \subseteq \cdots \subseteq R_{\epsilon_N}(\Sigma).
$$

Then $r$-dimensional homology groups with coefficients in a fixed field of the spaces of this filtration with the linear maps induced by inclusions define
a tame persistence vector space. The ``\textit{bar codes}'' of this tame persistence vector space are referred to as the bar codes of the PCD.
For the definition of bar code see Subsection \ref{secPLA}.

When the coefficient field of homology groups is $\mathbb{Z}_2$, the calculation of bar codes is done by the  ELZ algorithm.

When the coefficient field of homology groups is $\mathbb{R}$, the calculation of bar codes is done by a slight generalization of the Hodge decomposition as described  in Section \ref{sec 2}. \textit{This is a new contribution.}

The \textit{simultaneous persistence} used in Section \ref{sec 4} in relation with \textit{level persistence} \textit{is also a new contribution}, and requires appropriate modifications\,/\,improvements of the ELZ algorithm.

\vskip .2in

In Section \ref{sec 4},  Persistence Theory Refined, one considers  continuous maps $f:X\to \mathbb R$.
When $f$ is weakly tame (Definition \ref{weakly tame}), a finite filtration $
X_{-\infty,t_0} \subseteq X_{-\infty,t_1} \subseteq \cdots \subseteq X_{-\infty,t_N}
$ is defined, where $X_{-\infty,t_i} = f^{-1}((-\infty,t_i])$ and $t_0<t_1<\cdots<t_N$ are critical values.
 The homology groups in each dimension of the above filtration provide a persistence vector space whose bar codes are referred to as the bar codes of the (weakly tame) map $f$ in that dimension. In Section \ref{sec 4} the ``persistence'' for such filtration (which is the standard persistence considered in Section \ref{sec 3}) is referred  to  as the \textit{sub-level persistence} and its  invariants as the \textit{bar codes for the sub-level persistence} of $f$.
The sub-level persistence analyzes the changes in the homology of sub-levels.

In Section \ref{sec 4} we refine  the sub-level persistence to \textit{level persistence}.
The level persistence considers the changes in the homology of the level sets $X_t= f^{-1}(t)$ rather than the changes in the homology of the sub-level sets $X_{-\infty,t} = f^{-1}((-\infty,t])$.
In a more primitive form the level persistence was first considered in \cite {DW} under the name ``interval persistence''. For tame maps (Definition \ref{tamemap}) the level persistence is equivalent to Zigzag persistence previously introduced in \cite {CS}.

In this work the maps considered will be tame, i.e. the topology of the levels changes at finitely many $t$.
The theory requires the tameness hypotheses, however almost all continuous maps and in particular all maps of practical interest  like simplicial maps or generic smooth maps on smooth manifolds are tame. Tame maps are weakly tame (Corollary \ref{tametoweaklytame}).

The level persistence determines and is determined by
the \textit{relevant persistence numbers}  introduced in Section \ref{sec 4}.
For a tame map the relevant persistence numbers are equivalent to the bar codes for the Zigzag persistence, a concept introduced by Carlsson and Silva \cite{CS}. We call these bar codes \textit{bar codes for level persistence} and also provide their definition from a slightly different perspective.

For  a tame map the relevant level persistence numbers  consist of a  finite collection of numbers while the bar codes for level persistence
 (equivalently for Zigzag persistence)  consist of a finite collection of intervals of four types: closed, open, left open right closed and left closed
 right open. For a tame map $f:X\rightarrow \mathbb{R}$ they carry considerably more information than the bar codes for sub-level persistence of the map $f$. In fact the former implies the latter as
 explained and illustrated by examples in Section \ref{sec 4}.

There are two fundamental concepts  in the level persistence: \textit{death} and \textit{detectability} (or \textit{observability}). These concepts should be compared with \textit{birth} and \textit{death}  in sub-level persistence, however are not quite the same.

We also introduce  in Section \ref{sec 4} the concept of positive and negative bar codes
since they are related to relevant persistence numbers and can be calculated
efficiently when the underlying space of the map is a simplicial
complex and the map is linear on each simplex.

 We reduce the calculation of bar codes for level persistence or of relevant persistence numbers to the calculation of bar codes for sub-level
 persistence. However the sub-level persistence is not for the map $f$ but for other maps associated with $f$.  These associate maps are derived via the
 construction of \textit{cuts along levels}. For this purpose we provide new algorithms  to calculate ``cuts along levels'', of interest to
  computational geometry, and  improve on the existing algorithms (e.g ELZ).

Section \ref{sec 4} also contains a number of  examples and pictures to describe the implementation of the algorithms.

The level persistence as presented in this thesis was also considered in papers \cite{BDD} and \cite{DW}.

This paper is the Ph.D thesis written under the direction of D. Burghelea at OSU.

\newpage

\section{Mathematical Preliminaries} \label{sec 2}

In this section we review the definitions of
simplicial complex\,/\,simplicial homology\cite{H}\,/\,polytopal complex\,/\,cellular homology\cite{Me}
and singular homology and the methods to calculate these homologies when
the coefficient field \footnote{The homology considered here is always in a fixed field.} is $\mathbb{Z}_2$ or $\mathbb{R}$.

\subsection{Basic Knowledge}

\textbf{Affinely Independent} \cite{C89}:

A set of points (vectors) $x_0, \cdots, x_k \in \mathbb{R}^D$ is \textit{affinely independent}
  iff the vectors
  $$
 x_1 - x_0,x_2 - x_0, \cdots,  x_k - x_0
  $$ are linearly independent.

\vskip .3 cm

\textbf{Convex Hull} \cite{BCKO}:

Let $X = \{x_1,\cdots,x_k\}$ be a finite set of points(vectors) in $\mathbb{R}^D$.
The \textit{convex hull} of $X$ is
the set of all convex combinations of the elements of $X$
$$
H_{convex}(X) = \{\; \sum_{i=1}^k\alpha_i x_i \mid \alpha_i\in \mathbb{R},\alpha_i\geq 0, \sum_{i=1}^k\alpha_i = 1 \;\}
$$

\vskip .3 cm

\textbf{Affine Hull} \cite{D08}:

Let $X = \{x_1,\cdots,x_k\}$ be a finite set of points (vectors) in $\mathbb{R}^D$.
The \textit{affine hull} of $X$ is
the set of all affine combinations of the elements of  $X$
$$
\text{aff}(X) = \{\; \sum_{i=1}^k\alpha_i x_i \mid \alpha_i\in \mathbb{R}, \sum_{i=1}^k\alpha_i = 1 \;\}
$$

\vskip .3 cm

\textbf{Geometric Simplices in a Real Vector Space}{\label {Simplex}} \cite{M}:

An \textit{$n$-geometric simplex} $\sigma$ in $\mathbb{R}^D$ for $n\leq D$
is the convex hull of $n+1$ affinely independent points
$v_0,\cdots,v_n \in \mathbb{R}^D$ and sometimes denoted by $[v_0,\cdots,v_n]$. The points
$v_0,\cdots,v_n$ are called \textit{vertices} of $\sigma$.

The convex hull of a subset of vertices of $\sigma$ containing $m+1$
points is called an \textit{$m$-face} of $\sigma$.
The $m$-faces of $\sigma$ are also $m$-geometric simplices.

\vskip .3 cm

\textbf{Standard $n$-Simplex} \cite{H}:

The \textit{standard $n$-simplex} $\Delta^n \in \mathbb{R}^{n+1}$ is the convex hull of
vertices
$$
v_0 = (1,0,\cdots,0),  \cdots, v_n = (0,\cdots,0,1).
$$

\vskip .3 cm

\textbf{Geometric Simplicial Complexes} \cite{M}:

A \textit{geometric simplicial complex} $\mathcal{K} = \{\; \sigma_\iota \mid \iota\in \mathcal A \;\}$ in $\mathbb{R}^D$ is a collection  of geometric simplices
in $\mathbb{R}^D$ that satisfies the following conditions:

1. Any face of a simplex of $\mathcal{K}$ is also in $\mathcal{K}$.

2. The intersection of any two simplices $\sigma_1,\sigma_2\in \mathcal{K}$ is a face of both $\sigma_1$ and $\sigma_2$.

\vskip .3 cm

\textbf{Underlining Space of a Geometric Simplicial Complex} \cite{M}:

Let's denote the union of simplices of a geometric simplicial complex $\mathcal{K}$ as $|\mathcal{K}|$
and call it the \textit{underlying  space} of $\mathcal{K}$.

\vskip .3 cm

\textbf{Simplicial Maps Between Geometric Simplicial Complexes} \cite{M}:

Let $\mathcal{K}$ and $\mathcal{L}$ be two geometric simplicial complexes.
A map $f:|\mathcal{K}|\rightarrow |\mathcal{L}|$ is a \textit{simplicial map} if it sends each simplex of
$\mathcal{K}$ to a simplex of $\mathcal{L}$ by a linear map taking vertices to vertices.

\vskip .3 cm

\textbf{Abstract Simplicial Complexes} \cite{M}:

An \textit{abstract simplicial complex} $\mathcal{X} = (V,  \mathcal{S})$ consists of a set $V$ of vertices and a set $\mathcal{S}$ of finite subsets of $V$ called \textit{abstract simplices},
which satisfy the following property:
if $\sigma\in \mathcal{S}$ and $\tau\subseteq\sigma$, then $\tau\in \mathcal{S}$.

If $\sigma \in \mathcal{S}$ has $n+1$ elements, the dimension of $\sigma$ is $n$ and $\sigma$ is called an $n$-simplex.

If $\tau\subseteq\sigma$ has $m+1$ elements, $\tau$ is called an $m$-face of $\sigma$.

\vskip .3 cm

\textbf{Simplicial Maps Between Abstract Simplicial Complexes} \cite{M}:

Let $\mathcal{X}_1 = (V_1,  \mathcal{S}_1)$ and $\mathcal{X}_2 = (V_2,  \mathcal{S}_2)$ be two abstract simplicial complexes.
A map $f:V_1\rightarrow V_2$ is a \textit{simplicial map} if whenever the vertices $v_0, \cdots, v_n$ of $\mathcal{X}_1$
span a simplex, the points $f(v_0), \cdots, f(v_n)$ are vertices of a simplex of $\mathcal{X}_2$.

\vskip .3 cm

\textbf{Spatial Realization of an Abstract Simplicial Complex} \cite{M}:

A spatial realization $\mathcal{K} = <\mathcal{X}>$ of an abstract simplicial complex $\mathcal{X}$
 is provided by an injective map $f: V\rightarrow \mathbb{R}^D$ that sends each simplex in $\mathcal{X}$
  to an affinely independent set in $\mathbb{R}^D$ so that no interior of different geometric simplices intersect (in particular if $f(v_i)$ are all affinely independent).
  The union of the convex hulls of the images of all abstract simplices of $\mathcal{X}$ is a geometric simplicial complex, which
  is a spatial realization of $\mathcal{X}$.

\vskip .3 cm

Spatial realizations always exist for $D$ large enough and their underlying spaces are all homeomorphic to
each other.
  If we have a geometric simplicial complex $\mathcal{K}$, we can also get an abstract simplicial complex $\mathcal{X}$
  with $V$ the set of vertices of $\mathcal{K}$ and
  $$
  \mathcal{S} = \{\; \sigma\subseteq V \mid \sigma \text{ are vertices of some
  simplex of  }\mathcal{K} \;\}.
  $$
  Then $(V, \mathcal{S})$ defines an abstract simplicial complex $\mathcal{X}$.

It is not hard to check the above assignments are essentially inverses to each other and
they also induce a bijection between simplicial maps, so
an abstract simplicial complex $\mathcal{X}$ and its spatial realization $\mathcal{K}$ can be regarded as equivalent objects.

\vskip .3 cm

\textbf{Oriented Simplex} (page 26, \cite{M}):

Let $\sigma$ be a simplex (geometric or abstract).
One says that two orderings of its vertices are equivalent if
they differ from one another by an even permutation.
If $\sigma$ is a $0$-simplex then, obviously, the ordering is unique. If $\dim \sigma > 0$ then  the orderings of the vertices of $\sigma$
 fall into two equivalence classes.
Each of these classes is called an \textit{orientation} of $\sigma$.
An \textit{oriented simplex} $(\sigma, \epsilon)$ is a simplex $\sigma$ together with
an orientation $\epsilon$.

If the points $v_0,\cdots,v_p$ are affinely independent,
we shall use the symbol
$$
\langle v_0,\cdots,v_p \rangle
$$
to denote the oriented simplex consisting of the simplex $\{v_0,\cdots,v_p\}$
and the equivalence class of the ordering $(v_0,\cdots,v_p)$.

\vskip .3 cm

\textbf{Simplicial Homology with Coefficients in a Field} \cite{H}:

Let $\mathcal{X}$ be an abstract simplicial complex as defined above.
Let $C^{\Delta}_n(\mathcal{X};\kappa)$ be the $\kappa$-vector space generated by  the set of
 $n$-simplices in $\mathcal{X}$,  where $\kappa$ is a fixed field.
Elements of $C^{\Delta}_n(\mathcal{X}; \kappa)$ are finite formal sums $\Sigma_i \lambda_i \sigma_i$ where $\lambda_i\in \kappa$,
$\sigma_i$ are $n$-simplices of $\mathcal{X}$.
By choosing a total ordering of vertices of $\mathcal{X}$ (which provides an orientation on each simplex) a boundary map $\partial^{\Delta}_n: C^{\Delta}_n(\mathcal{X}; \kappa)\rightarrow C^{\Delta}_{n-1}(\mathcal{X}; \kappa)$ is defined as the linear extension of the map below
$$
\partial^{\Delta}_n(\langle v_0, \cdots v_n \rangle) = \sum_i (-1)^i \langle v_0,\cdots,\hat{v}_i,\cdots,v_n \rangle.
$$

It's easy to verify that $\partial^{\Delta}_n\circ\partial^{\Delta}_{n+1} = 0$, so that we can define the
\textit{simplicial homology groups}  as $\mathcal{H}^{\Delta}_n(\mathcal{X}; \kappa) = \ker\partial^{\Delta}_n \;/\;  \textmd{im} \partial^{\Delta}_{n+1}$.

\vskip .3 cm

Choosing different orderings on the vertices of $\mathcal{X}$ or different orientations on each simplex, the boundary maps might change (the sign of entries might change)  but the homology groups remain the same.

\vskip .3 cm

\textbf{Convex Polytope} \cite{G03}:

An \textit{extreme point} of a convex set $S$ in a real vector space is
 a point in $S$ which does not lie in any open line segment
  joining two points of $S$.
  The extreme points of a simplex in a real vector space are its
  vertices.

A \textit{convex polytope} is the convex hull of a finite set of points.
We call extreme points of a polytope its \textit{vertices}.
A simplex is a convex polytope such that its vertices are affinely independent.

Suppose $P$ is a convex polytope with $X$ the set of its vertices.
The dimension of $P$ is  the dimension of the affine hull of $X$.

A point $x \in P$ is an \textit{interior point of $P$}  if it is interior in the sense of point set topology when $P$ is regarded
 as a subspace of $\text{aff}(X)$.

Given $a_1, \cdots, a_D, b \in \mathbb{R}$, consider the closed half-space defined by
$$
\{\; (r_1,\cdots,r_D)\in \mathbb{R}^D \mid a_1 r_1+\cdots+a_D r_D \leq b \;\}
$$

The boundary of the above half-space is
$$
\{\; (r_1,\cdots,r_D)\in \mathbb{R}^D \mid a_1 r_1+\cdots+a_D r_D = b \;\}
$$

A  \textit{face} of a convex polytope is any intersection of the polytope with a closed
half-space such that no interior points of the polytope lie on the boundary of the half-space.

 The dimension of a face is  the dimension of its
affine hull.
Faces of a convex polytope are also convex polytopes.
Vertices of a polytope are its $0$-dimensional faces.

\vskip .3 cm

\textbf{Polytopal Complex} \cite{G03}:

\textit{Polytopal complex}, often named as polyhedra complex or cellular complex, consists of a collection of
convex polytopes in  some Euclidean space $\mathbb{R}^D$, satisfying two conditions:

(i) Every face of a polytope in $\Pi$ is also in $\Pi$.

(ii) The intersection of any two polytopes in $\Pi$ is a face of both.

The polytopes in $\Pi$ are called \textit{cells}. An $n$-dimensional polytope in $\Pi$ is called an $n$-cell.

For example, the set of all faces of a convex polytope defines a polytopal complex.

A geometric simplicial complex is a polytopal complex in which every cell is a geometric simplex.

\vskip .3 cm

\textbf{Oriented Convex Polytope} {\cite{Me}}:

Let $\sigma$ be a $d$-dimensional polytope in $\mathbb{R}^D$ and $V(\sigma)$
be the set of vertices of $\sigma$.
The map $\epsilon:  {V(\sigma) \choose d+1} \rightarrow \{1,-1\}$ is an \textit{orientation} of $\sigma$
if the following conditions are satisfied:

(1) $\epsilon(v_0,\cdots,v_i',\cdots,v_d) = \epsilon(v_0,\cdots,v_i,\cdots,v_d)$ if $v_i$ and $v_i'$
are in the same open half-space of $\mathbb{R}^D$ delimited by one of the supporting hyperplanes
of $v_0,\cdots,\hat{v_i},\cdots,v_d$ (hyperplanes that pass through $v_0,\cdots,\hat{v_i},\cdots,v_d$)
and $\epsilon(v_0,\cdots,v_i',\cdots,v_d) = -\epsilon(v_0,\cdots,$ $v_i,\cdots,v_d)$ if not.

(2) $\epsilon(v_0, v_1, \cdots, v_d) = \text{sign}(\pi)\epsilon(v_{\pi(0)},v_{\pi(1)},\cdots,v_{\pi(d)})$ for every permutation $\pi$.

$0$-dimensional polytope admits a unique orientation.
Every $d$-dimensional($d\geq 1$) polytope admits exactly two distinct orientations.
If $\epsilon$ is one of them, $-\epsilon$ is the other one.

An orientation $\epsilon$ on $\sigma$ induces an orientation $\epsilon|_\tau$ on its codimension-$1$ face $\tau$. The convention is that a base in the supporting hyperplane of $\tau$ followed by a vector pointing towards the interior of $\sigma$  provides a base of the supporting affine space of  $\sigma$ which defines the orientation $\epsilon$.

Let $\sigma$ be an $n$-cell, when $\tau$ is a $(n-1)$-face of $\sigma$, one defines the orientation $\epsilon_{\tau}$ by
$$
\epsilon|_{\tau}(v_0,\cdots,v_{n-1}) := \epsilon(v_0,\cdots,v_{n-1},v_n)
$$
for any $v_0,\cdots,v_{n-1}\in V(\tau)$ and $v_n\in V(\sigma)\setminus V(\tau)$.

\vskip .3 cm

\textbf{Cellular Homology with Coefficients in a Field} {\cite{Me}}:

Let $\Pi$ be a polytopal complex as defined above.
Let $C_n(\Pi; \kappa)$ be the vector space generated by  the set of
 $n$-cells in $\Pi$ with coefficients in a fixed field $\kappa$.
Elements of $C_n(\Pi; \kappa)$ are finite formal sums $\Sigma_i \lambda_i \sigma_i$ where $\lambda_i\in \kappa$ and
$\sigma_i$ are  $n$-cells of $\Pi$.

For any cell $\sigma$ choose an orientation $\epsilon (\sigma)$.

The boundary map $\partial_n: C_n(\Pi; \kappa)\rightarrow C_{n-1}(\Pi; \kappa)$ is defined as the linear extension of the map below
$$
\partial_n(\sigma) = \sum_\tau I_{\sigma, \tau} \;\tau.
$$
$I_{\sigma, \tau} = 0$  if $\tau$ is not a $(n-1)$-face of $\sigma$ and $I_{\sigma, \tau} = \pm 1$  if $\tau$ is a $(n-1)$-face of $\sigma$ and $\epsilon(\sigma)|_{\tau} = \pm \epsilon (\tau)$, where $\epsilon(\sigma)$ and $\epsilon (\tau)$ are orientations of $\sigma$ and $\tau$ respectively.

We can verify that $\partial_n\circ\partial_{n+1} = 0$, so the
\textit{polytopal homology groups} are defined as
$\mathcal{H}_n(\Pi; \kappa) = \ker\partial_n \;/\;  \textmd{im} \partial_{n+1}$.
$\mathcal{H}_n(\Pi;\kappa)$ is independent of the chosen orientation
although the boundary maps $\partial_n$ do depend.

When $\kappa = \mathbb{Z}_2$, orientation becomes irrelevant and a boundary map is defined by
$$
\partial_n(\sigma) = \sum_i \lambda_i \sigma_i
$$
where $\sigma$ is an $n$-cell and $\sigma_i$'s are $(n-1)$-cells, $\lambda_i = 1$ iff $\sigma_i$ is a face of $\sigma$.

\vskip .3 cm

\begin{note}

There is a canonical subdivision of a polytopal complex into a simplicial complex.
Each cell $\sigma$ can be decomposed into a simplicial complex.
Choose an interior point $x_{\sigma}$ of the cell $\sigma$.
The cell $\sigma$ can be decomposed as the union of cones
with base codimension-$1$ faces of $\sigma$ and apex $x_{\sigma}$. 
It's obvious that $2$-dimension cell can be decomposed as a simplicial complex.
By induction, the cell $\sigma$ can be decomposed into a simplicial complex.
This decomposition is possible because all cells are convex polytopes.
Once each cell is decomposed into a simplicial complex,
we can regard the polytopal complex as the union of those
simplicial complexes, hence as a simplicial complex.

Once we decompose a polytopal complex into a simplicial complex,
we can calculate the homology groups of this simplicial complex.
In fact we can use this simplicial homology to replace the homology of the polytopal complex,
since the underlying spaces of the simplicial complex and polytopal complex are
the same and two complexes with the same underlying spaces have the same
homology groups (as pointed below).

It is easier to write programs to get the boundary maps and calculate the homology groups of a simplicial complex than
a polytopal complex.
However, it costs much more time to  calculate the homology groups
by first decomposing a polytopal complex into a simplicial complex
 than to calculate directly without decomposition, since the decomposition increases number of cells greatly.

\end{note}

It turns out that the homology of a simplicial or polytopal complex depends only on the underlying topological space.
The easiest way to see this is to define the homology for a topological space, independent of simplices or cells, and
 show that it leads to the same results as the one described above.
 The homology defined in this way is known as \textit{singular homology} and is defined below.

\begin{defn} \label{sigularhomology}
\textbf{Singular Homology with Coefficients in a Field}(see page 108, 153 \cite{H})

A \textit{singular $n$-simplex} in a topological space $X$ is defined as a continuous map $\sigma:\Delta^n\rightarrow X$ on the standard $n$-simplex.
Let $C_n(X; \kappa)$ be the vector space with basis the set of singular $n$-simplices in $X$ and
coefficients in a fixed field $\kappa$.
Elements of $C_n(X; \kappa)$ are finite formal sums $\Sigma_i \lambda_i \sigma_i$ for $\lambda_i\in \kappa$ and
$\sigma_i:\Delta^n\rightarrow X$.
A boundary map $\partial_n: C_n(X; \kappa)\rightarrow C_{n-1}(X; \kappa)$ is defined by the formula:
$$
\partial_n(\sigma) = \sum_i (-1)^i \sigma|[v_0,\cdots,\hat{v}_i,\cdots,v_n]
$$
where  $\sigma: \Delta^n\rightarrow X$ is a singular $n$-simplex, $v_0,\cdots,v_n$ are vertices of $\Delta^n$,
$\hat{v}_i$ means $v_i$ is deleted.
Implicit in this formula is the canonical identification of $[v_0,\cdots,\hat{v}_i,\cdots,v_n]$
with $\Delta^{n-1}$, preserving the ordering of vertices, so that
$ \sigma|[v_0,\cdots,\hat{v}_i,\cdots,v_n]$ is regarded as a map
$\Delta^{n-1}\rightarrow X$, that is, a singular $(n-1)$-simplex.

It's easy to verify that $\partial_n\circ\partial_{n+1} = 0$, so that we can define the
\textit{singular homology group with coefficients in $\kappa$} as $\mathcal{H}_n(X; \kappa) = \ker\partial_n \;/\;  \textmd{im} \partial_{n+1}$.

$\mathcal{H}_n(X; \kappa)$ is a vector space since $\kappa$ is a field.

\end{defn}

\begin{defn}
\textbf{Betti Numbers}

Given a field $\kappa$ one can define $b_n(X; \kappa)$,
the $n$-th Betti number with coefficients in $\kappa$, as the dimension of the vector space
$\mathcal{H}_n(X; \kappa)$.

\end{defn}

We will show how to calculate Betti numbers of a polytopal complex $\Pi$ with $\mathbb{Z}_2$ or $\mathbb{R}$ coefficients.

When the coefficients are in  $\mathbb{Z}_2$, we have the following method referred to below as the \textit{ELZ algorithm} (\cite{EH}).

First order all cells, so that $i<j$ if $\sigma_i$ is a face of $\sigma_j$.
Then form a boundary matrix $\partial$, so that the entry on $i$-th row and $j$-th column
is
$$
\displaystyle\partial[i, j] = \left\{ \begin{array}{ll}
1 & \text{if $\sigma_i$ is a codimension-$1$ face of $\sigma_j$;}\\
0 & \text{otherwise.}
\end{array}\right.
$$
 The matrix $\partial$ is an upper triangular matrix with zeros on diagonal.
Let $low(j)$ be the row index of the lowest $1$ in column $j$.
If the entire column is zero, then $low(j)$ is undefined.
We call a matrix \textit{reduced} if $low(j)\neq low(j_0)$ whenever $j$ and $j_0$, with $j\neq j_0$,
specify two non-zero columns.
The algorithm below starts with the boundary matix $\partial$ and changes it by adding columns from left to right.
All matrices in this process are upper triangular with zeros on diagonal.
Finally we get a reduced matrix $R$.

\begin{center}
\textbf{Algorithm 2.1. The ELZ Algorithm}

\begin{tabular}{|l|}
\hline
R=$\partial$\\
for $j = 1$ to $m$ do\\
\;\;\; while there exists $j_0<j$ with $low(j_0) = low(j)$ do\\
\;\;\;\;\;\; add column $j_0$ to column $j$\\
\;\;\; endwhile\\
endfor\\
\hline
\end{tabular}
\end{center}

From  the reduced matrix $R$ one can read off the Betti numbers.
It can be proven \cite {EH} that  the zero columns of $R$ correspond to generators of cycles and
non-zero columns of $R$ correspond to generators of boundaries.
Since the homology group is quotient of cycles over boundaries,
the Betti number is the number of generators of cycle minus
the number of generators of boundary.

Therefore the $k$-th Betti number
$$
\begin{array}{ll}
b_k(\Pi; \mathbb{Z}_2) & = \text{the number of zero columns which correspond to a $k$-cycle} \\
&- \text{the number of non-zero columns which correspond to a $(k+1)$-boundary}
\end{array}
$$

Algorithms for the calculation of Betti numbers with coefficient in any finite field also exist, however they are slightly more complex.

When the coefficients are in  $\mathbb{R}$, we develop a different method using the Hodge decomposition
introduced below.

We will see that in Subsection 2.2., {Observation \ref{betarssintro}}, that the $k$-th Betti number is

$$b_k(\Pi; \mathbb{R}) = \text{number of $k$-cells in $\Pi$} - rank(\partial_{k+1}) - rank(\partial_k).
$$

\subsection{Hodge Decomposition of a Chain Complex with $\mathbb{R}$ Coefficients}\label{secHodge}

The constructions below follow the standard ``Hodge decomposition'' familiar in Riemannian geometry.
This finite dimensional elementary formulation was first considered by B. Eckmann \cite{E}.
One starts with a complex of finite dimensional $\mathbb{R}$-vector spaces
$$
\xymatrix{
\cdots \ar[r]^{\partial_{r+2}} & C_{r+1} \ar[r]^{\partial_{r+1}} & C_r \ar[r]^{\partial_r} & C_{r-1} \ar[r]^{\partial_{r-1}} & \cdots \ar[r]^{\partial_2} & C_1 \ar[r]^{\partial_1} & C_0 \ar[r]^-{\partial_0} & C_{-1}=0
}.
$$
The vector spaces in this complex are equipped with positive definite inner products.
If ${C_r}'s$ are equipped with a base, one can take the unique inner product which makes this base orthonormal.
In the cases under consideration, each complex comes from a finite simplicial\,/\,polytopal complex $L$, with $C_r$ the $\mathbb R$-vector space generated by the $r$-simplices\,/\,$r$-cells of $L$ and $\partial_r$ the boundary maps.
$\partial_{r+1}:C_{r+1}\rightarrow C_r$ are linear operators between inner product spaces.
Let $\delta_r=\partial_{r+1}^*$ be the adjoint operator of $\partial_{r+1}$.

\begin{lem}
If $A$ and $B$ are two finite dimensional inner product spaces, $f:A\rightarrow B$ is a linear map
and $f^*:B\rightarrow A$ is its adjoint, i.e. $\left<b, f(a)\right> = \left<f^*(b), a\right>$ for any $a\in A$ and $b\in B$, then

i) $\ker(f)=(\textmd{im} (f^*))^{\perp}$;

ii) For any $a\in A$, if $f\circ f^*\circ f(a)=0$ then $f(a)=0$.

\end{lem}

\vskip .5cm

Define $\Delta_r=\partial_{r+1}\circ\delta_r+\delta_{r-1}\circ\partial_r:C_r\rightarrow C_r$ for $r\in\mathbb Z_{\geq 0}$, $(C_r)_+=\textmd{im}(\partial_{r+1})$, $(C_r)_- = \textmd{im}(\delta_{r-1})$ and
$H_r=\ker(\Delta_r)$.

\begin{prop} \label{hodge decomposition}

i) $H_r=\ker(\delta_r)\cap \ker(\partial_r)$;

ii) \textbf{(Hodge Decomposition)} $C_r=(C_r)_+\oplus H_r\oplus (C_r)_-$, where
$(C_r)_+$, $H_r$ and $(C_r)_-$ are pairwise orthogonal;

iii) There is a canonical isomorphism
$$
j_r : H_r=\ker(\delta_r)\cap \ker(\partial_r)
\rightarrow \mathcal H_r=\ker(\partial_r)/\textmd{im}(\partial_{r+1}).
$$

\end{prop}

We will give an algorithm for calculating the Hodge decomposition, i.e.,
given a chain complex $\mathcal C(L)$,
we will calculate the three orthogonal projections:
$$
(p_r)_+:C_r\rightarrow C_r ~~with~(p_r)_+(C_r)=(C_r)_+
$$
$$
(p_r)_-:C_r\rightarrow C_r ~~with~(p_r)_-(C_r)=(C_r)_-
$$
and
$$
(p_r)_H:C_r\rightarrow C_r ~~with~(p_r)_H(C_r)=H_r
$$
for each $r\geq0$ respectively.

\vskip 0.5cm

\begin{lem} \label{orth}
Given any $m\times n$ matrix $A$ over $\mathbb R$,
if  the rank of $A$ is $k$ then
there exists an $m\times k$ matrix
$[A]$ of the form
\begin{equation}\label{1.2.1}
\begin{array}{c}
[A]=(v_1, v_2, \cdots, v_k)_{m\times k}
\end{array}
\end{equation}
where $\{v_1,v_2,\cdots,v_k\}$ is a collection of
 orthonormal column vectors which is equivalent to the
collection of column vectors of $A$.
\end{lem}

Two collections of vectors are equivalent if they generate
the same subspace.
Moreover, there is a canonical construction
of such orthonormal column vectors known as Gram-Schmidt Orthonormalization.

\begin{note}
 1. Given a linear map $A:\mathbb R^n\rightarrow \mathbb R^m$, one
view $A$ as a $m\times n$ matrix with respect to the standard basis of
$\mathbb R^n$ and $\mathbb R^m$. The collection of column vectors
of $[A]$ represents an orthonormal basis of $\textmd{im}(A)$.

2. \textbf{Matlab} contains a function \textbf{orth} with input an $m\times n$ matrix $A$ and
output a matrix $[A]$.

3. $[A][A]^T$ is unique although $[A]$ is not(See \textbf{Lemma \ref{given a linear map}}).
\end{note}

\begin{lem}\label{given a linear map}
Given a linear map $A:\mathbb R^n\rightarrow \mathbb R^m$,
the linear map
$$
\begin{array}{rl}
p_A: & \mathbb{R}^m \rightarrow \mathbb{R}^m \\
     & y  \mapsto  [A][A]^T y
\end{array}
$$
is the orthogonal projection on $\textmd{im}(A)$.

\end{lem}

\begin{prop}
Given a chain complex $\mathcal C(L)$ of a finite simplicial\,/\,polytopal complex $L$ over $\mathbb R$
$$
\xymatrix{
\cdots \ar[r]^{\partial_{r+2}} & C_{r+1} \ar[r]^{\partial_{r+1}} & C_r \ar[r]^{\partial_r} & C_{r-1} \ar[r]^{\partial_{r-1}} & \cdots \ar[r]^{\partial_2} & C_1 \ar[r]^{\partial_1} & C_0 \ar[r]^-{\partial_0} & C_{-1} =0
}
$$
Each $\partial_r$ can be regarded as an $n_{r-1}\times n_r$ matrix with respect to the standard basis of $C_r$ and
$C_{r-1}$.
The following linear maps are orthogonal projections onto $(C_r)_+$,$(C_r)_-$
and $H_r$, respectively.
$$
\begin{array}{rrl}
(p_r)_+: & C_r &\rightarrow C_r \\
     & y  &\mapsto  [\partial_{r+1}][\partial_{r+1}]^T y\\

(p_r)_-: & C_r &\rightarrow C_r \\
     & y  &\mapsto  [(\partial_{r})^T][(\partial_{r})^T]^T y\\

(p_r)_H: & C_r &\rightarrow C_r \\
     & y  &\mapsto  (I_{n_r}-[\partial_{r+1}][\partial_{r+1}]^T-[(\partial_{r})^T][(\partial_{r})^T]^T)y
\end{array}
$$
\end{prop}

The linear map
$$
\begin{array}{lrl}
k_r: & \mathcal H_r=\ker({\partial_r})/\textmd{im}({\partial_{r+1}}) & \rightarrow  H_r=\ker(\delta_r)\cap \ker(\partial_r)\\
 & y+\textmd{im}(\partial_{r+1})  & \mapsto (p_r)_H(y)
\end{array}
$$
is the inverse of $j_r$, which verifies {Proposition \ref{hodge decomposition}} iii).

\begin{obs}\label{rankAintro}
The rank of a real matrix $A$ equals to the rank of
$A A^T$ or $A^T A$.
\end{obs}

By the Hodge decomposition, $\dim(H_r) = \dim(C_r) - \dim((C_r)^+) - \dim((C_r)^-)$.
By the above observation,  $\dim((C_r)^+) = rank([\partial_{r+1}][\partial_{r+1}]^T) = rank([\partial_{r+1}])
= rank(\partial_{r+1})$.
We get $\dim((C_r)^-) = rank(\partial_{r})$ similarly.

\begin{obs}\label{betarssintro}
The $r$-th Betti number
$$
b_r(L; \mathbb{R}) = \dim(\mathcal{H}_r(L; \mathbb{R})) = \dim(C_r)-rank(\partial_{r+1})-rank(\partial_r).
$$
\end{obs}

\newpage

\section{Persistent Homology of a PCD} \label{sec 3}

\subsection{Introduction}

Informally, we call a finite set of points $X \subseteq \mathbb{R}^m$ a \textit{point
cloud data} (PCD for short).
A PCD can be regarded as a finite metric space (see Definition \ref{PCD}).

A first new idea in Data Analysis is to regard
a PCD $X$ as a filtration of the standard $n$-simplex $\Delta^n$ ($n = \sharp X - 1$)\footnote {Here and in the remaining part of this text ``$\sharp$'' will denote ``cardinality''.}  via a  construction known as ``Vietoris-Rips complex''
(see Subsection \ref{secPCD}).
The first term of the filtration is the $0$-skeleton of $\Delta^n$ or $X$ itself, the last term is $\Delta^n$
 with the remaining components giving an idea of the qualitative features of the set $X\subseteq \mathbb{R}^m$.

Another new idea is to use homology to describe the topological changes of the components of this filtration
(topological features ``are born'' in some component, and  ``die'' in some other component).
This leads to the persistent homology of this filtered simplicial complex.
The persistent homology provides the tools to measure and
explain the qualitative patterns of a PCD.
The role of the numerical invariants Betti numbers, when the homology of a space (simplicial\,/\,polytopal complex) is considered,
is taken by the invariants \textbf{bar codes}, when the persistent homology of a filtered simplicial\,/\,polytopal complex is under consideration.
The homology considered in this work is always with coefficients in a field $\kappa$ ($\kappa = \mathbb{Z}_2$ or
$\mathbb{R}$), so the homology groups are actually vector spaces.
The linear algebra for persistent homology is ``persistent linear algebra'' discussed in Subsection  \ref{secPLA}.
 The invariant ``dimension'' for a vector space is replaced by the invariant ``bar codes'' for a persistence vector space.
The \textit{bar codes} provides a complete invariant for a persistence vector space (cf. Theorem \ref {two tame}) as the \textit{dimension} provides a complete invariant for a vector space.
 For a filtered simplicial\,/\,cell complex as opposed to simplicial\,/\,cell complex, we will have ``bar codes'' derived from persistent homology
as opposed to Betti numbers derived from homology.
The bar codes can be explicitly calculated by algorithms of the same complexity as the algorithms used for
the calculation of the Betti numbers.
For $\kappa = \mathbb{Z}_2$, the matrix reduction and pairing algorithm described in \cite{EH} and referred below as the ELZ algorithm is the one we will use.
For $\kappa = \mathbb{R}$, we will use elementary Hodge theory to produce a new algorithm for the calculation of the
bar codes.

Subsection \ref{secPCD} discusses PCD and  Vietoris-Rips filtration.
Subsection \ref{secPLA} discusses persistent linear algebra, including persistence vector spaces
and their bar codes.
The essential features of a persistence vector space are the concept of ``birth and death time'' \footnote{Time here means the index of the component of the filtration.}
of its elements and the fact that the ``bar codes'' provides a complete invariant which carries significant  information
about ``birth and death time'' of its elements.
Subsection \ref{secPHBFSCC} defines
persistence vector space $\mathcal H_r({\mathcal K})$ and the bar codes $\mathcal{B(K)}$ of
a filtered simplicial\,/\,polytopal complex ${\mathcal K}$ and
uses the Hodge decomposition to calculate $\mathcal{B(K)}$.
Subsection \ref{secAlg} gives algorithms to calculate the bar codes of a PCD with coefficients
$\kappa = \mathbb{Z}_2$ or $\mathbb{R}$.
In case of the field $\mathbb{Z}_2$, the algorithm was introduced in \cite{ELZ}.
Subsection \ref{secNE} shows a numerical experiment of the bar codes of a PCD with
coefficients $\kappa = \mathbb{R}$.

\subsection{PCD, Vietoris-Rips Complex and Filtration}\label{secPCD}

\vskip .1in

\begin{figure}[h!]
\centering
   \setlength\fboxsep{0.5pt}
   \setlength\fboxrule{0.5pt}
   \fbox{\includegraphics[scale=.78]{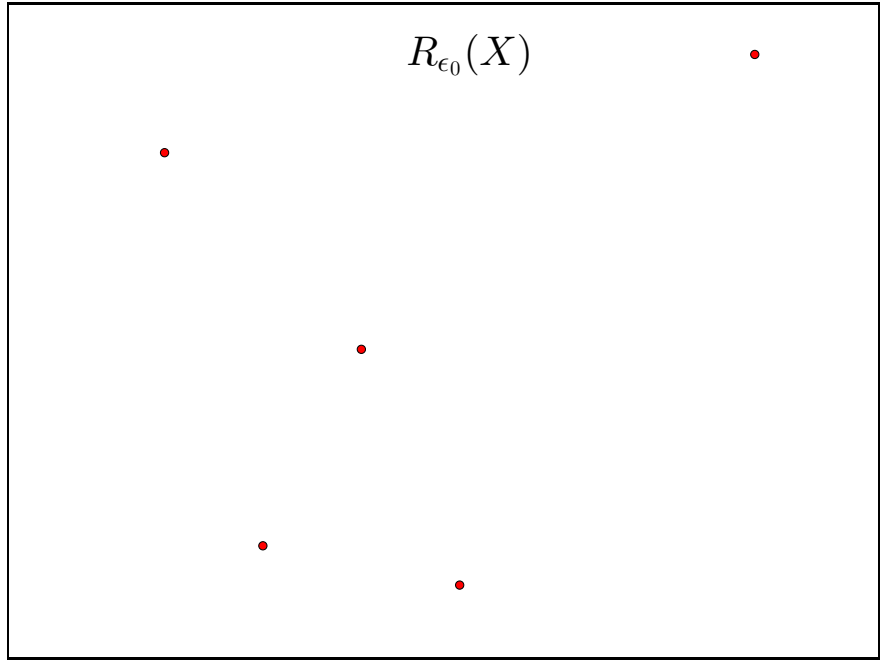}}
     \caption{A PCD in $\mathbb{R}^2$ with $5$ points.} \label{Figure3.1}
\end{figure}

Informally, we call a finite set of points $X \subseteq \mathbb{R}^m$ a point
cloud data (PCD for short). 

\begin{exmp}\label{example221}

A PCD in $\mathbb{R}^2$ with $5$ points (see Figure \ref{Figure3.1} above).

\end{exmp}

The inclusion   $X\subseteq \mathbb R^m$ defines a finite metric space $(X,d)$ with $d$ being the Euclidean distance  in $\mathbb R^m$. From mathematical point of view we will use the following definition

\begin{defn} \label{PCD}
A point cloud data is a finite metric space $X=(X,d)$.
\end{defn}

Let $X$ be a finite metric space.
Given $\epsilon\geq 0$, the Vietoris-Rips complex $R_{\epsilon}(X)$ of
PCD $X$ has $X$ as the  set of verticies \footnote {Other types of simplicial complexes can be associated to the metric space and $\epsilon$ but they calculate essentially the same invariants and are most often less economical. They will not be discussed here.}.
A $k$-simplex is any  subset of vertices $\sigma=[x_0,x_1,\dots,x_k]$
with  the property that  $d(x_i,x_j)\leq\epsilon$
for all pairs $x_i,x_j\in\sigma$.
One obtains in this way an \textit{abstract simplicial complex}.

Notice that Vietoris-Rips complex will be determined by its one skeleton.

If $\epsilon<\epsilon'$ then  there is an inclusion $R_{\epsilon}(X)\hookrightarrow R_{\epsilon'}(X)$.

Since a PCD $X$ is a finite set, Vietoris-Rips complex $R_{\epsilon}(X)$ will change at only finitely many
epsilons
\setcounter{equation}{-1}
\begin{equation}\label{2.2.0}
0 = \epsilon_0<\epsilon_1<\cdots<\epsilon_N = \sup_{x,x'\in X} (d(x, x')).
\end{equation}

We define the filtration of Vietoris-Rips complexes of PCD X as
\begin{equation}\label{2.2.1}
 R_{\epsilon_0}(X)\subseteq R_{\epsilon_1}(X) \subseteq \cdots \subseteq R_{\epsilon_N}(X).
\end{equation}

In this filtration, $R_{\epsilon_0}(X)$ is a zero dimensional simplicial complex
with $\sharp X$ vertices and $R_{\epsilon_N}(X)$ is a $\sharp X - 1$ standard simplex.

\begin{exmp}
For the example of PCD described in Figure \ref{Figure3.1}, one obtain eleven $\epsilon's: \epsilon_0,\epsilon_1,\cdots,\epsilon_{10}$
and a filtration with eleven components as indicated below.

\begin{figure}[h!]
  \centering
\vskip -6pt
$
\begin{array}{ccc}
    \setlength\fboxsep{0.5pt}
    \setlength\fboxrule{0.5pt}
    \fbox{\includegraphics[scale=.33]{00X4.eps}} &
    \setlength\fboxsep{0.5pt}
    \setlength\fboxrule{0.5pt}
    \fbox{\includegraphics[scale=.33]{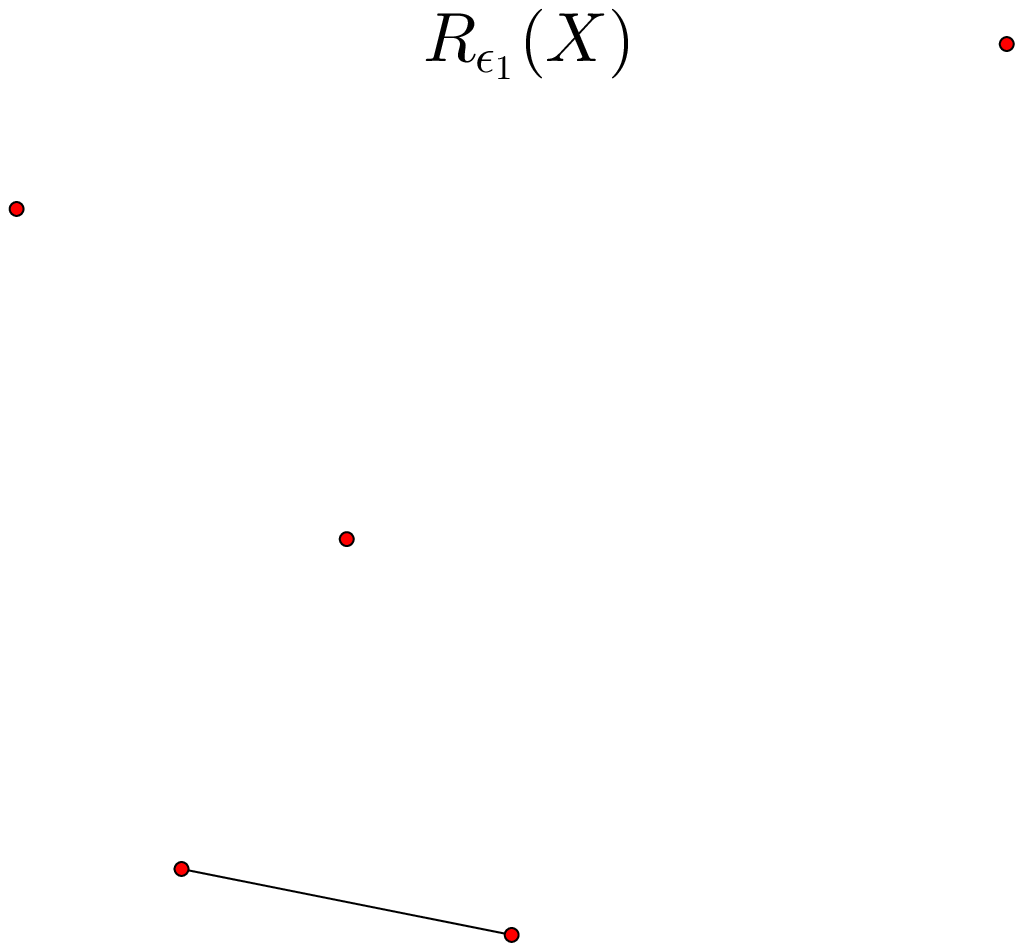}} &
    \setlength\fboxsep{0.5pt}
    \setlength\fboxrule{0.5pt}
    \fbox{\includegraphics[scale=.33]{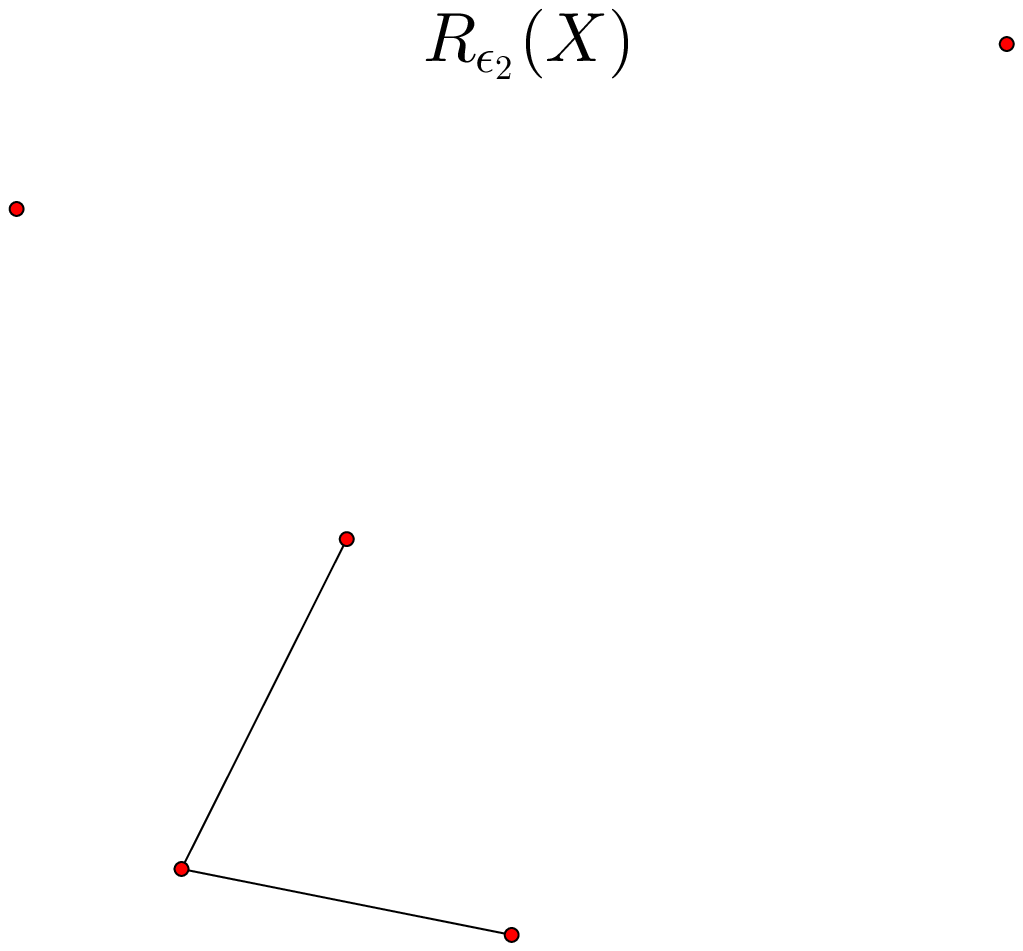}} \\
    \setlength\fboxsep{0.5pt}
    \setlength\fboxrule{0.5pt}
    \fbox{\includegraphics[scale=.33]{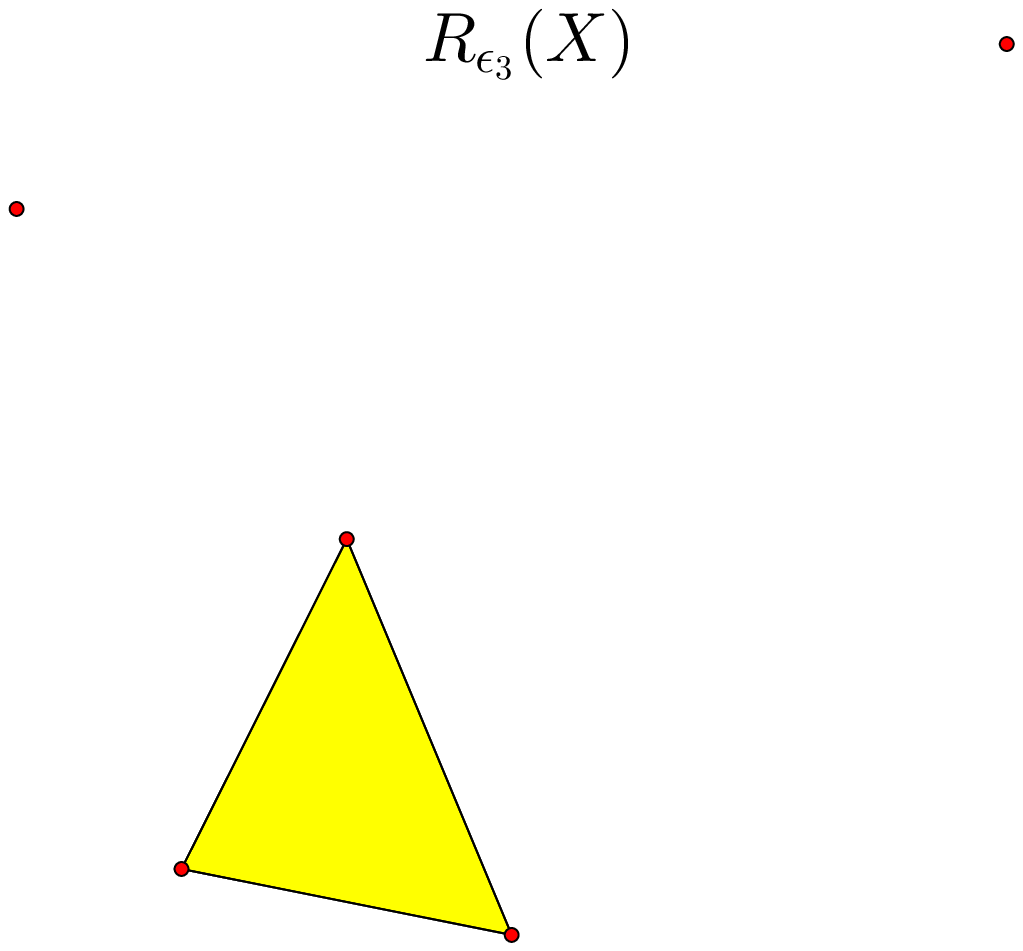}} &
    \setlength\fboxsep{0.5pt}
    \setlength\fboxrule{0.5pt}
    \fbox{\includegraphics[scale=.33]{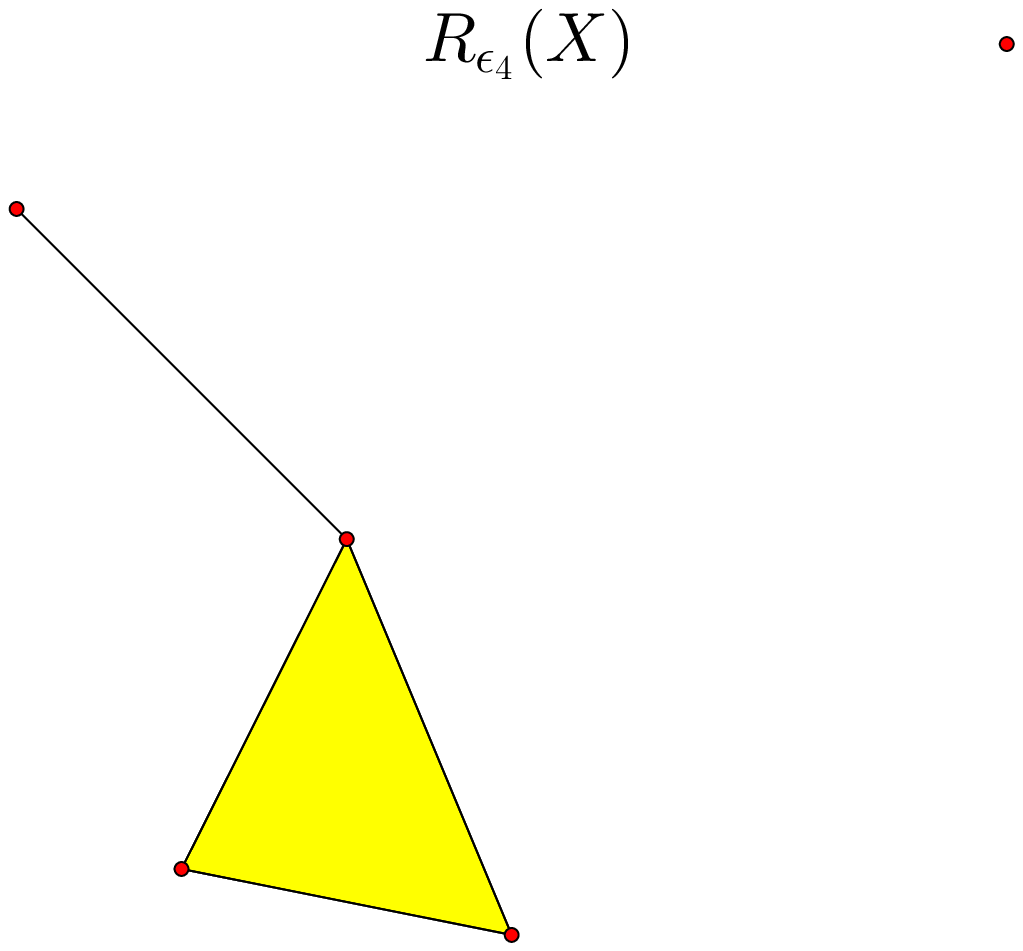}} &
    \setlength\fboxsep{0.5pt}
    \setlength\fboxrule{0.5pt}
    \fbox{\includegraphics[scale=.33]{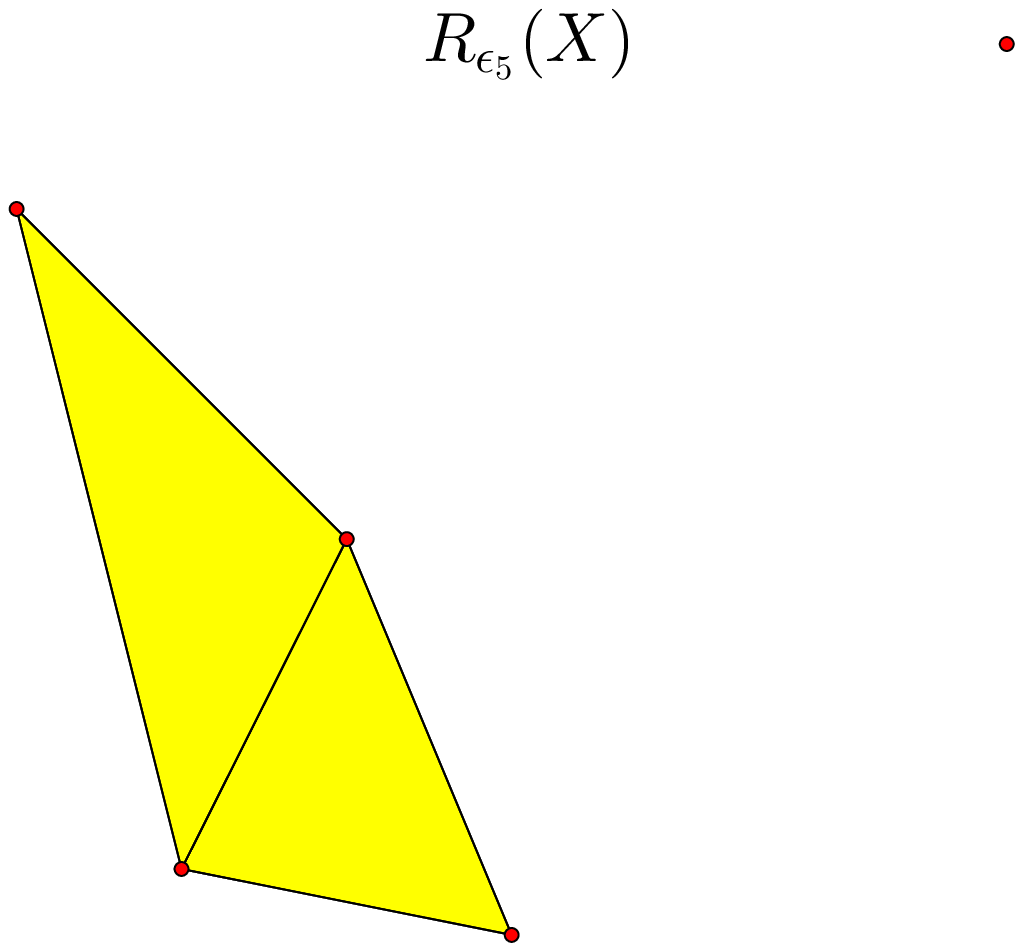}}
    \end{array}
$
$
\begin{array}{ccc}
    \setlength\fboxsep{0.5pt}
    \setlength\fboxrule{0.5pt}
    \fbox{\includegraphics[scale=.33]{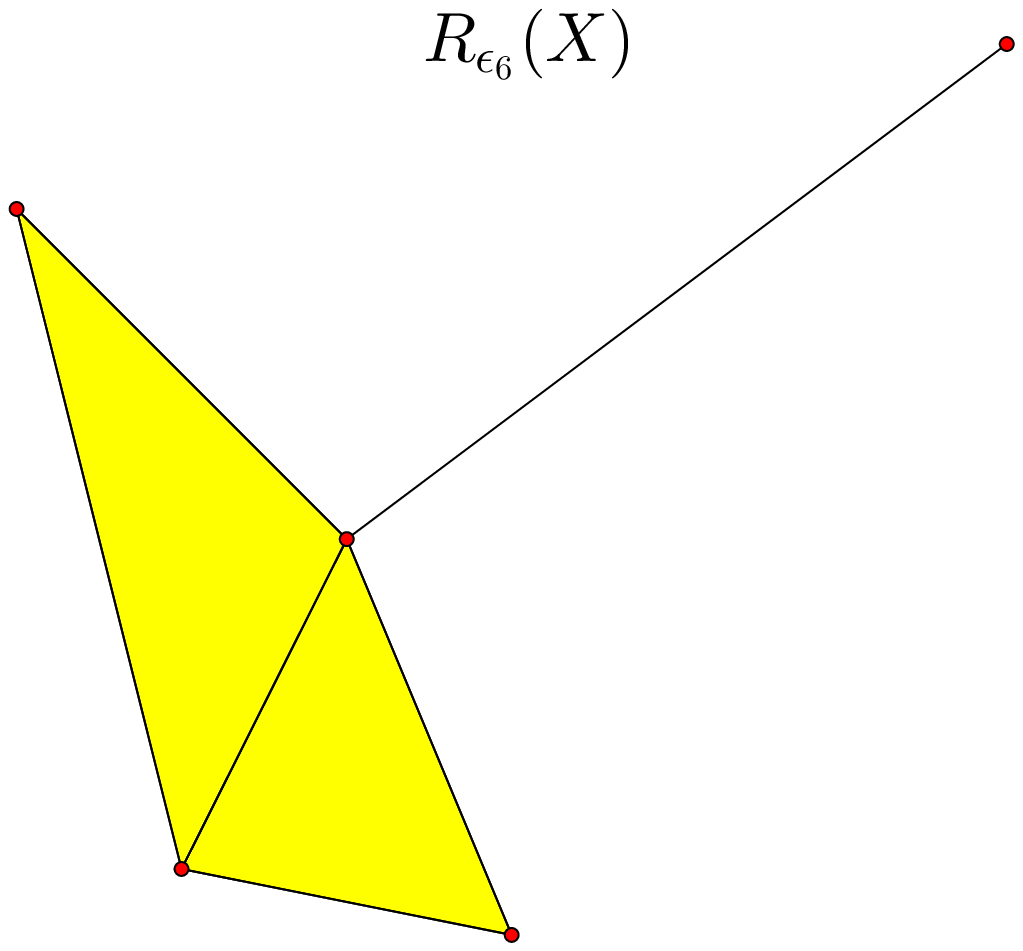}} &
    \setlength\fboxsep{0.5pt}
    \setlength\fboxrule{0.5pt}
    \fbox{\includegraphics[scale=.33]{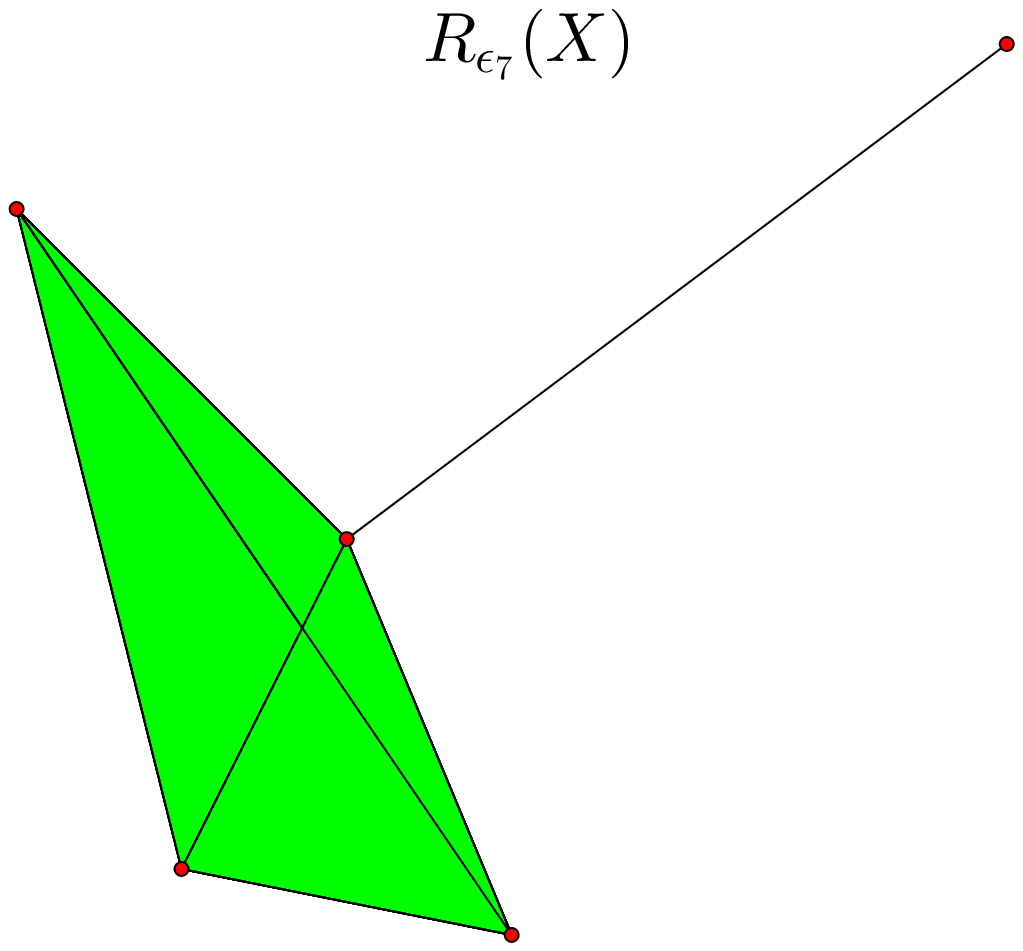}} &
    \setlength\fboxsep{0.5pt}
    \setlength\fboxrule{0.5pt}
    \fbox{\includegraphics[scale=.33]{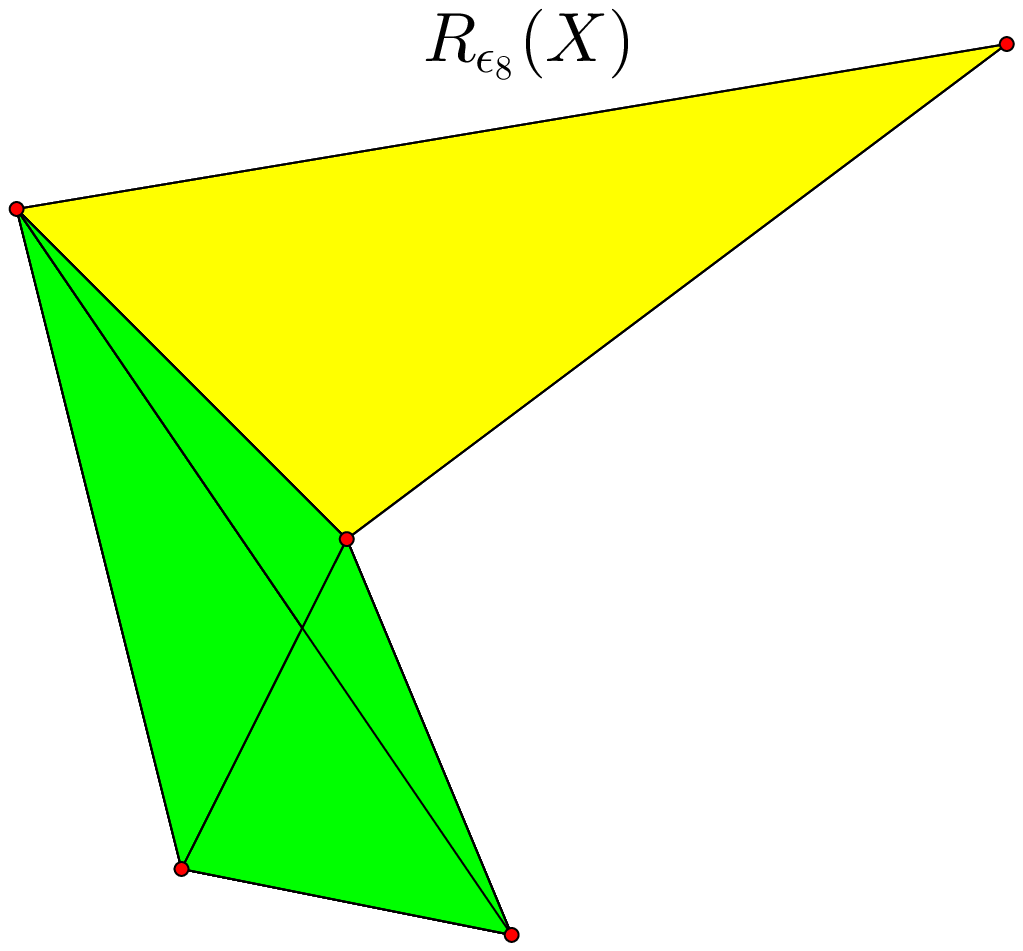}} \\
    \setlength\fboxsep{0.5pt}
    \setlength\fboxrule{0.5pt}
    \fbox{\includegraphics[scale=.33]{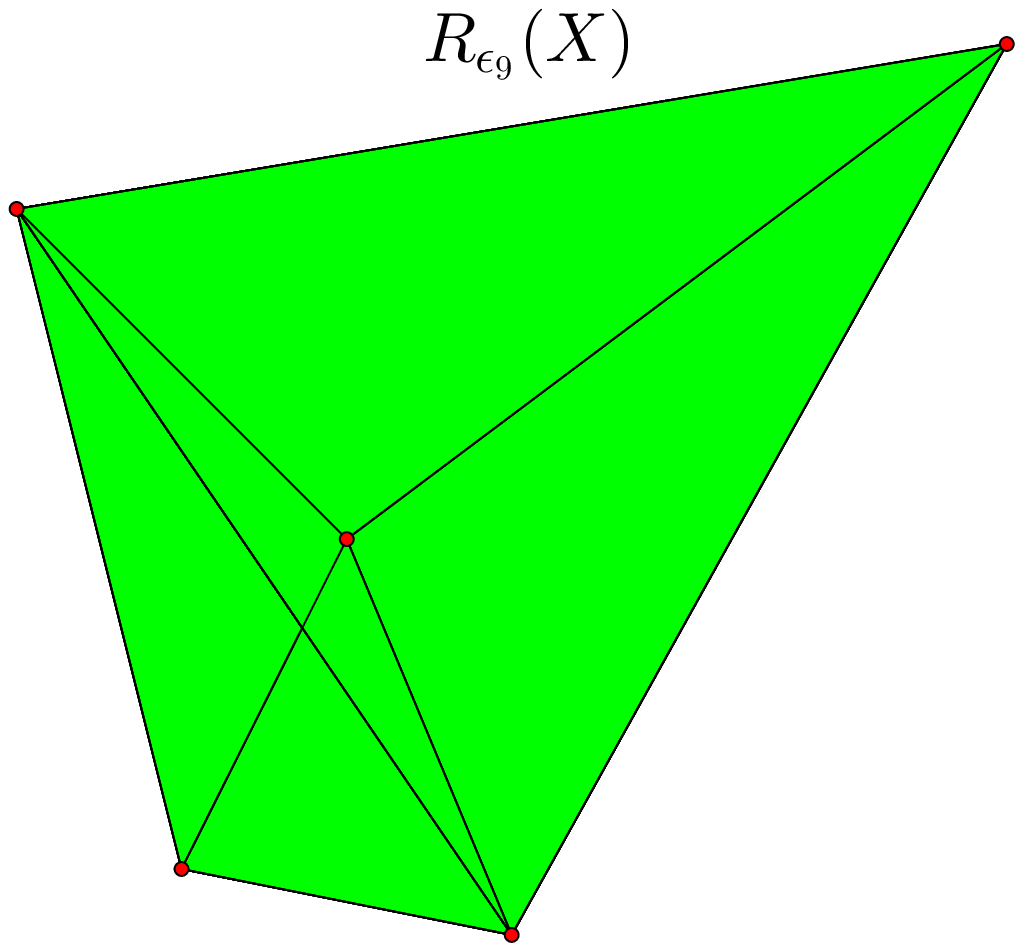}} &
    \setlength\fboxsep{0.5pt}
    \setlength\fboxrule{0.5pt}
    \fbox{\includegraphics[scale=.33]{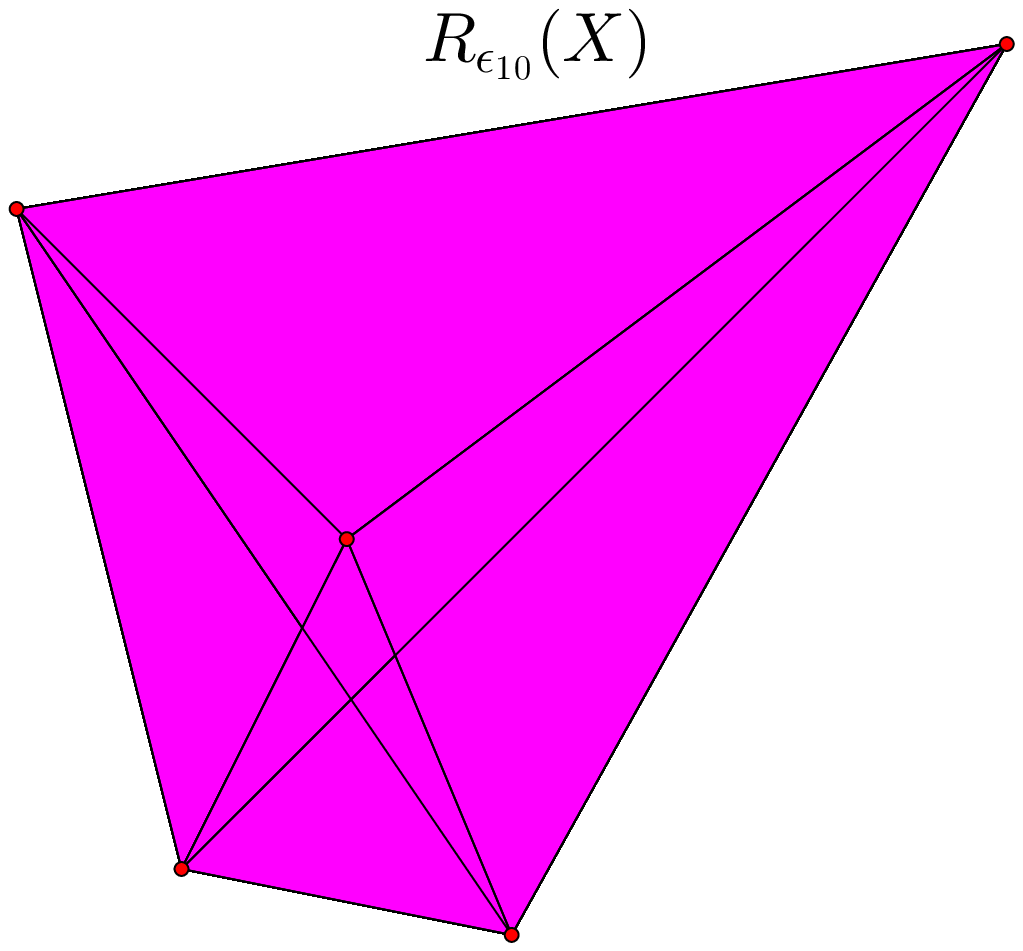}} &
\end{array}
$
\caption{Filtration of Vietoris-Rips complexes of the PCD in {Example \ref{example221}}.}
\end{figure}

\end{exmp}

The invariant we will calculate will be a collection of intervals for
any $r$, $0\leq r\leq \dim R_{\epsilon_N}(X)$.
Each interval will have as \textit{ends} the numbers $\epsilon_0,\cdots,\epsilon_N,+\infty$
or equivalently  $0, 1, \cdots, N, $ or $+\infty$ with the convention that $i$
specifies  the number $\epsilon_i$.

If we are interested in the bar codes for $r\leq m-1$, one can work with $R_{\epsilon}(X,m)$, the $m$-skeleton
of $R_{\epsilon}(X)$. The $m$-restricted filtered simplicial complex
\begin{equation}\label{2.2.2}
 R_{\epsilon_0}(X,m)\subseteq R_{\epsilon_1}(X,m) \subseteq \cdots \subseteq R_{\epsilon_N}(X,m).
\end{equation}
has the same persistent homology and bar codes as the original filtration (\ref{2.2.1}) up to dimension $m-1$
(see subsection  \ref{secPHBFSCC} for the definition of persistent homology and bar codes).
This supposes a much smaller amount of data to be stored in a computer and
permits to calculate all bar codes for $r\leq m-1$.

Given a positive integer number $P (P\leq N)$, we can make a further restriction of the filtration (\ref{2.2.2}) by stopping at level $P$
\begin{equation}\label{2.2.3}
 R_{\epsilon_0}(X,m)\subseteq R_{\epsilon_1}(X,m) \subseteq \cdots \subseteq R_{\epsilon_P}(X,m).
\end{equation}

This will permit the calculation of all bar codes for $r\leq m-1$ with both ends less than $P$
or the bar codes for $r\leq m-1$ with left end less than $P$ and right end $\geq P$.

\subsection{Persistent Linear Algebra (a Gentle Introduction)}\label{secPLA}

\begin{defn} \label{a persistence vector space} \label{pvs}
 1) A persistence vector space ${\mathcal V}$ is a sequence
$$\{V_n,\varphi_{n}:V_n\rightarrow V_{n+1} \mid n\in \mathbb{Z}_{\geq 0}\}$$
with $V_n$ vector spaces over a field $\kappa$ and $\varphi_{n}$
linear maps.

A persistence vector space is \textbf{tame} iff each $V_n$ has finite
dimension and $\varphi_{n}$ is an isomorphisms for $n$
large enough.

The main features of this concept are ``birth and death time'' \footnote{Time here means the index of the component of the filtration.} for elements
and information about these is provided by the bar codes.

2) A linear map of persistence vector spaces $\omega:{\mathcal V}\rightarrow {\mathcal W}$,
where ${\mathcal V}= \{V_n,\varphi_{n}:V_n\rightarrow V_{n+1} \mid n\in \mathbb{Z}_{\geq 0}\}$
and ${\mathcal W}= \{W_n,\phi_{n}:W_n\rightarrow W_{n+1} \mid n\in \mathbb{Z}_{\geq 0}\}$
is a commutative diagram
$$
\xymatrix{
V_0 \ar[d]^{\omega_0} \ar[r]^{\varphi_0} &
V_1 \ar[d]^{\omega_1} \ar[r]^{\varphi_1} &
\cdots  \ar[r]^{\varphi_{n-1}} &
V_n \ar[d]^{\omega_n} \ar[r]^{\varphi_n} &
V_{n+1} \ar[d]^{\omega_{n+1}} \ar[r]^{\varphi_{n+1}} &
\cdots
 \\
W_0 \ar[r]^{\phi_0}  &
W_1 \ar[r]^{\phi_1}  &
\cdots \ar[r]^{\phi_{n-1}} &
W_n \ar[r]^{\phi_n}  &
W_{n+1} \ar[r]^{\phi_{n+1}}  &
\cdots
}
$$
where $\omega_n$ is a linear map from $V_n$ to $W_n$ for each $n\geq0$.

 A linear map of persistence vector spaces $\omega:{\mathcal V}\rightarrow {\mathcal W}$
is an isomorphism if there exists another linear map of persistence vector spaces
${\omega}':{\mathcal W}\rightarrow {\mathcal V}$ such that $\omega'\circ\omega:{\mathcal V}\rightarrow {\mathcal V}$
and  $\omega\circ\omega':{\mathcal W}\rightarrow {\mathcal W}$ are identities.
${\mathcal V}$ and ${\mathcal W}$ are isomorphic if there is an isomorphism between them.

The existence of a linear map $\omega:{\mathcal V}\rightarrow {\mathcal W}$
so that each component $\omega_n:{\mathcal V_n}\rightarrow {\mathcal W_n}$ is an isomorphism
for every $n\geq0$ implies $\omega$ is an isomorphism. One takes $\omega_n'=(w_n)^{-1}$.

3) Let ${\mathcal U}= \{U_n,\psi_{n}:U_n\rightarrow U_{n+1} \mid n\in \mathbb{Z}_{\geq 0}\}$,
${\mathcal V}= \{V_n,\varphi_{n}:V_n\rightarrow V_{n+1} \mid n\in \mathbb{Z}_{\geq 0}\}$
and ${\mathcal W}= \{W_n,\phi_{n}:W_n\rightarrow W_{n+1} \mid n\in \mathbb{Z}_{\geq 0}\}$
be persistence vector spaces. A short exact sequence of persistence vector spaces
$$
\xymatrix{
0 \ar[r] & {\mathcal{U}} \ar[r]^{\mu} &{\mathcal V} \ar[r]^{\nu} & {\mathcal W} \ar[r] & 0
}
$$
is a sequence of linear maps of persistence vector spaces such that
$$
\xymatrix{
0 \ar[r] & {\ U}_n \ar[r]^{\mu_n} &{\ V}_n \ar[r]^{\nu_n} & {\ W}_n \ar[r] & 0
}
$$
is short exact for all $n\geq0$.

The short exact sequence
$
\xymatrix{
0 \ar[r] & {\mathcal{U}} \ar[r]^{\mu} &{\mathcal V} \ar[r]^{\nu} & {\mathcal W} \ar[r] & 0
}
$
splits if there exits a linear map $\alpha: {\mathcal V}\rightarrow {\mathcal U}$
such that $\alpha\circ\mu=identity$
or if there exists a linear map $\beta: {\mathcal W}\rightarrow {\mathcal V}$
such that $\nu\circ\beta=identity$.

\begin{note}
The alternative definitions are equivalent.
Given $\alpha:{\mathcal V}\rightarrow {\mathcal U}$ such that $\alpha\circ\mu=identity$, we can define
$\beta:{\mathcal W}\rightarrow {\mathcal V}$ such that $\nu\circ\beta=identity$ and vice versa.

i)Let $\alpha:{\mathcal V}\rightarrow {\mathcal U}$ such that $\alpha\circ\mu=identity$.
For any $x\in W_n$, there exists $x'\in V_n$ such that $\nu_n(x')=x$.
Define
$$
\begin{array}{lll}
\beta_n: &  W_n\rightarrow V_n\\
 & x  \mapsto x'-\mu_n\circ\alpha_n(x').
\end{array}
$$
To see that $\beta_n$ is well defined, i.e. independent of the choice of $x'$,
consider
$$
\begin{array}{lll}
\gamma_n: & V_n\rightarrow  U_n\oplus W_n\\
 & x  \mapsto (\alpha_n x, \nu_n x).
\end{array}
$$

$\gamma_n$ is injective. Indeed if $x\in \ker(\gamma_n)$ hence $\alpha_n x=0$ and $\nu_n x=0$.
Since ${\textmd im}{\mu_n}=\ker{\nu_n}$, there exists
$y\in U_n$ such that $x=\mu_n(y)$. Then $y=\alpha_n\circ\mu_n(y)=\alpha_n(x)=0$, so $x=\mu_n(0)=0$.

If $x''\in V_n$ such that $\nu_n(x'')=x$, then
$\gamma_n$ will map both $x'-\mu_n\circ\alpha_n(x')$
and $x''-\mu_n\circ\alpha_n(x'')$ to $(0,x)$.
Since $\gamma_n$ is injective, $x'-\mu_n\circ\alpha_n(x')=x''-\mu_n\circ\alpha_n(x'')$
which shows that $\beta_n$ is well defined.

It is straightforward that $\nu_n\circ\beta_n=identity$ and
the following diagram is commutative
$$
\xymatrix{
W_n \ar[r]^{\phi_n} \ar[d]^{\beta_n} &
W_{n+1} \ar[d]^{\beta_{n+1}} \\
 V_n \ar[r]^{\varphi_n} &
 V_{n+1}
}
$$

ii)Let $\beta:{\mathcal W}\rightarrow {\mathcal V}$ with  $\nu\circ\beta=identity$.
For any $x\in {V}_n$, let $y_x=x-\beta_n\circ\nu_n(x)$,
then $\nu_n(y_x)=0$ and there exist a unique $z_x\in \mu_n$ such that
$\mu_n(z_x)=y_x$.

Define
$$
\begin{array}{lll}
\alpha_n: &  V_n\rightarrow U_n\\
 & x  \mapsto z_x
\end{array}
$$

$\alpha_n$ is well defined since for each $x$, $y_x$ and $z_x$
is unique.

It is straightforward to check $\alpha_n\circ\mu_n=identity$ and
the following diagram is commutative
$$
\xymatrix{
V_n \ar[r]^{\varphi_n} \ar[d]^{\alpha_n} &
V_{n+1} \ar[d]^{\alpha_{n+1}} \\
U_n \ar[r]^{\psi_n} &
U_{n+1}.
}
$$
\end{note}

4) Direct sum of a finite collection of persistence vector spaces
 ${\mathcal V}^i= \{V^i_n,\varphi^i_{n}:V^i_n\rightarrow V^i_{n+1} \mid n\in \mathbb{Z}_{\geq 0}\}$, $i\in \Lambda$($\Lambda$ finite),
 is defined by $\displaystyle\bigoplus_{i\in \Lambda} {\mathcal V}^i= \{V_n,\varphi_{n}:V_n\rightarrow V_{n+1} \mid n\in \mathbb{Z}_{\geq 0}\}$,
where $\displaystyle V_n=\bigoplus_{i\in \Lambda} V^i_n$, $\displaystyle\varphi_n=\bigoplus_{i\in \Lambda} \varphi^i_n$ for $n\geq 0$.
Note that a direct sum of a finite collection of tame persistence vector spaces is tame.

\end{defn}

\begin{obs} \label{split short exact sequence}

1)
Given a split short exact sequence
$$
\xymatrix{
0 \ar[r] & {\mathcal{U}} \ar[r]^{\mu} &{\mathcal V} \ar[r]^{\nu} & {\mathcal W} \ar[r] & 0
},
$$
we have
${\mathcal V} \cong {\mathcal U}\oplus {\mathcal W} $,
 where ${\mathcal U}$,${\mathcal V}$,${\mathcal W}$ and $\mu$,$\nu$are the same as
 in { Definition \ref{a persistence vector space}} 3).

 2)
 Not any short exact sequence is split.

\end{obs}

\begin{proof}
1) Since the two conditions for a short exact sequence to be split are
equivalent, WLOG, we can assume that there exists a linear map $\beta: {\mathcal W}\rightarrow {\mathcal V}$
such that $\nu\circ\beta=identity$.
Define the linear map
$$
\begin{array}{lll}
\delta_n: &  U_n\oplus W_n\rightarrow V_n\\
 & (y,z)  \mapsto \mu_n y+ \beta_n z
\end{array}
$$
for all $n\geq0$.

It is easy to check the following diagram is commutative
$$
\xymatrix{
U_n\oplus W_n \ar[r]^-{(\psi_n,\phi_n)} \ar[d]^{\delta_n} &
U_{n+1}\oplus W_{n+1} \ar[d]^{\delta_{n+1}} \\
V_n \ar[r]^{\varphi_n} &
V_{n+1}
}
$$

Given any element $x\in V_n$.
Let $z=\nu_n(x)$, then $\nu_n(x-\beta_n z)=z-z=0$ and
there exist a unique $y\in U_n$ such that $\mu_n(y)=x-\beta_n z$.

Then $\delta_n(y,z)=\mu_ny+\beta_nz=x-\beta_nz+\beta_nz=x$ and $\delta_n$
is surjective.

Let $(y,z)\in\ker\delta_n$, then $\mu_ny+\beta_nz=0$.

Since $0=\nu_n(\mu_ny+\beta_nz)=\nu_n\circ\mu_n(y)+\nu_n\circ\beta_n(z)=0+z=z$, we have $z=0$.

Then $\mu_ny=0$ implies $y=0$, since $\mu_n$ is injective.

Therefore $(y,z)=(0,0)$ and $\delta_n$ is injective.

Hence $\delta_n$ is an isomorphism for all $n\geq0$.

Therefore $\delta:  {\mathcal U}\oplus {\mathcal W}\rightarrow {\mathcal V} $ is an isomorphism
and ${\mathcal V} \cong {\mathcal U}\oplus {\mathcal W} $.

2)
Counterexample:

Consider the following short exact sequence:
$$
\xymatrix{
{\mathcal U} \ar[d]_{\mu} & {\mathbb R} \ar[r]^0 \ar[d]^{i_1} & {\mathbb R} \ar[r]^0 \ar[d]^{i_1} & 0 \ar[r]^0 \ar[d]^0 & 0 \ar[r]^0 \ar[d]^0 & 0 \ar[r]^0 \ar[d]^0 & \cdots & & & \\
{\mathcal V} \ar[d]_{\nu} \ar@/_1pc/@{-->}[u]_{\alpha}& {\mathbb R}\oplus{\mathbb R} \ar[r]^c \ar[d]^{p_2} & {\mathbb R}\oplus{\mathbb R} \ar[r]^0 \ar[d]^{p_2} & 0\ar[r]^0 \ar[d]^0 & 0\ar[r]^0 \ar[d]^0& 0\ar[r]^0 \ar[d]^0& \cdots & & &\\
{\mathcal W} & {\mathbb R} \ar[r]^0 & {\mathbb R} \ar[r]^0 & 0 \ar[r]^0 & 0 \ar[r]^0& 0 \ar[r]^0& \cdots & & &
}
$$
where
$$
\begin{array}{rcrcrc}
i_1: &{\mathbb R}\rightarrow {\mathbb R}\oplus{\mathbb R}, & p_2:& {\mathbb R}\oplus{\mathbb R}\rightarrow{\mathbb R}, & c: &  {\mathbb R}\oplus{\mathbb R}\rightarrow {\mathbb R}\oplus{\mathbb R} \\
    & x  \mapsto (x,0) & & (x,y)\mapsto y & & (x,y)\mapsto (y,0)
\end{array}
$$

If the sequence splits, in view of $\alpha\circ\mu=identity$, $\alpha_1(x,0)=x$,
which makes the commutativity of the diagram
$$
\xymatrix{
{\mathbb R} \ar[r]^0 & {\mathbb R} \\
{\mathbb R}\oplus{\mathbb R} \ar[u]^{\alpha_0} \ar[r]^c & {\mathbb R}\oplus{\mathbb R} \ar[u]^{\alpha_1}
}
$$
impossible.

\end{proof}

\begin{nota}
Define the following tame persistence vector spaces
as \textbf {basic tame persistence vector spaces}.

1)
$\kappa[t]$ is the tame persistence vector space over a field $\kappa$
$$\{V_n,\varphi_{n}:V_n\rightarrow V_{n+1} \mid n\in \mathbb{Z}_{\geq 0}\}$$
where $V_n=\kappa$ and $\varphi_{n}=identity$ for all $n\geq0$.

It corresponds to the interval $[0,\infty)$.

2)
$S^r \kappa[t]$ is the tame persistence vector space over a field $\kappa$
$$\{V_n,\varphi_{n}:V_n\rightarrow V_{n+1} \mid n\in \mathbb{Z}_{\geq 0}\}$$
where $V_n=0$ for $0\leq n< r$, $V_n=\kappa$ for $n\geq r$
, $\varphi_n=0$ for $0\leq n< r $ and $\varphi_n=identity$ for $n\geq r$.

It corresponds to the interval $[r,\infty)$.

The notation $S^r$ indicates the right shift with $r$-units.

3)
$T_{r+1}\kappa[t]$ is the tame persistence vector space over a field $\kappa$
$$\{V_n,\varphi_{n}:V_n\rightarrow V_{n+1} \mid n\in \mathbb{Z}_{\geq 0}\}$$
where $V_n=\kappa$ for $0\leq n\leq r$, $V_n=0$ for $n> r$
, $\varphi_n=identity$ for $0\leq n< r $ and $\varphi_n=0$ for $n\geq r$.

It corresponds to the interval $[0,r]$.

The notation $T_{r+1}$ indicates the truncation at level $r+1$.

4)
$S^rT_{p+1}\kappa[t]$ is the tame persistence vector space over a field $\kappa$
$$\{V_n,\varphi_{n}:V_n\rightarrow V_{n+1} \mid n\in \mathbb{Z}_{\geq 0}\}$$
where $V_n=\kappa$ for $r\leq n\leq r+p$ and $V_n=0$ otherwise;
$\varphi_n=identity$ for $r\leq n\leq (r+p-1)$ and $\varphi_n=0$
otherwise.

It corresponds to the interval $[r,r+p]$.

\end{nota}

\begin{nota}

For tame persistence vector space ${\mathcal V}=\{V_n,\varphi_{n}:V_n\rightarrow V_{n+1} \mid n\in \mathbb{Z}_{\geq 0}\}$

1)
Denote $\varphi_{i,j}=\varphi_{j-1}\circ\cdots\circ\varphi_i:V_i\rightarrow V_j$ for $i<j$ and $\varphi_{i,i}=identity:V_i\rightarrow V_i$,
with $i,j\in {\mathbb Z}_{\geq 0}$.

2)
Denote $\beta(i,j)=\dim({\textmd im}(\varphi_{i,j}:V_i\rightarrow V_j))$ with $i\leq j\in{\mathbb Z}_{\geq 0}$.

Note $\beta(i,i)=\dim V_i$ and  $\beta(i,j)=\beta(i,j+1)$ for $j$ large enough.
If $\varphi_n$ is an isomorphism for $n\geq N$, denote $\beta(i,\infty)=\beta(i,m)$,
where $m$ is any integer larger than $i$ and $N$.

\end{nota}

\begin{defn} \label{order of basic space}

Define an order $\prec$ of all basic tame persistence vector spaces:

1) $S^r\kappa[t]\prec S^{r'}\kappa[t]$ if $r<r'$;

2) $S^r T_{p+1}\kappa[t]\prec S^{r'} T_{p'+1}\kappa[t]$ if $r<r'$ or ($r=r'$ and $p<p'$);

3) $S^r T_{p+1}\kappa[t]\prec S^{r'}\kappa[t]$.

Clearly, this order is a strict total order.

\end{defn}

\begin{lem} \label{if V and W}
If ${\mathcal V}$ and ${\mathcal W}$ are two basic tame persistence vector spaces
such that  ${\mathcal V}\prec{\mathcal W}$, then any map from ${\mathcal V}$ to
${\mathcal W}$ is trivial.
\end{lem}
\begin{proof}

We check situation 2) of {Definition \ref{order of basic space}} first.

Suppose ${\mathcal V}=S^r T_{p+1} \kappa[t]$ and ${\mathcal W}=S^{r'} T_{p'+1} \kappa[t]$ and
there is a linear map $\omega:{\mathcal V}\rightarrow {\mathcal W}$. We want to show $\omega=0$.

We have 2 cases:

Case 1: $r<r'$

In this case, $\omega_n=0$ for all $n<r$ or $n>r+p$, since $V_n=0$.

For $r\leq n\leq r+p$, consider the following commutative diagram:
$$
\xymatrix{
\kappa \ar[r]^{id} \ar[d]^{\omega_r} & \kappa \ar[d]^{\omega_n} \\
0 \ar[r]^0 & W_n
}
$$

We have $\omega_n=\omega_n\circ id=0\circ \omega_r=0$.

Hence, $\omega=0$, if $r<r'$.

Case 2: $r=r'$ and $p<p'$

Consider the commutative diagram:
$$
\xymatrix{
\kappa \ar[r]^{0} \ar[d]^{\omega_n} & 0 \ar[d]^{\omega_{r'+p'}} \\
\kappa \ar[r]^{id} & \kappa
}
$$
for $r\leq n\leq r+p$.

 The proofs of situations 1) and 3) of {Definition \ref{order of basic space}} are similar.

\end{proof}

\begin{prop} \label{any tame persistence vector space}
Any tame persistence vector space ${\mathcal V}=\{V_n,\varphi_{n}:V_n\rightarrow V_{n+1}|n\in \mathbb{Z}_{\geq 0}\}$
over a field $\kappa$ is isomorphic to
$$
\bigoplus_{1\leq i\leq p}  S^{r_i}  \kappa[t] \oplus \bigoplus_{1\leq j\leq q} S^{m_j} T_{n_j+1} \kappa[t]
$$

where $p,q,r_i,m_j \;and \; n_j \in {\mathbb Z}_{\geq 0}$.

{So, any tame persistence vector space can be decomposed in to a direct sum of
a finite collection of basic tame persistence vector spaces.}

\end{prop}

\begin{proof}

Since ${\mathcal V}$ is a tame persistence vector space $\dim(V_n)$ is finite
for all $n\geq0$ and there exists $N\geq0$ such that $\varphi_{n}$ is an isomorphisms for $n\geq N$.

Define $\displaystyle L({\mathcal V})=\sum_{0\leq n\leq N} \dim V_n$.

We prove by induction on $L({\mathcal V})$.

If $L({\mathcal V})=0$, then ${\mathcal V}=0$. We are done.

If $L({\mathcal V})\neq 0$, then $V_k\neq 0$ for some $0\leq k\leq N$, with $V_i=0$ if $i<k$.

WLOG, suppose $V_0\neq 0$ and choose a nonzero element $v_0$ from $V_0$.

Define $v_n=\varphi_{0,n}(v_0), \;\; n\geq 0$.

\vskip 0.5cm

\textbf{Case 1:} $v_n\neq 0$ for $0\leq n<j$ and $v_n=0$ for $n\geq j$,
where $j\geq 1$.

\begin{note} We must have $j\leq N$, otherwise $v_n\neq 0$ for all $n\geq 0$,
a contradiction. \end{note}

Consider the following ``short exact sequence'' of tame persistence vector spaces
$$
\xymatrix@C=.65cm{
 {\mathcal W}_0 \hskip -1cm &
\kappa \ar[d]_{f_0} \ar[r]^{\alpha_0} &
\kappa \ar[d]_{f_1} \ar[r]^{\alpha_1} &
\cdots \ar[r]^{\alpha_{j-1}} &
\kappa \ar[d]_{f_{j-2}} \ar[r]^{\alpha_{j-2}} &
\kappa \ar[d]_{f_{j-1}} \ar[r]^{\alpha_{j-1}} &
0 \ar[d]_{f_j} \ar[r]^{\alpha_j} &
0 \ar[d]_{f_{j+1}} \ar[r]^{\alpha_{j+1}} &
\cdots
\\
\mathcal{V} \hskip -1cm &
V_0 \ar[d]_{g_0} \ar[r]^{\varphi_0} \ar@/_.5pc/@{-->}[u]_{h_0} &
V_1 \ar[d]_{g_1} \ar[r]^{\varphi_1} \ar@/_.5pc/@{-->}[u]_{h_1} &
\cdots \ar[r]^{\varphi_{j-1}} &
V_{j-2} \ar[d]_{g_{j-2}} \ar[r]^{\varphi_{j-2}} \ar@/_.5pc/@{-->}[u]_{h_{j-2}}&
V_{j-1} \ar[d]_{g_{j-1}} \ar[r]^{\varphi_{j-1}} \ar@/_.5pc/@{-->}[u]_{h_{j-1}}&
V_{j} \ar[d]_{g_j} \ar[r]^{\varphi_j} \ar@/_.5pc/@{-->}[u]_{h_{j}}&
V_{j+1} \ar[d]_{g_{j+1}} \ar[r]^{\varphi_{j+1}} \ar@/_.5pc/@{-->}[u]_{h_{j+1}}&
\cdots
\\
{\mathcal V}'  \hskip -1cm &
V_0/{\kappa v_0} \ar[r]^-{\widetilde{\varphi}_0} &
V_1/{\kappa v_1} \ar[r]^-{\widetilde{\varphi}_1}  &
 \cdots  \ar[r]^-{\widetilde{\varphi}_{j-1}}  &
 V_{j-2}/{\kappa v_{j-2}} \ar[r]^-{\widetilde{\varphi}_{j-2}}  &
 V_{j-1}/{\kappa v_{j-1}} \ar[r]^-{\widetilde{\varphi}_{j-1}}  &
 V_j \ar[r]^-{\widetilde{\varphi}_j}  &
   V_{j+1} \ar[r]^-{\widetilde{\varphi}_{j+1}}  &
 \cdots
}
$$
where

$\alpha_n=identity$ for $0\leq n\leq j-2$ and
$\alpha_n=0$ for $n\geq j-1$;

$
\begin{array}{lll}
f_n: & \kappa\rightarrow V_n & for \;\;0\leq n\leq j-1 \\
 & \lambda\mapsto \lambda v_n &
 \end{array}
$

and $f_n=0$ for $n\geq j$;

$
\begin{array}{lll}
g_n: & V_n\rightarrow V_n/{\kappa v_n} & for \;\;0\leq n\leq j-1 \\
 & x \mapsto x+\kappa v_n &
 \end{array}
$

and $g_n=identity$ for $n\geq j$;

$
\begin{array}{lll}
\widetilde{\varphi}_n: & V_n/\kappa v_n\rightarrow V_{n+1}/{\kappa v_{n+1}} &
 for \;\;0\leq n\leq j-1 \\
 & x+\kappa v_n \mapsto \varphi_n(x)+\kappa v_{n+1} &
 \end{array}
$

and $\widetilde{\varphi}_n=\varphi_n$ for $n\geq j$.

The above diagram is commutative and
is a short exact sequence of tame persistence vector spaces.
Next we will show the above short exact sequence splits.

Extend $v_{j-1}$ to a basis $\{x_1=v_{j-1},x_2,\cdots,x_p \}$ of $V_{j-1}$.

Given any element $x\in V_{j-1}$, we have a unique expression
$$
x=\sum_{i=1}^pa_ix_i
$$

Define a linear map
$$
\begin{array}{ll}
h_{j-1}:&V_{j-1}\rightarrow \kappa \\
& \sum_{i=1}^pa_ix_i \mapsto a_1
\end{array}
$$

For $0\leq n\leq j-1$, define a linear map
$$
\begin{array}{ll}
h_n:&V_n\rightarrow \kappa \\
& x\mapsto h_{j-1}\circ \varphi_{n,j-1}(x)
\end{array}
$$

For $n\geq j$, define $h_n=0$.

Clearly $h:\mathcal{V}\rightarrow \mathcal{W}_0$ is well
defined and $h\circ f=identity$.

Therefore the above short exact sequence splits, and
$$
\begin{array}{ll}
\mathcal{V}&\cong \mathcal{W}_0\oplus\mathcal{V}'\\
 &= T_j \kappa[t]\oplus\mathcal{V}'
\end{array}
$$

Hence $L(\mathcal{V}')=L(\mathcal{V})-j<L(\mathcal{V})$.

\vskip 0.5cm

\textbf{Case 2:} $v_n\neq0$ for all $n\geq0$.

Consider the following ``short exact sequence'' of tame persistence vector spaces
$$
\xymatrix@C=0.65cm{
{\mathcal W}_1 \hskip -1cm&
\kappa \ar[d]_{f_0} \ar[r]^{\alpha_0} &
\kappa \ar[d]_{f_1} \ar[r]^{\alpha_1} &
\cdots \ar[r]^{\alpha_{N-2}} &
\kappa \ar[d]_{f_{N-1}} \ar[r]^{\alpha_{N-1}} &
\kappa \ar[d]_{f_N} \ar[r]^{\alpha_N} &
\kappa \ar[d]_{f_{N+1}} \ar[r]^{\alpha_{N+1}} &
\cdots
\\
\mathcal{V}   \hskip -1cm&
V_0 \ar[d]_{g_0} \ar[r]^{\varphi_0} \ar@/_.5pc/@{-->}[u]_{h_0} &
V_1 \ar[d]_{g_1} \ar[r]^{\varphi_1} \ar@/_.5pc/@{-->}[u]_{h_1} &
\cdots \ar[r]^{\varphi_{N-2}} &
V_{N-1} \ar[d]_{g_{N-1}} \ar[r]^{\varphi_{N-1}} \ar@/_.5pc/@{-->}[u]_{h_{N-1}} &
V_{N} \ar[d]_{g_N} \ar[r]^{\varphi_N} \ar@/_.5pc/@{-->}[u]_{h_{N}} &
V_{N+1} \ar[d]_{g_{N+1}} \ar[r]^{\varphi_{N+1}} \ar@/_.5pc/@{-->}[u]_{h_{N+1}} &
\cdots
\\
{\mathcal V}'   \hskip -1cm&
V_0/{\kappa v_0} \ar[r]^-{\widetilde{\varphi}_0} &
V_1/{\kappa v_1} \ar[r]^-{\widetilde{\varphi}_1}  &
 \cdots  \ar[r]^-{\widetilde{\varphi}_{N-2}}  &
 V_{N-1}/{\kappa v_{N-1}} \ar[r]^-{\widetilde{\varphi}_{N-1}}  &
 V_N/{\kappa v_N} \ar[r]^-{\widetilde{\varphi}_N}  &
   V_{N+1}/ {\kappa v_{N+1}} \ar[r]^-{\widetilde{\varphi}_{N+1}}  &
 \cdots
}
$$
$
\begin{array}{cccccc}
where & \alpha_n=identity\;, & f_n: & \kappa\rightarrow V_n\;,& g_n: & V_n\rightarrow V_n/{\kappa v_n}\\
 &&& \lambda\mapsto \lambda v_n && x \mapsto x+\kappa v_n\\
 \widetilde{\varphi}_n: & V_n/\kappa v_n\rightarrow V_{n+1}/{\kappa v_{n+1}}\\
& x+\kappa v_n \mapsto \varphi_n(x)+\kappa v_{n+1}
 \end{array}
$

 for all $n\geq0$.

We will show the above short exact sequence splits.

Extend $v_{N}$ to a basis $\{y_1=v_{N},y_2,\cdots,y_q \}$ of $V_{N}$.

Given any element $y\in V_{N}$, we have a unique expression
$$
y=\sum_{j=1}^q b_j y_j
$$

Define a linear map
$$
\begin{array}{ll}
h_N:&V_{N}\rightarrow \kappa \\
&  \sum_{j=1}^q b_j y_j \mapsto b_1
\end{array}
$$

For $0\leq n\leq N$, define a linear map
$$
\begin{array}{ll}
h_n:&V_n \rightarrow \kappa \\
& x\mapsto h_{N}\circ \varphi_{n,N}(x)
\end{array}
$$

For $n>N$, define a linear map
$$
\begin{array}{ll}
h_n:&V_n \rightarrow \kappa \\
& x\mapsto h_{N}\circ \varphi_{N,n}^{-1}(x)
\end{array}
$$

Clearly $h_n\circ f_n=identity$ and
$$
\xymatrix{
\kappa \ar[r]^{\alpha_n} & \kappa \\
V_n \ar[u]_{h_n} \ar[r]^{\varphi_n} & V_{n+1}\ar[u]_{h_{n+1}}
}
$$
is commutative for all $n\geq0$.

Therefore the  short exact sequence in {Case 2} splits, and
$$
\begin{array}{ll}
\mathcal{V}&\cong \mathcal{W}_1\oplus\mathcal{V}'\\
 &=\kappa[t]\oplus\mathcal{V}'
\end{array}
$$

Hence $L(\mathcal{V}')=L(\mathcal{V})-(N+1)<L(\mathcal{V})$.

\vskip 0.5cm

In both cases $L({\mathcal V}')<L({\mathcal V})$.

By induction on $L({\mathcal V})$,
$$
\bigoplus_{1\leq i\leq s}  S^{r_i}  \kappa[t] \oplus \bigoplus_{1\leq j\leq t} S^{m_j} T_{n_j+1}\kappa[t]
$$

\end{proof}

\begin{prop} \label{let V and V'}

Let $$\displaystyle {\mathcal V}=\bigoplus_{1\leq i\leq p}  S^{r_i}  \kappa[t] \oplus \bigoplus_{1\leq j\leq q} S^{m_j} T_{n_j+1} \kappa[t] $$

and
$$\displaystyle {\mathcal V} '=\bigoplus_{1\leq i\leq p'}  S^{r'_i}  \kappa[t] \oplus \bigoplus_{1\leq j\leq q'} S^{m'_j} T_{n'_j+1} \kappa[t]$$.

If ${\mathcal V}\cong{\mathcal V} '$, then $p=p'$, $q=q'$ and  $r_i=r_i'$, $m_j=m_j'$, $n_j=n_j'$ after
a suitable permutation.

\end{prop}

\begin{proof}
Reorder components in ${\mathcal V}$ and ${\mathcal V}'$ in increasing order (See {Definition \ref{order of basic space}})
and group all copies of the same basic tame persistence vector space together into isotypical components,
so we have
 $${\mathcal W}=\bigoplus_{1\leq n\leq a}W_n$$
 and
 $${\mathcal W}'=\bigoplus_{1\leq n\leq b}W_n'$$

Precisely each isotypical component of ${\mathcal W}$ or ${\mathcal W}'$ is a direct sum of isomorphic basic
tame persistence vector spaces.

We only need to prove that $a=b$ and $W_n=W_n'$ for $1\leq n\leq a=b$.

Define $S({\mathcal W})=a$, the cardinality of the isotypical components of $\mathcal W$.
 We will prove the statement by induction on $S({\mathcal W})$.

If $S({\mathcal W})=0$, clearly ${\mathcal W}=0={\mathcal W}'$.

If $S({\mathcal W})>0$, write ${\mathcal W}=W_1\oplus R$ and ${\mathcal W}'=P_1\oplus R'$,

where $ R=\oplus_{2\leq n\leq a} W_n$, $ P_1=\oplus_{1\leq n\leq b_1} W_n'$,
$ R'=\oplus_{b_1+1\leq n\leq b} W_n'$,

$b_1$ is an integer such that $W_n'\leq W_1$ for $1\leq n\leq b_1$ and $W_n'>W_1$ for $n>b_1$.

Here the order of $W_n$ and $W_n'$ is determined by the order of their basic components.

Since ${\mathcal W}\cong{\mathcal W}'$, there is a pair of isomorphisms
$\omega: {\mathcal W}\rightarrow{\mathcal W}'$ and $\omega':{\mathcal W}'\rightarrow{\mathcal W}$
such that $\omega\circ\omega'=id$ and $\omega'\circ\omega=id$.

Write $\omega: {\mathcal W}\rightarrow{\mathcal W}'$ as a matrix
$$
\begin{array}{cc}
   & \begin{array}{cc}
        W_1 & R
      \end{array}
  \\
 \begin{array}{c}
   P_1 \\
   R'
 \end{array}
   & \left(
       \begin{array}{cc}
         A & C \\
         B & D \\
       \end{array}
     \right)

\end{array}
$$

Since any component of $W_1<$ any component of $R'$, by {Lemma \ref{if V and W}}, $B=0$.

So matrix form of $\omega$
is
$$
\begin{array}{cc}
   & \begin{array}{cc}
        W_1 & R
      \end{array}
  \\
 \begin{array}{c}
   P_1 \\
   R'
 \end{array}
   & \left(
       \begin{array}{cc}
         A & C \\
         0 & D \\
       \end{array}
     \right)

\end{array}
$$

Similarly, $\omega'$ has matrix form
$$
\begin{array}{cc}
   & \begin{array}{cc}
        P_1 & R'
      \end{array}
  \\
 \begin{array}{c}
   W_1 \\
   R
 \end{array}
   & \left(
       \begin{array}{cc}
         A' & C' \\
         0 & D' \\
       \end{array}
     \right)

\end{array}
$$

Since $\omega\circ\omega'=id$ and $\omega'\circ\omega=id$, we have
$$
\left(
       \begin{array}{cc}
         A & C \\
         0 & D \\
       \end{array}
     \right)
\left(
       \begin{array}{cc}
         A' & C' \\
         0 & D' \\
       \end{array}
     \right)
=I  \;\;and
\left(
       \begin{array}{cc}
         A' & C' \\
         0 & D' \\
       \end{array}
     \right)
\left(
       \begin{array}{cc}
         A & C \\
         0 & D \\
       \end{array}
     \right)
=I.
$$

So $AA'=I$, $A'A=I$, $DD'=I$, $D'D=I$.

Then $W_1\cong P_1$ and $R\cong R'$.

Since $W_1$ and $P_1$ are isomorphic, each basic component of $P_1$
must be isomorphic to the basic component of $W_1$.
Their number in $W_1$ and $P_1$ should be the same.

So $W_1=P_1$.

Since $R\cong R'$ and $S(R)=a-1<a$, by induction we finish the proof.

\end{proof}

\begin{defn}
Bar codes is a finite collection of intervals
$$
[i,j]
$$
with $i\in{\mathbb Z}_{\geq 0}$,$j\in{\mathbb Z}_{\geq0}\bigcup\{\infty\}$ and $i\leq j$.

Given a tame persistence vector space $\mathcal V$, there exist a decomposition
$$\displaystyle {\mathcal V}\cong\bigoplus_{1\leq i\leq p}  S^{r_i} \kappa[t] \oplus \bigoplus_{1\leq j\leq q} S^{m_j} T_{n_j+1}\kappa[t] $$ by {Proposition \ref{any tame persistence vector space}}.

Then assign bar code
$$
{\mathcal B}({\mathcal V})=\{[r_i,\infty],[m_j,m_j+n_j]|1\leq i\leq p,1\leq j\leq q\}
$$
to ${\mathcal V}$. Call it the bar code of the tame persistence vector space ${\mathcal V}$.

This bar code ${\mathcal B}({\mathcal V})$ is unique by {Proposition \ref{let V and V'}}.

\end{defn}

\begin{thm} \label{two tame}
{Two tame persistence vector spaces are isomorphic iff their bar codes are the same.}
\end{thm}

\begin{proof}
{Theorem \ref{two tame}} is  obtained directly from {Proposition \ref{any tame persistence vector space}} and {Proposition \ref{let V and V'}}.
\end{proof}

\begin{obs} \label{betaij}
 $\beta(i,j)$=number of intervals in ${\mathcal B}({\mathcal V})$
which contain $[i,j]$ for  $i\in{\mathbb Z}_{\geq 0}$,$j\in{\mathbb Z}_{\geq0}\bigcup\{\infty\}$ and $i\leq j$.
In particular, $\dim(V_i)$=number of intervals in ${\mathcal B}({\mathcal V})$
which contain $\{i\}$ for $i\geq0$.
\end{obs}

\begin{proof}
Suppose $$\displaystyle {\mathcal V}\cong\bigoplus_{1\leq i\leq p}  S^{r_i}  \kappa[t] \oplus \bigoplus_{1\leq j\leq q} S^{m_j} T_{n_j+1} \kappa[t] = {\mathcal W} .$$

Since ${\mathcal V}\cong{\mathcal W}$, there exist an isomorphism $f:{\mathcal V}\rightarrow{\mathcal W}$ and following
commutative diagram:
$$
\xymatrix{
V_0 \ar[d]^{f_0} \ar[r]^{\varphi_0} &
V_1 \ar[d]^{f_1} \ar[r]^{\varphi_1} &
\cdots  \ar[r]^{\varphi_{n-1}} &
V_n \ar[d]^{f_n} \ar[r]^{\varphi_n} &
V_{n+1} \ar[d]^{f_{n+1}} \ar[r]^{\varphi_{n+1}} &
\cdots
 \\
W_0 \ar[r]^{\phi_0}  &
W_1 \ar[r]^{\phi_1}  &
\cdots \ar[r]^{\phi_{n-1}} &
W_n \ar[r]^{\phi_n}  &
W_{n+1} \ar[r]^{\phi_{n+1}}  &
\cdots
}
$$

Since $f_j:{\textmd im}(\varphi_{i,j}:V_i\rightarrow V_j)\rightarrow
{\textmd im}(\phi_{i,j}:W_i\rightarrow W_j)$
is an isomorphism,
$\beta(i,j)=\dim({\textmd im}(\phi_{i,j}:W_i\rightarrow W_j))$.

Observe that $\phi_{i,j}$ is a direct sum of linear maps $\xymatrix@1{\kappa\ar[r]^{id} & \kappa}$ or
$\xymatrix@1{0\ar[r] & \kappa}$ or $\xymatrix@1{\kappa\ar[r] & 0}$. $\beta(i,j)=\dim({\textmd im}\phi_{i,j})$
is the number of linear maps $\xymatrix@1{\kappa\ar[r]^{id} & \kappa}$. Each $\xymatrix@1{\kappa\ar[r]^{id} & \kappa}$ corresponds to an interval in the bar code ${\mathcal B}({\mathcal W})$ that contains $[i,j]$.

Therefore $\beta(i,j)=$ the number of intervals in bar code ${\mathcal B}({\mathcal V})$ that contains $[i,j]$.

\end{proof}

\begin{defn}
Given a tame persistence vector space $\mathcal V$, define $\mu(i,j)=$ number of intervals
in ${\mathcal B}({\mathcal V})$ which equal to $[i,j]$.
\end{defn}

\begin{obs} \label{muijbetaij}

1) \cite{EH}
$$
\mu_(i,j) = \left\{
     \begin{array}{lr}
       \beta(i,j)-\beta(i-1,j)-\beta(i,j+1)+\beta(i-1,j+1) & 0<i\leq j<\infty\\
       \beta(0,j)-\beta(0,j+1) & i=0,0\leq j<\infty\\
       \beta(i,\infty)-\beta(i-1,\infty)& 0<i<\infty,j=\infty\\
       \beta(0,\infty) & i=0,j=\infty
     \end{array}
   \right.
$$

2)
$
\beta(i,j)=\sum_{l\leq i,\; m\geq j}\mu (l,m)
$
\end{obs}
\begin{proof}
1)
In the case $0<i\leq j<\infty$, we have
$$
\begin{array}{rl}
\mu(i,j) =& \sharp\{\tau\in{\mathcal B}({\mathcal V}) \mid \tau=[i,j]\} \\
=& \sharp\{\tau\in{\mathcal B}({\mathcal V}) \mid \tau\supseteq[i,j]\} - \sharp\{\tau\in{\mathcal B}({\mathcal V}) \mid \tau\supseteq[i-1,j]\} \\
& -\sharp\{\tau\in{\mathcal B}({\mathcal V}) \mid \tau\supseteq[i,j+1]\}+ \sharp\{\tau\in{\mathcal B}({\mathcal V}) \mid \tau\supseteq[i-1,j+1]\}\\
=& \beta(i,j)-\beta(i-1,j)-\beta(i,j+1)+\beta(i-1,j+1)
\end{array}
$$

The second identity in the above identities holds because
$$
\begin{array}{rl}
& \left \{\tau\in{\mathcal B}({\mathcal V}) \mid \tau=[i,j]\right \}  \\
=& \left \{\tau\in{\mathcal B}({\mathcal V}) \mid \tau\supseteq[i,j] \right \}
-\left \{\tau\in{\mathcal B}({\mathcal V}) \mid \tau\supseteq[i-1,j] \right \}
 -\left \{\tau\in{\mathcal B}({\mathcal V}) \mid \tau\supseteq[i,j+1]\right \}
 \end{array}
$$

The other three cases are easier to prove.

2)
Follows directly from definition.

\end{proof}

\begin{defn}
Denote $k(i,j)=\dim({\ker}(\varphi_{i,j}:V_i\rightarrow V_j))$,
$k(i) = \dim(V_i)$
with $i\leq j\in{\mathbb Z}_{\geq 0}$.
\end{defn}

\begin{obs} \label{betak}

1) $$k(i,j) = \beta(i,i) - \beta(i,j);$$

2) $$\beta(i,j) = k(i) - k(i,j).$$
\end{obs}

From the above observations, we see that the set of numbers $\{ \mu(i,j) \}$, $\{ \beta(i,j) \}$
and $\{ k(i,j), k(i) \}$ are equivalent.
Sometimes, it is more convenient to define and calculate bar codes
using the numbers $k(i,j)$ and $k(i)$ instead of $\beta(i,j)$.

\vskip .5 cm

\textbf{Edelsbrunner-Letscher-Zomorodian's Interpretation} \cite{ELZ}

\vskip .2 cm

Let $\mathcal{V}$ be a tame persistence vector space (see {Definition \ref{pvs}}).

One says that the nonzero element $x\in V_j$ is \textbf{born} in $V_i, i\leq j$ hence ``is born at time $i$'' if it is in the image of $\varphi_{i,j}$ and
not in the image of $\varphi_{i-1,j}$ and \textbf{dies} in $V_k,k\geq j+1$, hence ``dies at time $j$'' if
$\varphi_{j,k}(x)=0$ but $\varphi_{j,k-1}\neq 0$.

A subset $S\subseteq V$ ($V$ is a vector space) is called linearly independent if its elements form a collection of linearly independent vectors in $V$.
If $S_i\subseteq V_i$ is linearly independent in $V_i$ and $\varphi_{i,j}(S_i)$
is linearly independent in $V_j$, then we call $S_i$ linearly independent on
interval $[i,j]$.

With this in mind, $\beta(i,j)$  is  the maximal cardinality of linearly independent sets on the interval
$[i,j]$ and
$\mu(i,j)$ = the maximal cardinality of linearly independent sets on the interval
$[i,j]$, whose elements are born in $V_i$ and die in $V_{j+1}$.

Each interval $[i,j]$ in ${\mathcal B}({\mathcal V})$ guarantee the existence of  an element
in $V_i$, which is born in $V_i$, survives in each $V_r$($i\leq r\leq j$) and dies
in $V_{j+1}$.

In \cite{ELZ} the authors propose to collect $\mu(i,j)$ as points in the extended half space $\overline{HR}:= \{ (x,y)\in \mathbb R^2\sqcup R \times \infty \mid x\leq y\}$. With this convention Cohen-Steiner, Edelsbruner, Harer  \cite {CEH} have established a strong stability result
associated with a real valued map(cf. Stability Theorem, page 182 in \cite{EH}) which relates
the distance between two real valued functions and the bottle neck distance of their persistence
diagrams. This concept will not be used in the sequel.

\vskip .5cm

\textbf{Zomorodian - Carlsson's Interpretation} \cite{ZC}

\begin{defn}
1)
A $\kappa[t]$-module ${V}$ is the vector space over $\kappa$ equipped
with a linear map $A:{V}\rightarrow {V}$.
The module action is defined as
$$
\begin{array}{rl}
\kappa[t]\times V&\rightarrow V\\
(\sum_{i=0}^n a_i t^i,v)&\mapsto \sum_{i=0}^na_iA^i(v)
\end{array}
$$

2)
A $\kappa[t]$-module ${V}$ with a linear map $A:V\rightarrow V$ is finitely generated if there exists
$\{v_1,v_2,\cdots,v_r\}$ such that any $v\in {V}$ is a linear
combination of
 $$\{v_1,v_2,\cdots,v_r,
A(v_1),A(v_2),\cdots,A(v_r),A^2(v_1),A^2(v_2),\cdots,A^2(v_r),\cdots \}.$$

3)
A graded $\kappa[t]$-module ${V}$ is a $\kappa[t]$-module
together with a decomposition of the vector space $\displaystyle  {V}=\bigoplus_{n\geq 0}V_n$
and a linear map $A:{V}\rightarrow {V}$ such that $A(v_n)\in V_{n+1},\;\forall v_n\in V_n$.

4)
Let $\displaystyle  {V}=\bigoplus_{n\geq 0}V_n$ with the linear map $A:{V}\rightarrow {V}$
and $\displaystyle  {W}=\bigoplus_{n\geq 0}W_n$ with the linear map $B:{V}\rightarrow {V}$
be two graded $\kappa[t]$-modules.

A morphism of graded $\kappa[t]$-modules $f:{ V}\rightarrow { W}$
is a linear map of vector spaces such that $f(V_n)\subseteq W_n$ and $f\circ A=B\circ f $.

 A morphism of graded $\kappa[t]$-module $f:{ V}\rightarrow { W}$
is an isomorphism if there exists another morphism of graded $\kappa[t]$-module
$g:{ W}\rightarrow { V}$ such that $g\circ f:{ V}\rightarrow { V}$
and  $f\circ g:{ W}\rightarrow { W}$ are identities.
${ V}$ and ${ W}$ are isomorphic if there is an isomorphism between them.

\end{defn}

One of the main results of \cite{ZC} is the equivalence of the category of finitely generated
$\kappa[t]$-modules and the category of tame persistence vector spaces.

\begin{prop} \label{finitely generated graded}
Finitely generated graded $\kappa[t]$-modules  identify to tame persistence vector spaces and
so do their morphisms.
\end{prop}

\begin{proof}

1)
Let $\displaystyle  {V}=\bigoplus_{n\geq 0}V_n$ be a finitely generated $\kappa[t]$-module with a linear map $A:V\rightarrow V$ such that $A(v_n)\in V_{n+1},\;\forall v_n\in V_n$.

There exist $\{v_1,v_2,\cdots,v_r\}$ such that any $v\in {V}$ is a linear combination of
 $$(\ast) ~~ \{v_1,v_2,\cdots,v_r,
A(v_1),A(v_2),\cdots,A(v_r),A^2(v_1),A^2(v_2),\cdots,A^2(v_r),\cdots \}.$$

Since each $v_i$ is a sum of homogeneous components, WLOG, we can assume $\{v_1,v_2,\cdots,v_r\}$ themselves
are homogeneous, i.e., $v_i\in V_{n_i}$, $1\leq i\leq r$.

Let $n_i$ be the degree of $v_i$, $1\leq i\leq r$. WLOG, we can assume $n_i\leq n_{i+1}$.

Let $\varphi_n=A|_{V_n}:V_n\rightarrow V_{n+1}$, then
$$
{\mathcal V}=\{V_n,\varphi_n:V_n\rightarrow V_{n+1}|n\in Z_{\geq 0}\}
$$
is a persistence vector space.

To see that ${\mathcal V}$ is tame we denote by $\displaystyle N_0=\max_{1\leq i\leq r} n_i$.
Clearly $\dim V_n(n\leq N_0)$ is finite and in view of $(\ast)$ the linear maps $\varphi_{n,n+k}=\varphi_{n+k-1}\circ\varphi_{n+k-2}\circ\dots\circ\varphi_n$
are surjective hence $\dim V_n$ is finite for any $n>N_0$. Since $\varphi_n$ is surjective for any $n>N_0$ and $\dim V_{N_0}$ is finite,
there exists $N$ so that $\dim V_n$ is constant for $n>N$. Since any surjective map between vector spaces of the same finite dimension is
an isomorphism, the linear map $\varphi_n$ is an isomorphism for $n>N$.

Hence, $\mathcal{V}$ is a tame persistence vector space.

\vskip 0.5cm

If ${\mathcal V}=\{V_n,\varphi_n:V_n\rightarrow V_{n+1} \mid n\in Z_{\geq 0}\}$
is a tame persistence vector space,define the graded $\kappa[t]$-module $\displaystyle V=\bigoplus_{n\geq0}V_n$ with linear map $\displaystyle A=\bigoplus_{n\geq 0} \varphi_n:V\rightarrow V$.

There exist $N\geq0$ such that $\varphi_n$ is an isomorphism for $n\geq0$.
Let $\{v_1,v_2,\cdots,v_r\}$ be the set of all generators of $V_1,V_2,\cdots,V_N$,
then any $v\in {V}$ is a linear combination of  $$\{v_1,v_2,\cdots,v_r,
A(v_1),A(v_2),\cdots,A(v_r),A^2(v_1),A^2(v_2),\cdots,A^2(v_r),\cdots \}.$$

\vskip 0.5cm

2) From definitions of morphisms of finitely generated graded $\kappa[t]$-modules
and linear maps of tame persistence vector spaces, we can see that two finitely
graded $\kappa[t]$-modules are isomorphic iff the tame persistence vector spaces
associated with them are isomorphic.

\end{proof}

Recall that the ring $\kappa[t]$ is a principal ideal domain and a basic theorem in algebra \cite{L} claims
that any finitely generated modules over a principal ideal domain decompose uniquely.
In particular, any finitely generated modules over $\kappa[t]$ decomposes uniquely as a finite
direct sum of free modules $\kappa[t]$'s and torsion modules $T_{d_i}\kappa[t]$'s. As noticed by \cite{ZC}, the above result extends to
finitely generated graded modules where the free module $\kappa[t]$ is to be replaced by
$S^r\kappa[t]$ for some $r$ and torsion module $T_d \kappa[t]$ by $S^r T_d \kappa[t]$ for some $r$ and $d$.
Notice that the module $S^r A$ has the component $(S^rA)_p=A_{p-r}$ for $p\geq r$
and equal to zero for $p<r$. Note that each component $S^r\kappa[t]$ corresponds to a bar code
$[r,\infty)$ and each component  $S^r T_{d+1} \kappa[t]$ corresponds to a bar code $[r,r+d]$(cf. \cite{ZC}).

\vskip .5cm

\textbf{Quiver Representation Perspective}

\vskip .2cm

As noticed by G. Carlsson and Vin de Silva \cite{CS}, persistence vector spaces which
stabilize for $n\geq N$ can be regarded as representation of the oriented graph
$$
\xymatrix{
{\bullet_1} \ar[r] & {\bullet_2} \ar[r] & {\bullet_3} \ar[r] &  {\cdots} \ar[r] & {\bullet_{N-1}} \ar[r] & {\bullet_N} }
$$

Any such representation is a sum of indecomposable representations which are classified by the intervals
$\left[i,j\right], 1\leq i\leq j\leq N$. The interval $\left[i,j\right]$ with $j < N$ correspond to bar code $\left[i,j\right]$
while the interval $\left[i,N\right]$ to the bar code $\left[i,\infty\right)$.
The interval $\left[i,j\right]$ corresponds to the representation
$$
\xymatrix{
{0_1} \ar[r] & {0_2} \ar[r] & {\cdots} \ar[r] & {0_{i-1}} \ar[r] & {\kappa_i} \ar[r]^{id} & {\kappa_{i+1}} \ar[r]^{id} & \cdots \ar[r]^{id} &{\kappa_j} \ar[r] & 0_{j+1} \ar[r] & \cdots }
$$

This is a result in the theory of quiver representation due to P. Gabriel\cite{G}.

\subsection{Persistent Homology and Bar Codes of a Filtered Simplicial/Poly-topal Complex}\label{secPHBFSCC}

\begin{defn}
1) A filtered space $\mathcal{X}$ consists of  a space $X$ and a finite filtration
$$
X_0\subseteq X_1\subseteq\cdots\subseteq X_N = X. \eqno{(\ast)}
$$

2) Sometimes it is convenient to suppose that $X$ is embedded in $X_{\infty}$ which is contractible.
For a filtration without $X_{\infty}$, we can complete it with $X_{\infty} = C(X)$, the cone over $X$.
Then  $(\ast)$ becomes
$$
X_0\subseteq X_1\subseteq\cdots\subseteq X_N = X\subseteq X_{\infty}. \eqno{(\ast\ast)}
$$

\end{defn}

\begin{exmp}
Filtrations of Vietoris-Rips complexes (\ref{2.2.1}) - (\ref{2.2.3})  associated to a PCD $X$ with $N+1$ points in $\mathbb{R}^m$
provide examples of filtered simplicial complexes
where $X_{\infty} = R_{\epsilon_N}(X) = \Delta^n$.
\end{exmp}

Below are two other relevant examples.

\begin{defn} \label{weakly tame}

A continuous map $f:X\rightarrow  \mathbb{R}$  is called \textbf{weakly tame}
if $X$ is compact and there exists finitely many values $\min(f(X)) = t_0<t_1<\cdots<t_N = \max(f(X))$ (so called \textbf{critical values}) so that:

(i) for any $t$ the closed sub-level $X_{-\infty,t}$ is a deformation retract of an open neighborhood;

(ii) for any $i$ and $t$, $t\in [t_i t_{i+1})$, $X_{-\infty,t}$  retracts by deformation to $X_{-\infty,t_i}$.

\end{defn}
Informally this means that each sub-level is homotopically well behaved (neighborhood retract) and the topology(homotopy type) of sub-levels change only for finitely many $t$'s.

\begin{exmp}

A weakly tame map $f:X\rightarrow \mathbb{R}$ with critical values
$t_0<t_1<\cdots<t_N$ provides another example with
$X_i=f^{-1}((-\infty,t_i]), 0\leq i\leq N$, $X_{\infty} = CX_N$,
where $CX_N$ is the cone with base $X_N$.

\end{exmp}

The properties (i) and (ii) are sufficient  hypotheses to
ensure  that in each dimension the homology vector spaces of $X_{t_i}$
provide \textit{tame persistence vector spaces}.

\begin{exmp}
There is a natural filtration of an $N$-dimensional simplicial complex or polytopal complex,
 the skeleton filtration
$$
X^0\subseteq X^1\subseteq\cdots\subseteq X^N = X\subseteq X^{\infty}
$$
where $X^i(0\leq i\leq N)$ are $i$-skeletons of $X$ and $X^{\infty} =  CX$.

\end{exmp}

This subsection defines and studies persistent homology and bar codes of filtered spaces
of {finite} simplicial/polytopal complexes.

Given a filtered space $ \mathcal{K}$ of a finite simplicial/polytopal complex $K$
\begin{equation}\label{2.4.1}
K_0\subseteq K_1\subseteq\cdots\subseteq K_N = K
\end{equation}
we can associate a commutative diagram defined below.

Denote by
$C_r^s$ the ${\kappa}$-vector space $C_r(K_s)$ with basis the
$r$-simplices\,/\,$r$-cells of $K_s$,
$\partial_r^s:C_r^s\rightarrow C_{r-1}^s$ the boundary
map from $C_r(K_s)$ to $C_{r-1}(K_s)$,
$i_r^s: C_r^s \rightarrow C_r^{s+1}$  the linear map induced by the inclusion
from $K_s$ to $K_{s+1}$ (clearly $i_r^s$ is one to one).
Let $C_r^s = C_r^N$, $\partial_r^s = \partial _r^N$ and $i_r^s = identity$ when $s\geq N$,
one obtains the commutative diagram

\begin{equation}\label{2.4.2}
\begin{array}{c}
\xymatrix{
 \vdots \ar[d]^{\partial_{M+1}^0} &
   &
 \vdots \ar[d]^{\partial_{M+1}^s} &
  &
 \vdots \ar[d]^{\partial_{M+1}^N} &
 \vdots \ar[d]^{\partial_{M+1}^{N+1}} &
 \\
 C_M^0 \ar[r]^-{i_M^0} \ar[d]^{\partial_M^0} &
    \cdots\ar[r]^-{i_{M}^{s-1}}&
 C_{M}^s \ar[r]^-{i_{M}^{s}} \ar[d]^{\partial_{M}^s}&
 \cdots\ar[r]^-{i_M^{N-1}}                   &
 C_M^N \ar[r]^-{i_M^N}_\cong \ar[d]^{\partial_M^N} &
 C_M^{N+1} \ar[r]^-{i_M^{N+1}}_\cong \ar[d]^{\partial_M^{N+1}} &
 \cdots \\
 \vdots \ar[d]^{\partial_{r+1}^0} &
  &
 \vdots \ar[d]^{\partial_{r+1}^s} &
   &
 \vdots \ar[d]^{\partial_{r+1}^N} &
 \vdots \ar[d]^{\partial_{r+1}^{N+1}} &
 \\
 C_r^0 \ar[r]^-{i_r^0} \ar[d]^{\partial_r^0} &
  \cdots\ar[r]^-{i_r^{s-1}}&
 C_r^s \ar[r]^-{i_r^{s}} \ar[d]^{\partial_r^s}&
 \cdots\ar[r]^-{i_r^{N-1}}                   &
 C_r^N \ar[r]^-{i_r^N}_\cong \ar[d]^{\partial_r^N} &
 C_r^{N+1} \ar[r]^-{i_r^{N+1}}_\cong \ar[d]^{\partial_r^{N+1}} &
 \cdots \\
  C_{r-1}^0 \ar[r]^-{i_{r-1}^0} \ar[d]^{\partial_{r-1}^0} &
   \cdots\ar[r]^-{i_{r-1}^{s-1}}&
 C_{r-1}^s \ar[r]^-{i_{r-1}^{s}} \ar[d]^{\partial_{r-1}^s}&
  \cdots\ar[r]^-{i_{r-1}^{N-1}}                   &
 C_{r-1}^N \ar[r]^-{i_{r-1}^N}_\cong \ar[d]^{\partial_{r-1}^N} &
 C_{r-1}^{N+1} \ar[r]^-{i_{r-1}^{N+1}}_\cong \ar[d]^{\partial_{r-1}^{N+1}} &
 \cdots \\
  \vdots \ar[d]^{\partial_1^0} &
  &
 \vdots \ar[d]^{\partial_1^s} &
   &
 \vdots \ar[d]^{\partial_1^N} &
 \vdots \ar[d]^{\partial_1^{N+1}} &
 \\
  C_0^0 \ar[r]^-{i_0^0} \ar[d]^{\partial_0^0}&
 \cdots\ar[r]^-{i_0^{s-1}}&
 C_0^s \ar[r]^-{i_0^{s}} \ar[d]^{\partial_0^s}&
 \cdots\ar[r]^-{i_0^{N-1}}&
 C_0^N \ar[r]^-{i_0^N}_\cong \ar[d]^{\partial_0^N}&
 C_0^{N+1} \ar[r]^-{i_0^{N+1}}_\cong \ar[d]^{\partial_0^{N+1}}&
 \cdots\\
 0&
 &
 0&
 &
 0&
 0&
}
\end{array}
\end{equation}
\begin{note}
Since $K$ is finite, each row is a tame persistence vector space and
$C_r^s=0$ for $r>\dim(K)$.
\end{note}

Passing to homology with coefficient in a field $\kappa$,  consider $\mathcal H_r^s = \ker ({\partial_r^s})/ \textmd{im}({\partial_{r+1}^s})$,
and $\widetilde{i}_r^s:\mathcal H_r^s\rightarrow\mathcal H_r^{s+1}$ the linear maps induced in homology by the linear inclusions $i_r^s$.
For each $r(\geq 0)$
$$
\mathcal H_r({\mathcal K}) := \{\mathcal H_r^s,\widetilde{i}_r^s:\mathcal H_r^s\rightarrow\mathcal H_r^{s+1} | s\in \mathbb{Z}_{\geq 0}\}
$$
is a persistence vector space with bar code ${\mathcal B}(\mathcal H_r({\mathcal K}))$.

Following \cite{ELZ}, the collection of vector spaces $\mathcal H_r^{s,p} = \textmd{im}(\widetilde{i_r}^{s,s+p}:\mathcal H_r^s\rightarrow\mathcal H_r^{s+p})$
is referred to as the persistent homology.

2) The collection of bar codes ${\mathcal B}(\mathcal H_r({\mathcal K}))$ for all $r$
is denoted by $\mathcal{B(K)}$ and referred to as the bar code of ${\mathcal K}$.

Suppose $\mathcal K$ is a filtered simplicial\,/\,polytopal complex as in (\ref{2.4.1}),
$$
K_0\subseteq K_1\subseteq\cdots\subseteq K_N = K,
$$
we calculate the bar codes of $\mathcal K$ in
two cases: $\kappa = \mathbb{Z}_2$ and $\kappa = \mathbb{R}$.
We consider only filtered simplicial complex below, filtered polytopal complex can be
dealt with similarly.

\vskip .5cm

\textbf{Case A:} $\kappa = \mathbb{Z}_2$

In this case, we can calculate the bar codes of $\mathcal K$ using the persistent algorithm
given in \cite{EH}. There are two steps of this algorithm: 1) matrix reduction; 2) pairing.

\textbf{Matrix reduction} (algorithm (ELZ))

Suppose $S(K)$ is the set of all simplices of $K$ and $\sharp S(K) = m$.

Consider the function $f_{ind}: S(K)\rightarrow \{0, 1, \cdots, N\}$,
such that $f_{ind}(\sigma) = i$ for $\sigma\in K_i \setminus K_{i-1}$.

Choose a \textit{compatible ordering} of the simplices, that is, an ordering of all simplices of $K$,
such that $\sigma\prec\tau$ if $f_{ind}(\sigma)<f_{ind}(\tau)$ or if $\sigma$ is a face of $\tau$.
After we order all simplices of $K$ according to this compatible ordering, we get a sequence of
simplices
$\sigma_1, \sigma_2,\cdots, \sigma_m$.

Consider the $m$-by-$m$ boundary matrix $\partial$ given by
$$
\displaystyle\partial[i, j] = \left\{ \begin{array}{ll}
1 & \text{if $\sigma_i$ is a codimension$-1$ face of $\sigma_j$;}\\
0 & \text{otherwise.}
\end{array}\right.
$$

Let $low(j)$ be the row index of the lowest $1$ in column $j$.
If the entire column is zero, then $low(j)$ is undefined.
We call R $reduced$ if $low(j)\neq low(j_0)$ whenever $j$ and $j_0$, with $j\neq j_0$,
specify two non-zero columns.
The algorithm reduces $\partial$ by adding columns from left to right.
\begin{center}
\textbf{Algorithm 3.1}

\begin{tabular}{|l|}
\hline
R=$\partial$\\
for $j = 1$ to $m$ do\\
\;\;\; while there exists $j_0<j$ with $low(j_0) = low(j)$ do\\
\;\;\;\;\;\; add column $j_0$ to column $j$\\
\;\;\; endwhile\\
endfor.\\
\hline
\end{tabular}
\end{center}

The running time is at most cubic in the number of simplices \cite{KB79}, \cite{Mo05}.
In matrix notation, the algorithm computes the reduced matrix as $R = \partial\cdot V$.
Since each simplex is preceded by its proper faces, $\partial$ is upper triangular.
Since we only add from left to right, $V$ is also upper triangular and so is $R$
(See page 152-153, \cite{EH} for details).

\textbf{Pairing}

Once there are two cases of columns of reduced matrix $R$.

Case 1: column $j$ of $R$ is zero. We call $\sigma_j$ positive since it creates
a new cycle and thus gives birth to a new homology class unless it dies in the same
filtration.

Case 2: column $j$ or $R$ is non-zero. We call $\sigma_j$ negative because it provides
the death to a homology class.

See page 154-155, \cite{EH} for details.

We can read bar codes directly from the matrix $R$.
Suppose column $j$ is zero and $\dim(\sigma_j) = r$.
If there exists a column $k$ with $low(k) = j$
and $f_{ind}(\sigma_k)>f_{ind}(\sigma_j)$
then column $j$ provides a closed interval $[f_{ind}(\sigma_j),f_{ind}(\sigma_k)-1]\in {\mathcal B}(\mathcal H_r({\mathcal K}))$.
If there exists a column $k$ with $low(k) = j$
and $f_{ind}(\sigma_k)=f_{ind}(\sigma_j)$
then column $j$ doesn't provides any interval since it indicates a cycle which  dies in the same filtration.
Note that $f_{ind}(\sigma_k)<f_{ind}(\sigma_j)$ is impossible according to the definition of compatible
ordering.
If there is no  column $k$ with $low(k) = j$
then  column $j$ provides an infinite interval $[f_{ind}(\sigma_j),\infty]\in {\mathcal B}(\mathcal H_r({\mathcal K}))$.

\vskip .5cm

\textbf{Case B:}  $\kappa = \mathbb{R}$

For each $s$ consider the chain complex $\mathcal C(K_s)$ with coefficients in $\mathbb{R}$
$$
\xymatrix{
\cdots \ar[r]^{\partial_{r+2}^s} & C_{r+1}^s \ar[r]^{\partial_{r+1}^s} & C_r^s \ar[r]^{\partial_r^s} & C_{r-1}^s \ar[r]^{\partial_{r-1}^s} & \cdots \ar[r]^{\partial_2^s} & C_1^s \ar[r]^{\partial_1^s} & C_0^s \ar[r]^-{\partial_0^s} & =0
}
$$
equipped with the scalar products defined by the standard basis(see subsection \ref{secHodge}).

Let

(1) $\delta_{r-1}^s=(\partial_{r}^s)^*$ the adjoint operator of $\partial_{r}^s$ \footnote  {When represented as a matrix w.r.t. the standard base, ``adjoint'' w.r.t. the scalar product defined by the standard base  is actually ``transpose''. } for $r\in \mathbb Z_{\geq0}$,

(2) $\Delta_r^s=\partial_{r+1}^s\circ\delta_r^s+\delta_{r-1}^s\circ\partial_r^s:C_r^s\rightarrow C_r^s$ for $r\in \mathbb Z_{\geq0}$ and

(3) $(C_r^s)_+=\textmd{im}(\partial_{r+1}^s)$,
$(C_r^s)_- = \textmd{im}(\delta_{r-1}^s)$ and
$H_r^s=\ker(\Delta_r^s)$, with
$$H_r^s=\ker(\delta_r^s)\cap \ker(\partial_r^s)$$
and
$$C_r^s=(C_r^s)_+\oplus H_r^s\oplus (C_r^s)_-$$
where $(C_r^s)_+$, $ H_r^s$ and $(C_r^s)_-$ are pairwise orthogonal.

Recall that $[A]$ represents the orthonormalization of matrix $A$(see {Lemma {\ref{orth}}}).
With respect to the standard basis provided by simplices\,/\,cells, $\partial_r^s$ can be regarded as an $n_{r-1}^s\times n_r^s$ matrix.
In view of the above considerations, the following linear maps are orthogonal projections onto $(C_r^s)_+$,$(C_r^s)_-$
and $H_r^s$, respectively.
\begin{equation}\label{2.4.3}
\begin{array}{rrl}
(p_r^s)_+: & C_r^s &\rightarrow C_r^s \\
     & y  &\mapsto  [\partial_{r+1}^s][\partial_{r+1}^s]^T y\\

(p_r^s)_-: & C_r^s &\rightarrow C_r^s \\
     & y  &\mapsto  [(\partial_{r}^s)^T][(\partial_{r}^s)^T]^T y\\

(p_r^s)_H: & C_r^s &\rightarrow C_r^s \\
     & y  &\mapsto  (I_{n_r^s}-[\partial_{r+1}^s][\partial_{r+1}^s]^T-[(\partial_{r}^s)^T][(\partial_{r}^s)^T]^T)y
\end{array}
\end{equation}

$$
\begin{array}{lrl}
j_r^s: & H_r^s=\ker(\delta_r^s)\cap \ker(\partial_r^s) & \rightarrow  \mathcal H_r^s=\ker({\partial_r^s})/\textmd{im}({\partial_{r+1}^s})\\
 & x  & \mapsto x+\textmd{im}(\partial_{r+1}^s)
\end{array}
$$
and
$$
\begin{array}{lrl}
k_r^s: & \mathcal H_r^s=\ker({\partial_r^s})/\textmd{im}({\partial_{r+1}^s}) & \rightarrow  H_r^s=\ker(\delta_r^s)\cap \ker(\partial_r^s)\\
 & y+\textmd{im}(\partial_{r+1}^s)  & \mapsto (p_r^s)_H(y)
\end{array}
$$
is a pair of isomorphisms between $H_r^s$ and $\mathcal H_r^s$.

Since $H^r_s$ identifies to $\mathcal H_r^s$ by the
 pair of isomorphisms $j_r^s$ and $k_r^s$, the persistence vector space
$$
 H_r({\mathcal K})=\{ H_r^s,g_r^s: H_r^s\rightarrow H_r^{s+1} \mid s\in \mathbb{Z}_{\geq 0}\},
$$
where $g_r^s=k_r^{s+1}\circ \widetilde{i}_r^s\circ j_r^s$,
is isomorphic to
$$
\mathcal H_r({\mathcal K})=\{\mathcal H_r^s,\widetilde{i}_r^s:\mathcal H_r^s\rightarrow\mathcal H_r^{s+1} \mid s\in \mathbb{Z}_{\geq 0}\}
$$

By {Theorem \ref{two tame}} one has $\mathcal B(H_r({\mathcal K})) = \mathcal B(\mathcal H_r({\mathcal K})) $.

\begin{nota}
1) Denote $$g_r^{s,t}=g_r^{t-1}\circ\cdots\circ g_r^s:H_r^s\rightarrow H_r^t$$ for $s<t$
and $$g_r^{s,s}=identity:H_r^s\rightarrow H_r^s$$ with $s,t\in \mathbb Z_{\geq0}$.

Define $i_r^{s,t}:C_r^s\rightarrow C_r^t$ and $\widetilde{i}_r^{s,t}:\mathcal H_r^s\rightarrow \mathcal H_r^t$
in similar way.

2) Denote $\beta_r(s,t)=\dim({\textmd im}(g_r^{s,t}:H_r^s\rightarrow H_r^t))$ with $s\leq t\in{\mathbb Z}_{\geq 0}$.

\begin{note} $\beta_r(s,s)=\dim H_r^s$ and  $\beta_r(s,t)=\beta_r(s,t+1)$ for $t\geq N$. \end{note}

Hence $\beta_r(s,\infty)=\beta_r(s,m)$ for $m$ larger than $s$ and $N$.

Notice that $\beta_r(s,t)$=number of intervals in $\mathcal B(H_r({\mathcal K}))$
which contain $[s,t]$ for  $0\leq s\leq t\leq \infty$. (See {Observation \ref{betaij}})

3) Denote $\mu_r(s,t)$=number of intervals in $\mathcal B(H_r({\mathcal K}))$
which equal to $[s,t]$ for  $0\leq s\leq t\leq \infty$. We have
\begin{equation}\label{2.4.4}
\mu_r(s,t) = \left\{
     \begin{array}{lr}
       \beta_r(s,t)-\beta_r(s-1,t)-\beta_r(s,t+1)+\beta_r(s-1,t+1) & 0<s\leq t<\infty\\
       \beta_r(0,t)-\beta_r(0,t+1) & s=0,0\leq t<\infty\\
       \beta_r(s,\infty)-\beta_r(s-1,\infty)& 0<s<\infty,t=\infty\\
       \beta_r(0,\infty) & s=0,t=\infty
     \end{array}
   \right.
\end{equation}

\end{nota}

\begin{note}
From the above formula and {note} in 2) we only need to
know $\beta_r(s,t)$ for all $0\leq s\leq t \leq N$
in order to calculate $\mu_r(s,t)$ and the bar code $\mathcal B(H_r({\mathcal K}))$.
\end{note}

\begin{thm} \label{beta rss}
$\beta_r(s,s)=rank((p_r^s)_H)$ and $\beta_r(s,t)=rank((p_r^t)_H\circ i_r^{s,t} \circ (p_r^s)_H)$ for $s<t$,
where
$$
i_r^{s,t}=\left(
          \begin{array}{c}
            I_{n_r^s} \\
            0_{(n_r^t-n_r^s)\times n_r^s} \\
          \end{array}
        \right).
$$
For the definition of $(p_r^s)_H$, refer to (\ref{1.2.1}) and (\ref{2.4.3}).

\end{thm}

We will use MATLAB and in order to avoid the use of
the function ``rank'' which is not  reliable for large
matrices we will need the following observations to
complement  Theorem \ref{beta rss}.

\begin{obs}\label{rankA}
The rank of a real matrix $A$ equals to the number of positive eigenvalues of
$A A^T$ or $A^T A$.
\end{obs}

\begin{obs}\label{betarss}
$\beta_r(s,s) = \dim(H_r^s) = \dim(C_r^s)-rank(\partial_{r+1}^s)-rank(\partial_r^s)$.
\end{obs}

\begin{obs}\label{nonzero}
If $\beta_r(s,t)\neq 0,s<t$, then

i) $H_r^i$ is nonzero for all $s\leq i\leq t$;

ii) $\beta_r(s,i)$ is nonzero for all $s\leq i\leq t$.

\end{obs}

\vskip .5cm

\textbf{Simultaneous Persistence}

\vskip .2cm

In the {case} $\kappa= \mathbb Z_2$   we describe how to calculate the bar codes of
$\mathbb{Z}_2$-homology groups of the filtered simplicial complex
$$
K_0\subseteq K_1\subseteq\cdots\subseteq K_N = K.
$$

Let's discuss a more general case where the filtration has two directions.

Suppose $X$ is a polytopal complex and $X^{\pm}\subset X$ are two subcomplexes
with $X_0=X^+\cap X^-$.
Suppose $X^+$ and $X^-$ are equipped with finite filtration
$X_0^+\subseteq X_1^+\subseteq\cdots\subseteq X_{N^+}^+$ and
$X_0^-\subseteq X_1^-\subseteq\cdots\subseteq X_{N^-}^-$
where $X_0^{\pm} = X_0$, $X_{N^-}^- = X^-$ and $X_{N^+}^+ = X^+$.
\medskip

Define the simultaneous persistence number $\omega_r(s,t)$ of the above
filtration to be the maximal number of linearly independent elements in $H_r(X_0)$
which dies in $X^-$ at $s$ and in $X^+$ at $t$(See \cite{BDD} for details).
These numbers are useful in the calculation of relevant level persistence numbers.

One can put the cells of $X$ in three groups I, II, III.
The group I contains the cells of $X_0$,
the group II the cells of $X^-\setminus X_0$ and
the group III the cells of $X^+\setminus X_0$.

Consider the following filtration of $X$
$$
X_0\subseteq X_1\subseteq \cdots \subset X_{N^-+N^+} = X
$$
where $X_i = X_i^-$ for $0\leq i\leq N^-$, $X_i = X^-\cup X_{i-N^-}^+$
for $N^-+1\leq i\leq N^-+N^+$.

Suppose $S(X)$ is the set of all cells of $X$ and $\sharp S(X) = m$.

We can define a function $f_{ind}: S(X)\rightarrow \{0, 1, \cdots, N^-+N^+\}$,
such that $f(\sigma) = i$ for $\sigma\in X_i \setminus X_{i-1}$.

We use a \textit{compatible ordering} of the cells, that is, an ordering of all cells of $X$,
such that $\sigma\prec\tau$ if $f_{ind}(\sigma)<f_{ind}(\tau)$ or if $\sigma$ is a face of $\tau$.
After we order all cells of $X$ according to this compatible ordering, we get a sequence of
cells
$\sigma_1, \sigma_2,\cdots, \sigma_m$.

With this compatible ordering, the incidence matrix $M$ is given by
$$ M =   \begin{tabular}{||c||}
$
\begin{array}{ccc}
A & B^- & C^+\\
0 & B & 0\\
0 & 0 & C
\end{array}$
\end{tabular}
$$
with $A$ being the incidence matrix for the subcomplex $X_0$,
$M^- = \begin{tabular}{||c||}
$
\begin{array}{ccc}
A & B^-\\
0 & B
\end{array}$
\end{tabular}
$
being the incidence matrix for $X^-$ and
$M^+ = \begin{tabular}{||c||}
$
\begin{array}{ccc}
A & C^+\\
0 & C
\end{array}$
\end{tabular}
$
being the incidence matrix for $X^+$.

The matrix $M$ is said to be in \textit{relative reduced form} if both $M^-$
and $M^+$ are in reduced form.
This relative reduced form $R_{rel}$ can be achieved by applying {Algorithm 2.4.1}
to matrix $M$.

We can read the number $\omega_r(i,j)$ directly from the relative reduced form $R_{rel}$ of $M$.
$\omega_r(i,j) = $ the number of all triples $(\sigma_k, \sigma_{k'}, \sigma_{k''})$
with $\sigma_k \in X_0$, column $k$ of $R_{rel}$ is zero, $\dim(\sigma_k) = r$, $low(k') = low(k'') = k$, $\sigma_{k'}\in X_i^-\setminus X_{i-1}^-$
and $\sigma_{k''}\in X_j^+\setminus X_{j-1}^+$.

\subsection{Calculate the Bar Codes of a PCD}\label{secAlg}

Given a point cloud data $X=\{x_1,x_2,\cdots,x_p\}\in \mathbb{R}^n$(suppose all
points in $X$ are distinct), we will first consider its distance matrix
$$
D=\left(
    \begin{array}{cccc}
      d(x_1,x_1) & d(x_1,x_2) & \cdots & d(x_1,x_p) \\
      d(x_2,x_1) & d(x_2,x_2) & \cdots & d(x_2,x_p) \\
     \vdots & \vdots & \ddots & \vdots \\
      d(x_p,x_1) & d(x_p,x_2) & \cdots & d(x_p,x_p) \\
    \end{array}
  \right)
$$
Notice that the $D$ is a symmetric matrix, the diagonal elements of $D$ are zeros and
all the other elements are positive real numbers.

Let $S$ be the set of the entries in D and order these numbers
increasingly (since they will be our epsilons), so we have
$$
\mathcal{E}:0=\epsilon_0<\epsilon_1<\epsilon_2<\cdots<\epsilon_N.
$$

$\mathcal{E}$ is exactly (\ref{2.2.0}).

{\bf Given a positive integer $m$ and $P$, let's
compute the {bar code of PCD $X$ up to dimension $m-1$
and step $P$}}(for definition see (\ref{2.2.3})).

To simplify the notation, from now on we will write $K_i$ instead of $R_{\epsilon_i}(X,m)$.
$K_0$ is just a collection of vertices $\{x_1,x_2,\cdots,x_p\}$.

We will represent  the vertex $x_{i_0}$ simply
by the index $i_0$, the edge $[x_{i_0},x_{i_1}]$ by $[i_0,i_1]$
and $k$-simplex $[x_{i_0},\cdots,x_{i_k}]$ is represented
by $[i_0,i_1,\cdots,i_k]$.

We want first to have the information about the filtered simplicial complex stored as

(a) An $(m+1)\times(P+1)$ matrix named $\textbf{Dimension}$ where
the entry $\textbf{Dimension}_{r+1,s+1}$ is the integer $n_r^s =$ dimension of $C^s_r$.

(b) An $(m+1)$ array of matrices named $\textbf{Simplex}$.
The component $\textbf{Simplex}_{r+1}$ is a matrix describing $r$-simplices of $K_P$; precisely is a matrix of $r+1$ columns,
each row representing an  $r$-simplex $[i_0,i_1,\cdots,i_k]$.
These rows are ordered consistently with the filtration.
The order for simplices of $K_s \backslash K_{s-1}$ is however arbitrary.

There is an obvious way to create these data, which is not very efficient but worth to mention.
Given $K_s$, we can easily find its $1$-simplices from
the distance matrix $D$. All the nonzero numbers in $D$ which are less
than $\epsilon_s$ correspond to the edges in $K_s$. In order to find the $r$-simplices,
we scan all possible combination of $r+1$ numbers of $\{1,2,\cdots,p\}$.
$[i_0,i_1,\cdots,i_r]$ ($i_0<i_1<\cdots<i_r$)belongs to the $r$-simplices of $K_s$ iff
every pair $[i_j,i_k]$ belongs to the $1$-simplices of $K_s$.
As $s$ increases, we have more $r$-simplices in $K_s$.
Store each $r$simplex as a row in $\textbf{Simplex}_{r+1}$
in filtration order. In another word,
if $[i_0,i_1,\cdots,i_r]$ belongs to $K_s$ but not in $K_{s-1}$, then it is ``new''
and can be added as a new row to the rows provided by the $r$-simplices of $K_{s-1}$.

We use however the package \textbf{JPlex} \cite{HA} \cite{SJ} to compute $\textbf{Dimension}$ and $\textbf{Simplex}$
,
which can be download from http://comptop.stanford.edu/programs/jplex/.

Once we get $\textbf{Simplex}$, we get the standard basis of $C_r^P$,
which is stored in $\textbf{Simplex}_{r+1}$, for $0\leq r\leq m$.
Given  $\textbf{Simplex}_{r}$ and $\textbf{Simplex}_{r+1}$, we compute $\partial_r^P$ with respect
to the standard basis:

Write $\textbf{Simplex}_{r+1}$ as $$\left(\begin{array}{ccc}a_{11} & \cdots & a_{1,r+1} \\ \vdots & \ddots & \vdots \\a_{n_r^P,1} & \cdots & a_{n_r^P,r+1}\end{array}\right)$$
 and
  $\textbf{Simplex}_{r}$ as $$\left(\begin{array}{ccc}b_{11} & \cdots & b_{1,r} \\ \vdots & \ddots & \vdots \\b_{n_{r-1}^P,1} & \cdots & a_{n_{r-1}^P,r}\end{array}\right),$$
   where  $n_r^P$ and $n_{r-1}^P$ can be obtained from $\textbf{Dimension}$.

$\partial_r^P$ is an $n_{r-1}^P \times n_{r}^P$ matrix with element $(\partial_r^P)_{i,j}$ equals to the incidence number
of the simplices $[b_{i,1},\cdots,b_{i,r}]$ and $[a_{j,1},\cdots,a_{j,r+1}]$.
If $[b_{i,1},\cdots,b_{i,r}]$ is not contained in $[a_{j,1},\cdots,a_{j,r+1}]$, $(\partial_r^P)_{i,j}=0$.
Otherwise, $[b_{i,1},\cdots,b_{i,r}]$ equals to $[a_{j,1},\cdots,\hat{a_{j,k}},\cdots,a_{j,r+1}]$ for some $k$, where  `` $\hat{}$ ''
means deletion. If $k$ is odd, $(\partial_r^P)_{i,j}=1$; if $k$ is even, $(\partial_r^P)_{i,j}=-1$.

$\partial_r^s$ is the $n_{r-1}^s \times n_{r}^s$ upper-left block of $\partial_r^P$.

Once $\textbf{Dimension}$ and $\textbf{Simplex}$ are determined, we apply
{Theorem \ref{beta rss}} and the {Observations \ref{rankA}, \ref{betarss} and \ref{nonzero}} to get $\beta_r(s,t)$
and then $\mu_r(s,t)$ according to (\ref{2.4.4}).

\vskip 0.5cm

Here is a brief description of functions written in Matlab. Those Matlab programs contained in ``Computing\_Barcode\_of\_PCD.zip'' can be downloaded from \verb|http://people.math.osu.edu/du.50|
\vskip 0.5cm

$\textbf{ Initial Data}$:

A $p\times n$ matrix $\textbf{ X}$, which stores $p$ points $x_1,\cdots,x_p$ in $\mathbb R^n$
as rows, two positive integers $\textbf {m}$ and $\textbf{S}$.
$\textbf {m}$ is the given upper bound of dimension of restricted Vietoris-Rips Complex.
$\textbf{S}$ is the given upper bound of steps of filtration.
 $\textbf{P} = \min(\textbf{S}, \textbf{N})$. See (\ref{2.2.3}).

\vskip 0.5cm

1. \textbf{Function getDistance}

$\textbf{ Input}$: $\textbf{ X}$.

$\textbf{ Output}$: A  $p\times p$ upper triangular matrix $\textbf{ D}$. $\textbf{D}_{ij} = d(x_i,  x_j), 1 \leq i \leq p - 1, i + 1 \leq j \leq p$.

\vskip 0.5cm

2. \textbf{Function scaleX}

$\textbf{ Input}$: $\textbf{ X}$, $\textbf{ D}$.

$\textbf{ Output}$: $\textbf{ X}$.

Given a point cloud data $\textbf{ X}$, we find all the distances between two points and sort them as
$\epsilon_0<\epsilon_1<\cdots<\epsilon_N$.
$\textbf{ Note: }$ If $\epsilon_{i+1} - \epsilon_i$ is very small(less than $3\cdot \textbf{diam(X)}\cdot \textbf{ eps}$,
the upper bound of round off errors,
where $\textbf{diam(X)}$ is the maximal distance between any two points of $\textbf X$, $\textbf {eps} = 2.2204\times10^{-16}$ is floating-point relative accuracy in Matlab), we delete $\epsilon_{i+1}$.
We scale $\textbf{ X}$ such that the minimum difference between
two consecutive distances is larger than $3\times 10^{-4}$. For details and more explanation,
read the comments in my program.

\vskip 0.5cm

3. \textbf{Function getEpsilon}

$\textbf{ Input}$: $\textbf{ D}$, $\textbf{ S}$.

$\textbf{ Output}$: Two positive integers $\textbf{P} = \min(\textbf{S}, \textbf{N})$ and $\textbf{ e} = \epsilon_P$.
Two row vectors $\textbf{ epsiorg}$, which stores different distances  $\epsilon_0<\epsilon_1<\cdots<\epsilon_N$ of $\textbf{ D}$ in increasing order and $\textbf{ epsiavg}$, which stores  $$\frac{\epsilon_0+\epsilon_1}{2},\cdots,\frac{\epsilon_k+\epsilon_{k+1}}{2},\cdots,\frac{\epsilon_{N-1}+\epsilon_N}{2},\epsilon_N+\frac{1}{2}.$$

\vskip 0.5cm

4. \textbf{Function getDimensionSimplex}

$\textbf{ Input}$: $\textbf{ X}$, $\textbf{ P}$, $\textbf{ epsiavg}$, $\textbf {m}$ and $\textbf{e}$.

$\textbf{ Output}$:

a) An integer ${\bf m_a} = \displaystyle \dim(R_{\epsilon_P}(X,m))$.

b) An $({\bf m_a}+1)\times(\textbf{P}+1)$ matrix  $\textbf{ Dimension}$.
$\textbf{ Dimension}_{r+1,s+1} = n_r^s$, the dimension of $C_r^s$.

c) An ${\bf m_a}+1$ array $\textbf{ Simplex}$.

The component $\textbf{ Simplex}_{r+1}$ is a matrix with $r+1$ columns, which stores $r$-simplex of $R_{\epsilon_P}(X,m)$
as rows in filtration order.

$\textbf{ Note}$: Here we use the package \textbf{edu.stanford.math.plex.*}
of \textbf{JPlex}\cite{HA}\cite{SJ}.

\vskip 0.5cm

5. \textbf{Function getDeltaP}

$\textbf{ Input}$: $\textbf{ Dimension}$, $\textbf{ Simplex}$ and ${\bf m_a} $.

$\textbf{ Output}$: An ${\bf m_a} $ array $\textbf{ DeltaP}$.
 $\textbf{DeltaP}_r $ = the matrix of $ \partial_r^P: C_r^P \rightarrow C_{r-1}^P $ with respect to the
 the  standard basis of $ C_r^P $ and $ C_{r-1}^P $.

\vskip 0.5cm

6. \textbf{Function getDelta}

$\textbf{ Input}$: $\textbf{ DeltaP}$, $\textbf{ Dimension}$, $r$, $s$

$\textbf{ Output}$: The matrix $ \partial_r^s: C_r^s \rightarrow C_{r-1}^s $ with respect to standard basis of $C_r^s$ and $C_{r-1}^s$.

\vskip 0.5cm

7. \textbf{Function getRank}

$\textbf{ Input}$: $\textbf{P}$, $\textbf{m}$, and ${\bf m_a}$.

$\textbf{ Output}$: An $(\textbf{m}+1)\times (\textbf{P}+1)$ matrix $\textbf{ R}$.
$\textbf{ R}_{r+1,s+1}$ = rank of $\partial_r^s$.

$\textbf{ Algorithm}$: $\textbf{ R}_{r+1,s+1}$ = number of eigenvalues of $\partial_r^s \times {\partial_r^s}'$
which is greater than $10^{-8}$. Here we use  { Observation \ref{rankA}}.

\vskip 0.5cm

8. \textbf{Function getDimensionHrs}

$\textbf{ Input}$: $\textbf{ Dimension}$, $\textbf{ R}$, $\textbf{P}$, $\textbf{m}$ and ${\bf m_a}$.

$\textbf{ Output}$: An integer ${\bf m_c}$ and an $({\bf m_c}+1)\times (\textbf{P}+1)$ matrix $\textbf{ dimHrs}$.
$\textbf{ dimHrs}_{r+1,s+1} = \dim(H_r^s)$.

$\textbf{ Algorithm}$: $\textbf{ dimHrs}_{r+1,s+1} = \textbf{ Dimension}_{r+1,s+1} - \textbf{ R}_{r+2,s+1} - \textbf{ R}_{r+1,s+1}$.

Here we use use  { Observation \ref{betarss}}.

\vskip 0.5cm

9. \textbf{Function getHarmonicProjection}

$\textbf{ Input}$: $\textbf{ Dimension}$, $r$ and $s$.

$\textbf{ Output}$: The harmonic projection matrix ${(p_r^s)_H}$ from $C_r^s$ to itself with respect to the standard basis of $C_r^s$.

$\textbf{ Algorithm}$: See (\ref{2.4.3}).

\vskip 0.5cm

10. \textbf{Function getBeta}

$\textbf{ Input}$: $\textbf{ Dimension}$, $\textbf{P}$, $\textbf{ dimHrs}$ and ${\bf m_c}$.
${\bf m_c}$ is the maximal dimension of nonzero $H_r^s$'s.

$\textbf{ Output}$: A $(\textbf{P}+1)\times (\textbf{P}+1) \times ({\bf m_c}+1)$ matrix $\textbf{Beta}$.
$\textbf{Beta}(s+1,t+1,r+1) = \beta_r(s,t)$.

$\textbf{ Algorithm}$:

$
\beta_0(\cdot,\cdot) = \left(\begin{array}{cccc}\dim(H_0^0) & \dim(H_0^1) & \cdots & \dim(H_0^P) \\ & \dim(H_0^1) & \cdots & \dim(H_0^P) \\ &  & \ddots & \vdots \\ &  &  & \dim(H_0^P)\end{array}\right)
$;

$
\beta_r(s,s) = rank((p_r^s)_H) = \dim(H_r^s)
$(see { Observation \ref{betarss}});

$
\beta_r(s,t) = rank((p_r^t)_H\circ i_r^{s,t} \circ (p_r^s)_H)
$
for $s<t$ (see {Theorem \ref{beta rss}}).

Apply  { Observation \ref{rankA}} when computing rank as in { function getRank}.

Apply  { Observation \ref{nonzero}} to avoid unnecessary computing.

\vskip 0.5cm

11. \textbf{Function getMu}

$\textbf{ Input}$: $\textbf{Beta}$, $\textbf{P}$ and ${\bf m_c}$.

$\textbf{ Output}$: A $(\textbf{P}+1)\times (\textbf{P}+1) \times ({\bf m_c}+1)$ matrix $\textbf{Mu}$.
$\textbf{Mu}(s+1, t+1, r+1) = \mu_r(s, t)$

$\textbf{ Algorithm}$: Apply (\ref{2.4.4}).

\vskip 0.5cm

12. \textbf{Function getBar codeMatrix}

$\textbf{ Input}$: $\textbf{Mu}$, $\textbf{P}$, and ${\bf m_c}$.

$\textbf{ Output}$: $\textbf{A}$, a $q\times 4$ matrix which stores the bar code in increasing order
according to dimension, left endpoint and right endpoint.
 $M$ intervals $[a, b]$ in bar code ${\mathcal B}(\mathcal H_r({\mathcal K}))$ will be represented by [r a b M] in $\textbf{A}$.

\vskip 0.5cm

13. \textbf{drawBar code}

$\textbf{ Input}$: $\textbf{A}$ and $\textbf{P}$.

$\textbf{ Output}$: A picture of bar code stored in \textbf{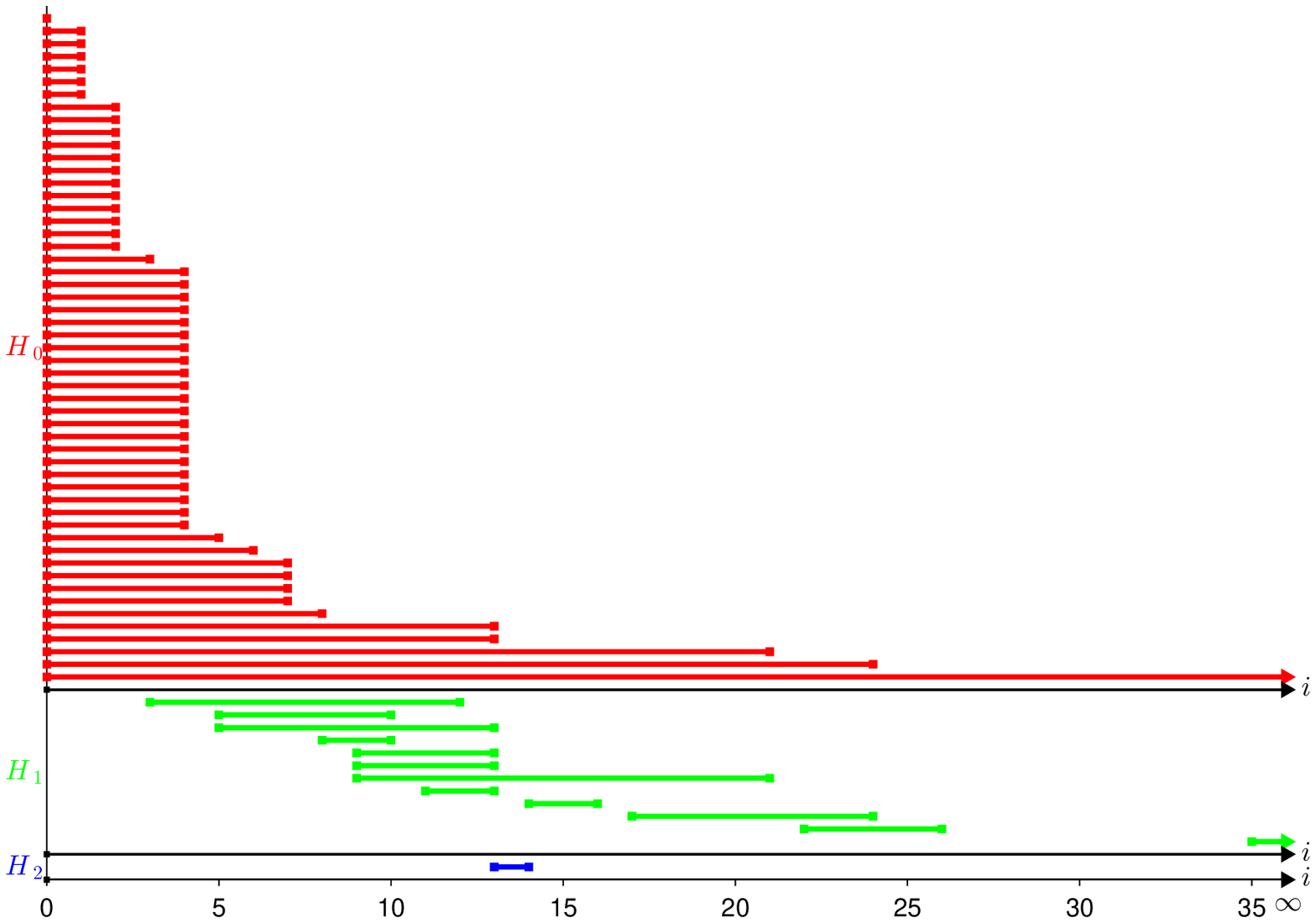}.

\vskip 0.5cm

14. \textbf{Function main}

$\textbf{ Input}$: $\textbf{X}$, $\textbf{S}$ and $\textbf{m}$.

$\textbf{ Output}$: A matrix $\textbf{A}$ stored in \textbf{A.mat} and a picture of bar code stored in \textbf{barcode.eps}.

$\textbf{ Algorithm}$: Run above functions except 6. and 9. in order.

\subsection{Numerical Experiment}\label{secNE}

We have a PCD in $\mathbb R^2$ with $53$ points:
\begin{figure}[h!]

\centering
\setlength\fboxsep{0.5pt}
\setlength\fboxrule{0.5pt}
\fbox{\includegraphics[scale=.8]{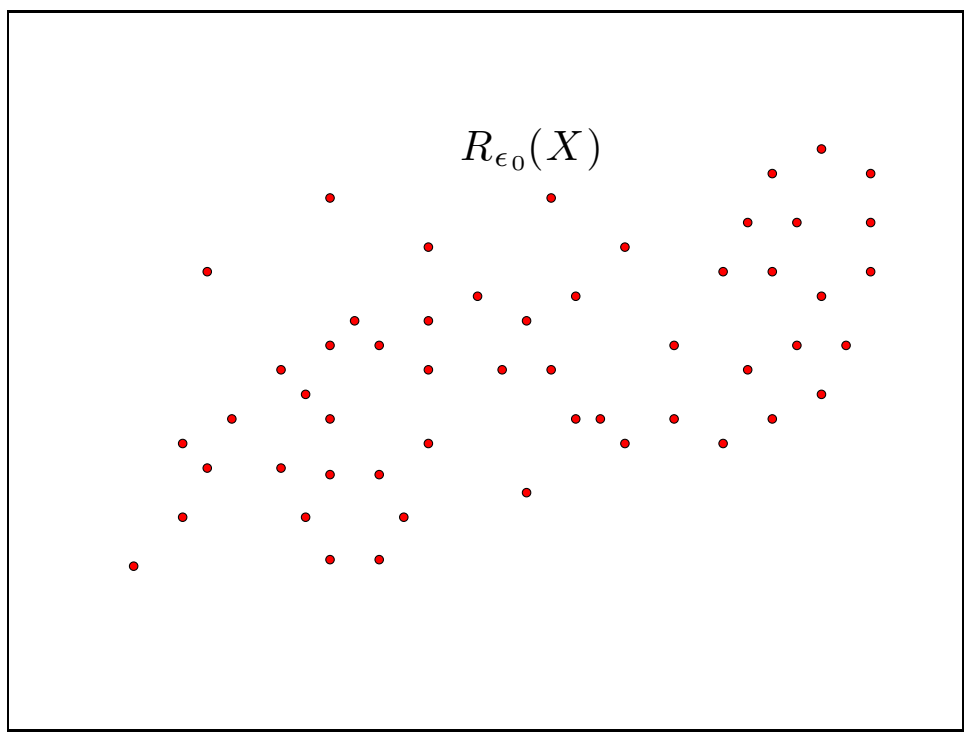}}

\caption{A PCD in $\mathbb R^2$ with $53$ points}
\end{figure}

The restricted filtered simplicial complex of this PCD ($m = 4$, $P = 35$) is
\begin{figure}
\centering
\vskip -6pt
$
\begin{array}{cc}
\setlength\fboxsep{0.5pt}
\setlength\fboxrule{0.5pt}
\fbox{\includegraphics[scale=.43]{0.eps}} &
\setlength\fboxsep{0.5pt}
\setlength\fboxrule{0.5pt}
\fbox{\includegraphics[scale=.43]{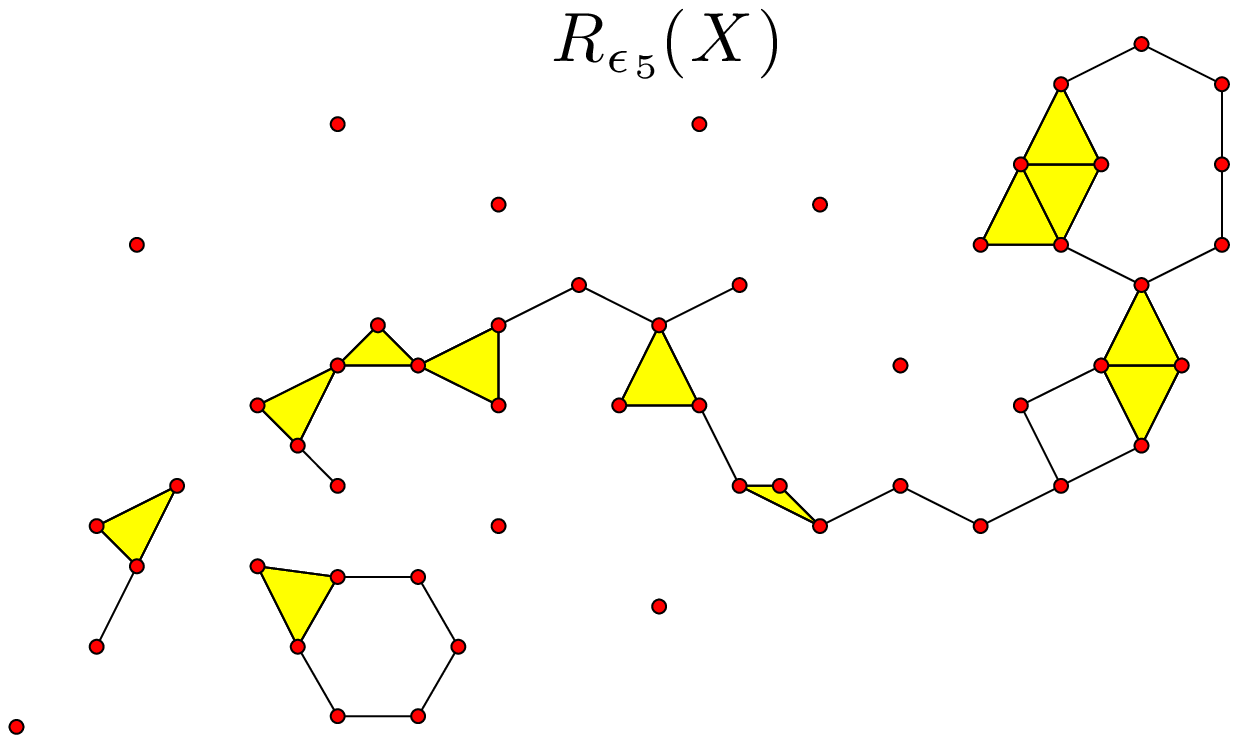}} \\
\setlength\fboxsep{0.5pt}
\setlength\fboxrule{0.5pt}
\fbox{\includegraphics[scale=.43]{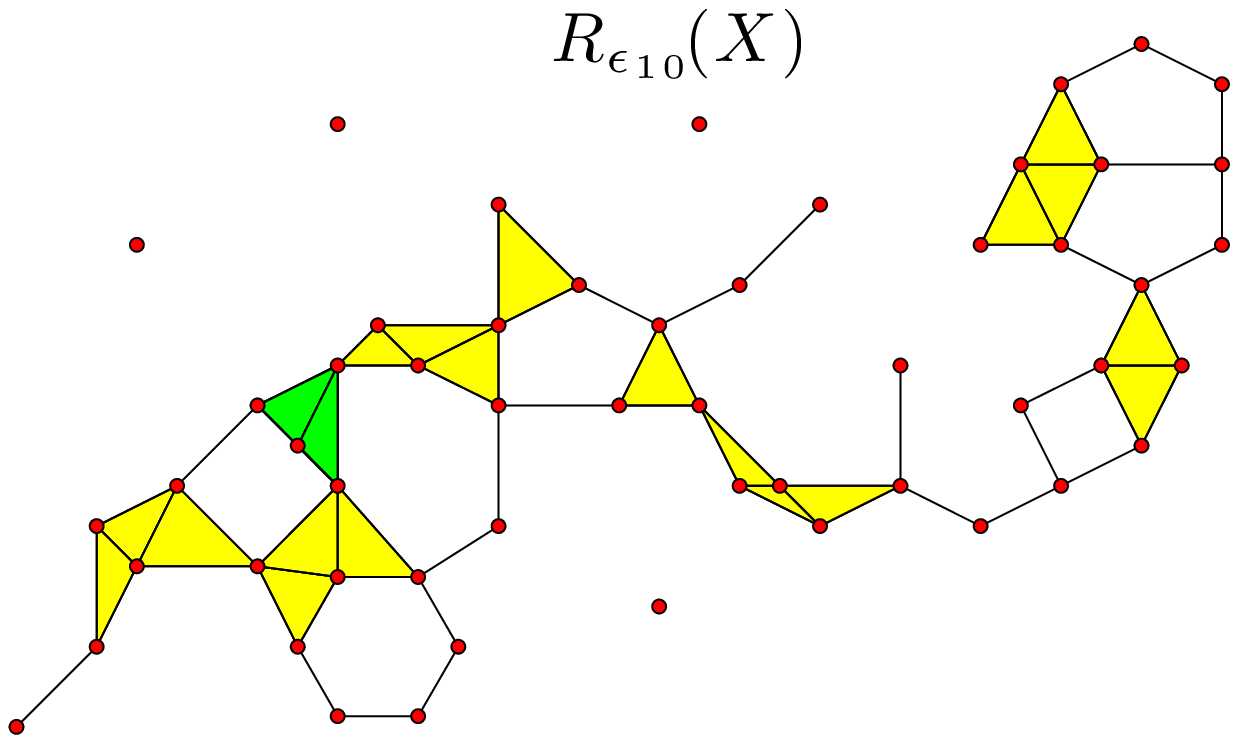}} &
\setlength\fboxsep{0.5pt}
\setlength\fboxrule{0.5pt}
\fbox{\includegraphics[scale=.43]{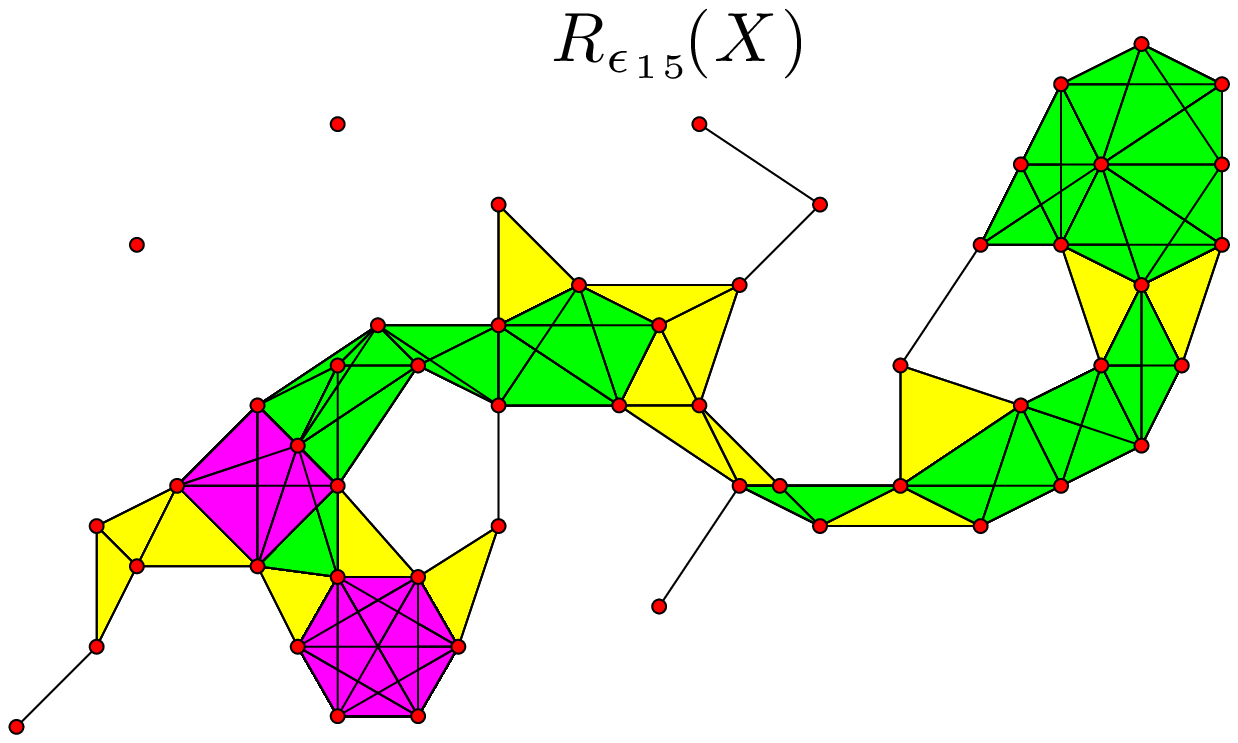}} \\
\setlength\fboxsep{0.5pt}
\setlength\fboxrule{0.5pt}
\fbox{\includegraphics[scale=.43]{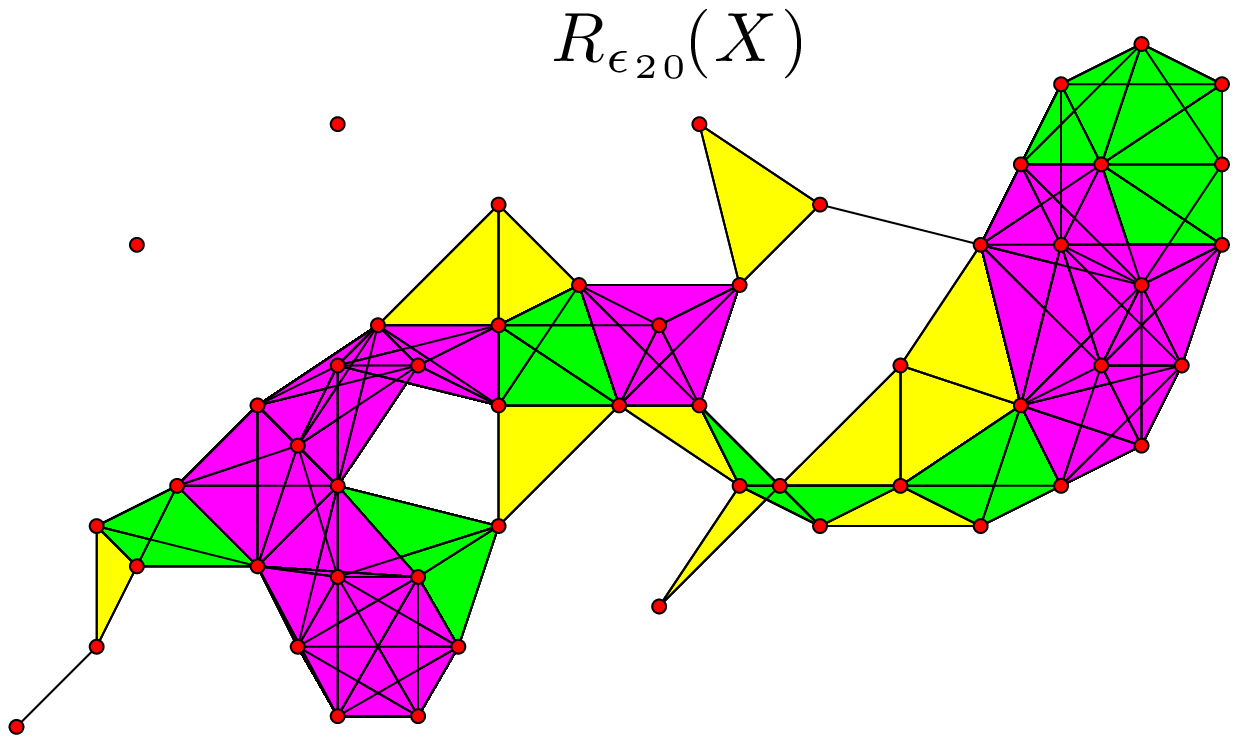}} &
\setlength\fboxsep{0.5pt}
\setlength\fboxrule{0.5pt}
\fbox{\includegraphics[scale=.43]{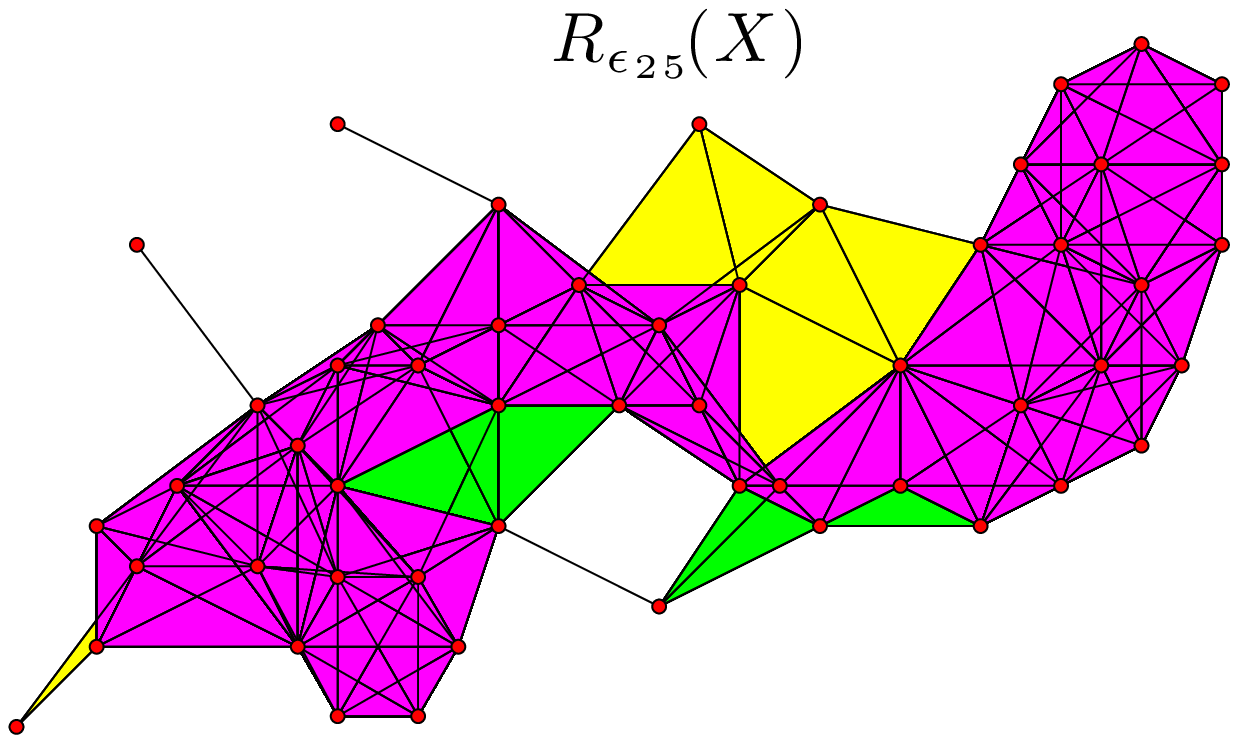}} \\
\includegraphics[scale=.7]{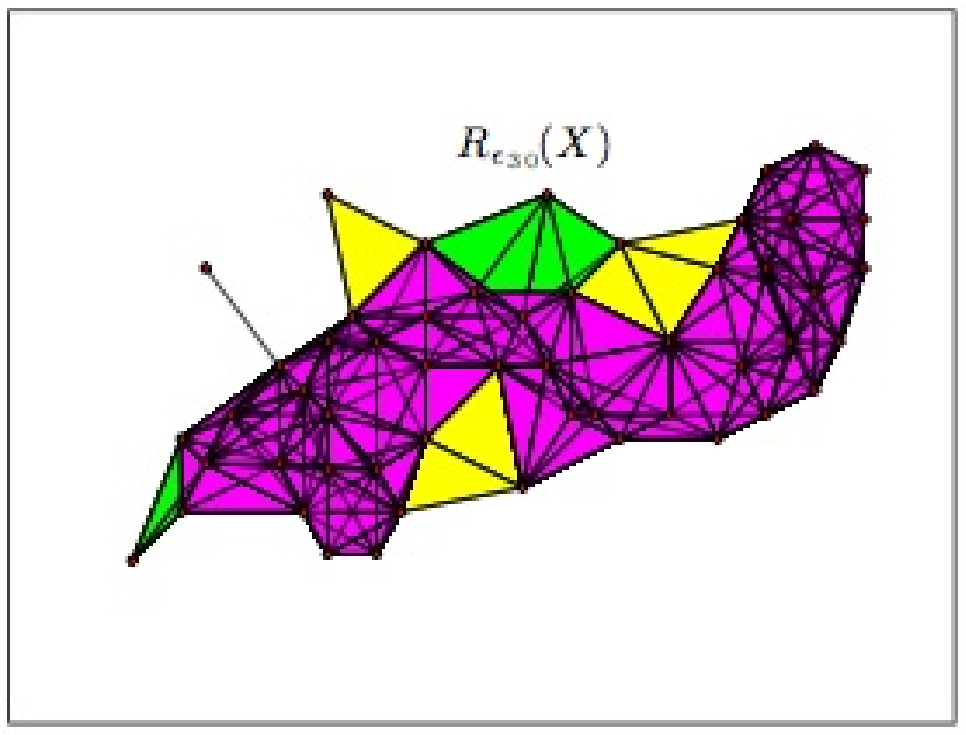} &
\includegraphics[scale=.7]{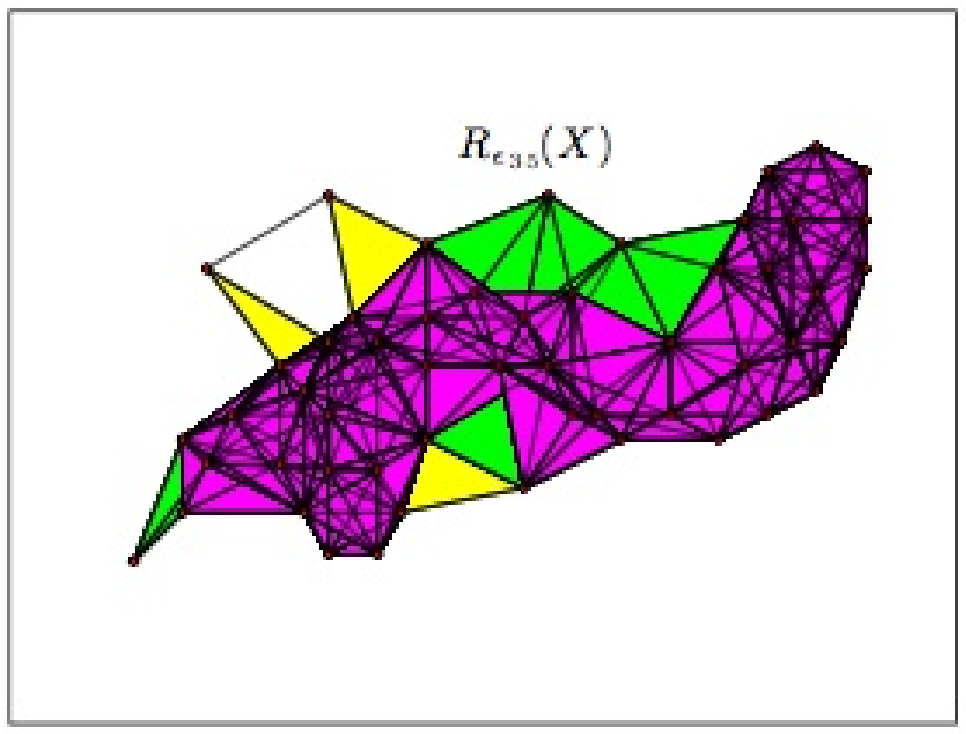} \\
\end{array}
$
\caption{Restricted filtered simplicial complex of this PCD ($m = 4$, $P = 35$).}
\end{figure}

\begin{figure}[h!]
\centering
\includegraphics[scale=.75]{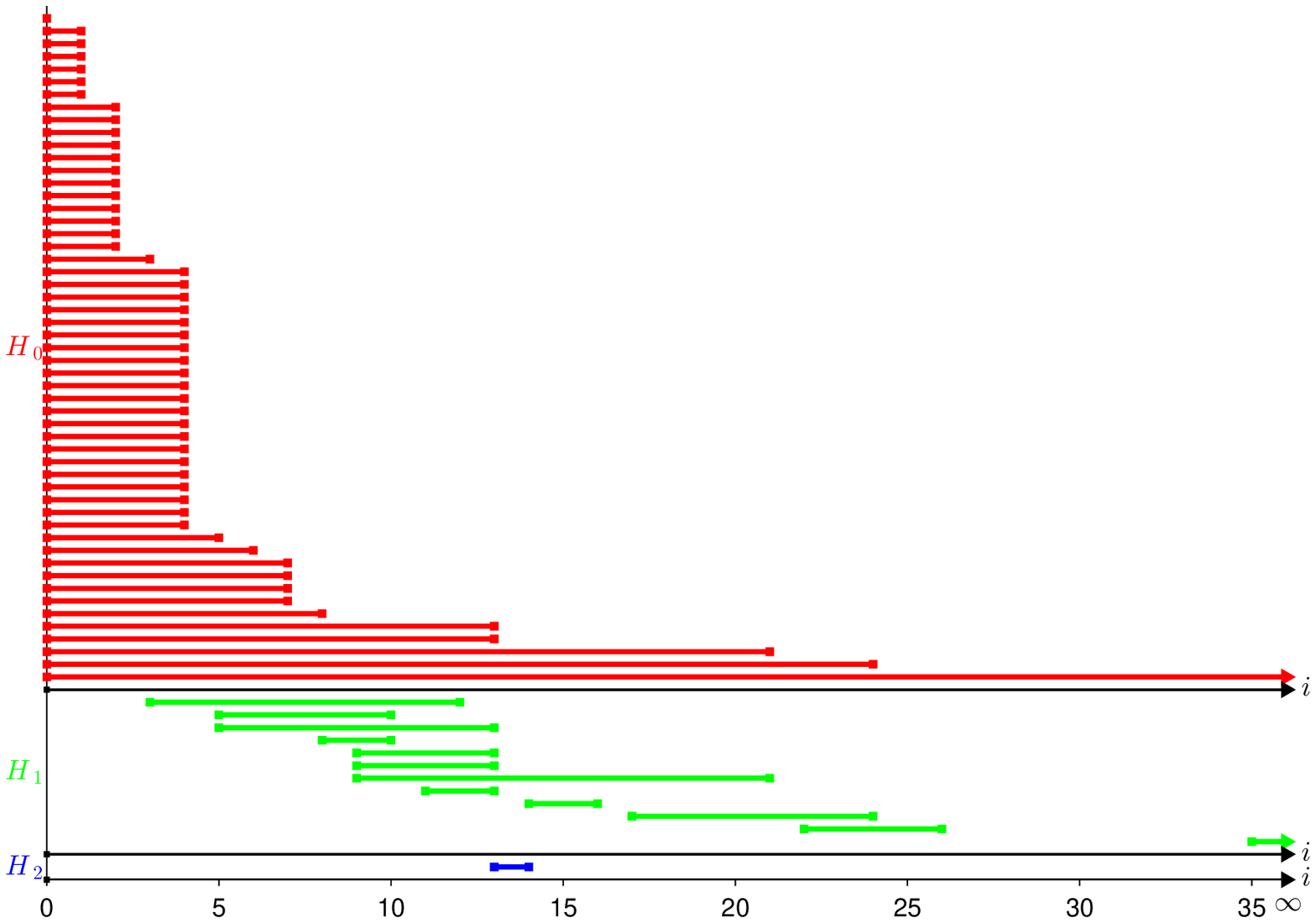}
\caption{Bar codes of the above filtered simplicial complex.}
\end{figure}

\mbox{}

\newpage

\mbox{}

\newpage

\section{Persistence Theory Refined}  \label{sec 4}

In this section we consider refinements of standard persistence for a real valued map  and describe the calculation if its invariants in the case of a linear map $f:X\to \mathbb R$,  $X$ a simplicial complex.
In order to understand the novelty of our considerations we will review in subsection \ref{SLP} the standard persistence named in this section sub-level persistence.

Let $f:X\rightarrow \mathbb{R}$ be a continuous map, denote by
$$
X_{-\infty,t}:=f^{-1}((-\infty,t])
$$
$$
X_{t,\infty}:=f^{-1}([t,\infty))
$$
$$
X_t:=f^{-1}(t) \text{ and}
$$
$$
X_{s,t}:=f^{-1}([s,t]), s\leq t
$$

\subsection{Sub-level Persistence \cite{BD}} \label{SLP}

Given a continuous map $f:X\rightarrow \mathbb{R}$, the sub-level persistent homology introduced in \cite{ELZ} and further developed in \cite{ZC}
is concerned with the following questions:

Q1. Does the class $x\in H_r(X_{-\infty,t})$ originate in $H_r(X_{-\infty,t''})$ for $t''\leq t$? Does the class
$x\in H_r(X_{-\infty,t})$ vanish in $H_r(X_{-\infty,t'})$ for $t<t'$?

Q2. What are the smallest $t'$ and $t''$ such that this happens?

The information that is contained in the linear maps $H_r(X_{-\infty,t})\rightarrow H_r(X_{-\infty,t'})$ for any $t\leq t'$
is known as \textbf{sub-level persistence}.

Recall that (sub-level){\bf persistent homology} is the collection of vector spaces and linear maps $\{ H_r(X_{-\infty,t})\rightarrow H_r(X_{-\infty,t'}), t<t', t,t'\in \mathbb R\}.$
\medskip 
Let $0\neq x\in H_r(X_{-\infty,t}))$. One says that

(i) $x$ is born at $t''$, $t''\leq t$, if $x$ is contained in img$(H_r(X_{-\infty,t''}))\rightarrow H_r(X_{-\infty,t})))$
but is not contained in img$(H_r(X_{-\infty,t''-\epsilon}))\rightarrow H_r(X_{-\infty,t})))$ for any $\epsilon>0$.

(ii) $x$ dies at $t'$, $t'>t$, if its image is zero in img$(H_r(X_{-\infty,t}))\rightarrow H_r(X_{-\infty,t'})))$
but is nonzero in  img$(H_r(X_{-\infty,t}))\rightarrow H_r(X_{-\infty,t'-\epsilon})))$ for any $0<\epsilon<t'-t$.

(iii) $x$ dies at $\infty$, if its image is always nonzero in img$(H_r(X_{-\infty,t}))\rightarrow H_r(X_{-\infty,t'})))$ for any $t'>t$.

The standard construction ``telescope'' in homotopy theory
permits to replace any finite filtered space
$K_0\subseteq K_1\subseteq \cdots\subseteq K_N$
by a weakly tame map  $f:X\rightarrow \mathbb{R}$ (cf. Corollary \ref{weaklytame}), simply by taking
$X = K_0\times [t_0,t_1] \cup_{\phi_1} K_1\times [t_1,t_2] \cup_{\phi_2}\cdots\cup_{\phi_{N-1}}K_{N-1}\times [t_{N-1},t_N]\cup_{\phi_N}K_N$
where $\phi_i: K_i\times \{t_{i+1}\}\rightarrow K_{i+1}\times \{t_{i+1}\}$ is the inclusion
and $f|_{K_i\times [t_i,t_{i+1}]}$  the projection of
$K_i\times [t_i,t_{i+1}]$ on $[t_i,t_{i+1}]$.

The sub-level persistence  for a filtered space is the sub-level persistence of the associated weakly tame map.

When $f$ is weakly tame, the sub-level persistence for each $r = 0, 1, \cdots, \dim X$ is determined by a finite collection of invariants
referred to as \textbf{bar codes for sub-level persistence} \cite{ZC}.
The $r$-bar codes for sub-level persistence of $f$ are  intervals of the form $[t, t')$ or $[t, \infty)$
with $t< t'$.

The number $\mu_r(t,t')$ of  $r$-bar codes which identify to the interval $[t, t')$ is the maximal number of linearly independent
homology classes in $H_r(X_{-\infty,t})$, which are born at $t$, die at $t'$ and remain independent in
img($H_r(X_{-\infty,t})\rightarrow H_r(X_{-\infty,s})$) for any $s$, $t\leq s<t'$.

The number $\mu_r(t,\infty)$ of $r$-bar codes which identify to the interval $[t, \infty)$ is the maximal number of linearly independent
homology classes  in $H_r(X_{-\infty,t})$ which are born at $t$, never die
and remain independent in
img($H_r(X_{-\infty,t})\rightarrow H_r(X_{-\infty,s})$) for any $s>t$.

\begin{lem} \label{lemma311}
The set of $r$-bar codes for sub-level persistence is  finite and any $r$-bar code  is an intervals of the form $[t_i, t_j)$ or $[t_i, \infty)$
with $t_i< t_j$ $t_i, t_j$  critical values of $f$.
\end{lem}

\begin{proof}
Indeed the ends of the ``bar codes'' have to be critical values.
Suppose there is an $r$-bar code of the form $[t, t')$, we have $t_0\leq t<t'\leq t_N$.
Suppose $t_i\leq t<t_{i+1}$ for some $i$ and $t_j\leq t'< t_{j+1}$ for some $j$.

There exists a class $x\in H_r(X_{-\infty,t})$ that is born at $t$ and dies at $t'$.

Since the canonical inclusion $X_{-\infty,t_i}\hookrightarrow X_{-\infty,t}$ is a deformation
retraction and induces an isomorphism on homology groups, we have
$x\in \text{img}(H_r(X_{-\infty,t_i})\rightarrow H_r(X_{-\infty,t}))$.
If $t>t_i$, it will cause a contradiction with the fact that $x$ is born at $t$.
So we have $t=t_i$ and the bar code is of the form $[t_i, t')$.

Consider the commutative diagram below
$$
\xymatrix{
X_{-\infty,t_i} \ar[r] \ar[rd] & X_{-\infty,t_j} \ar[d]^\cong \\
 & X_{-\infty,t'}
}
$$

Since the canonical inclusion $X_{-\infty,t_j}\hookrightarrow X_{-\infty,t'}$ is a deformation
retraction, the vertical map in the above diagram an isomorphism.
The image of $x$ in img$(H_r(X_{-\infty,t_i})\rightarrow H_r(X_{-\infty,t_j}))$ is zero.
If $t'>t_j$, it will cause a contradiction with the fact that $x$ dies at $t'$.
So we have $t'=t_j$ and the bar code is of the form $[t_i, t_j)$.

For the same reason, an $r$-bar code $[t, \infty)$ is of the form $[t_i, \infty)$.
The finiteness of the collection of $r$-bar codes follows from
the finiteness of $\dim H_r(X_{-\infty,t})$.

\end{proof}

From these bar codes one can derive the \textbf{Betti numbers} $\beta_r(t, t')$, the dimension of

$\text{img}(H_r(X_{-\infty,t})\rightarrow H_r(X_{-\infty,t'}))$,for any $t\leq t'$ and get the answers to questions Q1 and Q2.
For example,
\begin{equation}\label {E1}
\beta_r(t, t') = \text{ the number of } r \text{-bar codes which contain the interval } [t, t'].
\end{equation}

From the Betti numbers $\beta_r(t,t')$one can also derive these $r$-bar codes.

Denote $\mu_r(t_i,t_j)$=number of $r$-bar codes
which equal to $[t_i,t_j)$ for  $t_0\leq t_i< t_j\leq \infty$, where
$t_0$ is the smallest critical value. We have
\begin{equation}\label{E2}
\begin{array}{ll}
& \mu_r(t_i,t_j) \\ [3mm]
= & \hskip -1.5mm
 \left\{ \hskip -1.5mm
     \begin{array}{ll}
       \beta_r(t_i,t_{j-1})-\beta_r(t_{i-1},t_{j-1})-\beta_r(t_i,t_j)+\beta_r(t_{i-1},t_j), & t_0<t_i< t_j<\infty\\
       \beta_r(t_0,t_{j-1})-\beta_r(t_0,t_j), & t_i=t_0, t_0< t_j<\infty\\
       \beta_r(t_i,\infty)-\beta_r(t_{i-1},\infty), & t_0<t_i<\infty,t_j=\infty\\
       \beta_r(t_0,\infty), & t_i=t_0,t_j=\infty
     \end{array}
   \right.
\end{array}
\end{equation}

The computation of the bar codes for a filtration of simplicial\,/\,polytopal
complex is discussed in {subsection \ref{secPHBFSCC}}
with coefficient field of homology groups being  $\mathbb{Z}_2$ or $\mathbb{R}$.

\subsection{Level Persistence\cite{BD}}

Level persistence for a map $f:X\to \mathbb R$ was first considered in \cite{DW} and was better understood when the Zigzag persistence
was introduced and formulated in \cite{CSM}.
Given a continuous map $f:X\rightarrow \mathbb{R}$,
level persistence is concerned with the homology of the fibers $H_r(X_t)$ and addresses
questions of the following type.

Q1. Does the image of $x\in H_r(X_t)$ vanish in $H_r(X_{t, t'})$, where $t'>t$ or in $H_r(X_{t'',t})$, where $t''<t$?

Q2. Can $x$ be detected in $H_r(X_{t'})$ where $t'>t$ or in $H_r(X_{t''})$ where $t''<t$? The precise meaning of
detection is explained below.

Q3. What are the smallest $t'$ and $t''$ for the answers to Q1 and Q2 to be affirmative?

To answer such questions one has to record information about the following maps:
$$
H_r(X_t)\rightarrow H_r(X_{t, t'})\leftarrow H_r(X_{t'}).
$$
The \textbf{level persistence} is the information provided by this collection of vector spaces and linear maps  considered
above for all $t$, $t'$.

Let $0\neq c\in H_r(X_t)$. One says that

(i) $c$ dies downward at $t'< t$, if its image is zero in img$(H_r(X_t)\rightarrow H_r(X_{t', t}))$
but is nonzero in img$(H_r(X_t)\rightarrow H_r(X_{t'+\epsilon,t}))$ for any $0<\epsilon<t-t'$.

(ii) $c$ dies upward at $t''>t$, if its image is zero in img$(H_r(X_t)\rightarrow H_r(X_{t, t''}))$
but is nonzero in img$(H_r(X_t)\rightarrow H_r(X_{t, t''-\epsilon}))$ for any $0<\epsilon<t''-t$.

We say that $x\in H_r(X_t)$ can be detected at $t'\geq t$, if its image in $H_r(X_{t,t'})$ is nonzero and is
contained in the image of $H_r(X_{t'})\rightarrow H_r(X_{t,t'})$.
Similarly, the detection of $x$ can be defined for $t''<t$ also.

\begin{defn} \label{tamemap}

A continuous map $f:X\rightarrow  \mathbb{R}$  is called \textbf{tame} (Definition 3.5, \cite{BDD}) if
$X$ is compact and there exists finitely many values $\min(f(X)) = t_0<t_1<\cdots<t_N = \max(f(X))$ (so called critical values) so that:

(i) for any $t\neq t_0,t_1,\cdots,t_N$ there exists $\epsilon>0$ so that $f:f^{-1}(t-\epsilon,t+\epsilon)\rightarrow(t-\epsilon,t+\epsilon)$
and the second factor projection $X_t\times (t-\epsilon,t+\epsilon)\rightarrow(t-\epsilon,t+\epsilon)$ are fiber-wise homotopy equivalent.

(ii) for any $t_i$ there exists $\epsilon>0$ so that canonical inclusions $X_{t_i}\hookrightarrow  X_{t_i,t_i+\epsilon}$ and
$X_{t_i}\hookrightarrow  X_{t_i-\epsilon,t_i}$ are deformation retractions.

\end{defn}

This means that the homology group of the level set changes at finitely many $t$'s. We have

\begin{lem}\label{lemma321}
Given $t_{i-1}<s<t_i$, $X_{t_{i-1},s}$ and $X_{s,t_i}$ retracts
by deformation onto $X_{t_{i-1}}$ and $X_{t_i}$ respectively.
\end{lem}

\begin{proof}

We will prove  that $X_{s,t_i}$  retracts by deformation onto $X_{t_{i-1}}$.
The proof that $X_{s,t_i}$ retracts by deformation onto $X_{t_{i}}$
is essentially the same.

For any $a$, $t_{i-1}<a\leq s$, there exists an $\epsilon_a>0$ such that
$f:f^{-1}(a-\epsilon_a,a+\epsilon_a)\rightarrow(a-\epsilon_a,a+\epsilon_a)$
and the second factor projection $X_t\times (a-\epsilon_a,a+\epsilon_a)\rightarrow(a-\epsilon_a,a+\epsilon_a)$
 are fiberwise homotopy equivalent.
Hence for any $a-\epsilon_a<t''<t'<a+\epsilon_a$, we have $X_{t'',t'}$ retract by deformation onto $X_{t''}$.
For $a$, $t_{i-1}<a<s$ we can choose $0<\epsilon_a<min(a-t_{i-1},s-a)$,
so that $(a-\epsilon_a,a+\epsilon_a)$ is contained in $(t_{i-1},s)$.

For $t_{i-1}$, there exists an $\epsilon_{t_{i-1}}>0$ so that canonical inclusion
$X_{t_i}\hookrightarrow  X_{t_i,t_i+\epsilon_{t_{i-1}}}$  is a deformation retraction.

The union of open intervals $(a-\epsilon_a, a+\epsilon_a)$ for $t_{i-1}\leq a\leq s$ covers
the closed interval $[t_{i-1}, s]$.
Then  finitely many of them $(a_k, b_k)$, $1\leq k\leq m$  cover $[t_{i-1}, s]$.
Suppose  these intervals are not contained  one  in another, otherwise we can delete it.
Observe that $(t_{i-1}-\epsilon_{t_{i-1}}, t_{i-1}+\epsilon_{t_{i-1}})$ and
$(s-\epsilon_s, s+\epsilon_s)$ have to be contained in those finite collection
of intervals.

Order $(a_k, b_k)$, $1\leq k\leq m$ according to their right endpoints increasingly,
we have $b_j<b_k$, when $j<k$.
Observe that $b_{k-1}\in (a_k,b_k)$ and $a_k\in (a_{k-1},b_{k-1})$.
Therefore we also have $a_j<a_k$, when $j<k$ and
the intersection of $(a_{k-1},b_{k-1})$ and $(a_k,b_k)$ is always nonempty.
Since both left endpoints and right endpoints are ordered increasingly
 $(a_1, b_1) =  (t_{i-1}-\epsilon_{t_{i-1}}, t_{i-1}+\epsilon_{t_{i-1}})$ and
$(a_m,b_m) = (s-\epsilon_s, s+\epsilon_s)$.

Choose any $c_k\in (a_k,b_k)\cap (a_{k+1},b_{k+1})$, $2\leq k\leq m-1$.
Let $c_0 = t_{i-1}$, $c_1 = t_{i-1} + \epsilon_{t_{i-1}}$ and $c_m = s$.

We have $X_{c_{k-1},c_k}$ retracts by deformation onto $X_{c_{k-1}}$
for any $1\leq k \leq m$.

Hence $X_{t_{i-1},s} = X_{c_0, c_m}$ retracts by deformation onto  $X_{t_{i-1}} = X_{c_0}$.

\end{proof}

\begin{cor} \label{tametoweaklytame}

Tame maps are weakly tame.

\end{cor}

Note that the tame maps form a generic set of maps in the set of all  continuous maps on a simplicial complex or on a smooth manifold (and more general on a compact ANR).

\begin{lem}\label{lemma323}
Given $t_{i-1}<s<t<t_i$, $X_{s,t}$ retracts by deformation
onto $X_s$ or $X_t$.
\end{lem}

In case of a tame map the collection of the vector spaces and linear maps is determined up to coherent isomorphisms
by a collection of invariants called \textbf{bar codes for level persistence} which are intervals of the form
$[t, t']$ with $t\leq t'$ and  $(t, t')$, $(t, t']$, $[t, t')$ with $t< t'$.

These bar codes are called invariants because two tame maps $f: X\rightarrow \mathbb{R}$ and $g: Y\rightarrow \mathbb{R}$
which are fiber-wise homotopy equivalent have the same associated bar codes.
The above result was established for the Zigzag persistence from which can be derived, but can also be proven directly. However the proof is not contained in this paper.

An open end of an interval signifies the death of a homology class at that end (left or right) whereas a closed end signifies
that a homology class cannot be detected beyond this level (left or right).

There exists an $r$-bar code $(t'', t')$ if there exists a class $x\in H_r(X_{t})$ for some $t''<t<t'$
which is detectable for $t''< s< t'$ and dies at $t''$ and $t'$.
The multiplicity of $(t'', t')$ is the maximal number of linearly independent classes
in $H_r(X_{t})$ such that:

(i) all remain linearly independent in img($H_r(X_{t})\rightarrow H_r(X_{t,s})$) for $t\leq s< t'$
and  img($H_r(X_{t})\rightarrow H_r(X_{s,t})$) for $t''< s\leq  t$;

(ii) all die at $t''$ and $t'$.

Notice that the change of $t$ above will not affect the multiplicity of $(t'',t')$.

There exists an $r$-bar code $(t'', t']$ if there exists an element $x\in H_r(X_{t'})$
which is not detectable for $s>t'$ and detectable for $t''< s\leq t'$ and dies at $t''$.
The multiplicity of $(t'', t']$ is the maximal number of linearly independent elements
in $H_r(X_{t'})$ such that

(i) neither one is detectable for $s>t'$;

(ii) all remain linearly independent in img($H_r(X_{t'})\rightarrow H_r(X_{s,t'})$) for $t''< s\leq t'$;

(iii) all dies at $t''$.

There exists an $r$-bar code $[t'', t')$ if there exists an element $x\in H_r(X_{t''})$
which is not detectable for $s<t''$ and detectable for $t''\leq s<t'$ and dies at $t'$.
The multiplicity of $[t'', t')$ is the maximal number of linearly independent elements
in $H_r(X_{t''})$ such that

(i) neither one is detectable for $s<t''$;

(ii) all remain linearly independent in img($H_r(X_{t''})\rightarrow H_r(X_{t'',s})$) for $t''\leq s<t'$;

(iii) all dies at $t'$.

There exists an $r$-bar code $[t'', t']$ if there exists an element $x\in H_r(X_{t''})$
which is not detectable for $s<t''$ or $s>t'$ and detectable for $t''\leq s\leq t'$.
The multiplicity of $[t'', t']$ is the maximal number of linearly independent elements
in $H_r(X_{t''})$ such that

(i) neither one is detectable for $s<t''$ or $s>t'$;

(ii) all remain linearly independent in img($H_r(X_{t''})\rightarrow H_r(X_{t'',s})$) for $t''\leq s\leq t'$.

\begin{lem} \label{lemma324}
For a tame map, the set of $r$-bar codes for level persistence is finite. Any  $r$-bar code  is an  interval of the form
$[t_i, t_j]$ with $t_i\leq t_j$ critical values or   $(t_i, t_j)$, $(t_i, t_j]$, $[t_i, t_j)$ with $t_i< t_j,$ $t_i, t_j$ critical values.
\end{lem}

\begin{proof}
Using {Lemma \ref{lemma321}} and {Lemma \ref{lemma323}},
the proof is similar to {Lemma \ref{lemma311}}.

\end{proof}

\begin{nota}
Given a tame map $f:X\rightarrow \mathbb{R}$ with critical values $t_0<\cdots<t_N$,
denote by
$$
\begin{array}{l}
\,\, BL_r(f) := \text{the number of all $r$-bar codes for level persistence (with respect to \emph{r-}th homology groups)}.\\
N_r(t_i,t_j) := \text{the number of intervals } (t_i, t_j) \text{ in } BL_r(f). \\
N_r(t_i,t_j] := \text{the number of intervals } (t_i, t_j] \text{ in } BL_r(f). \\
N_r[t_i,t_j) := \text{the number of intervals } [t_i, t_j) \text{ in } BL_r(f). \\
N_r[t_i,t_j] := \text{the number of intervals } [t_i, t_j] \text{ in } BL_r(f). \\
\end{array}
$$
Hence $\sharp BL_r(f) = N_r(t_i,t_j) + N_r(t_i,t_j] + N_r[t_i,t_j) + N_r[t_i,t_j]$.

\end{nota}

In {Figure \ref{figure421}}, we indicate the bar codes both for sub-level and level persistence for some simple map in order to illustrate their
differences and connections. The class consisting of the sum of two circles at level $t$ is not detected on the right,
but is detected at all levels on the left up to(but not including) the level $t'$.

Level persistence provides considerably more information than the sub-level persistence \cite{BDD}
and the bar codes for the sub-level persistence can be recovered from the bar codes for the level
persistence.
An $r$-bar code $[s, t)$ for level persistence contributes an $r$-bar code $[s,t)$ for subl evel persistence.
An $r$-bar code $[s, t]$ for level persistence contributes an $r$-bar code $[s,\infty)$ for sub-level persistence.
$r$-bar codes $(s,t]$ and $(s,t)$ for level persistence contribute nothing to $r$-bar codes for sub-level persistence.
An $r$-bar codes $(s,t)$  for level persistence contributes an $r+1$-bar code $[t,\infty)$  for sub-level persistence.
See {Figure \ref{figure421}} and {Lemma \ref{levsublev}} below.

\begin{figure}[h!]

\includegraphics[scale=.5]{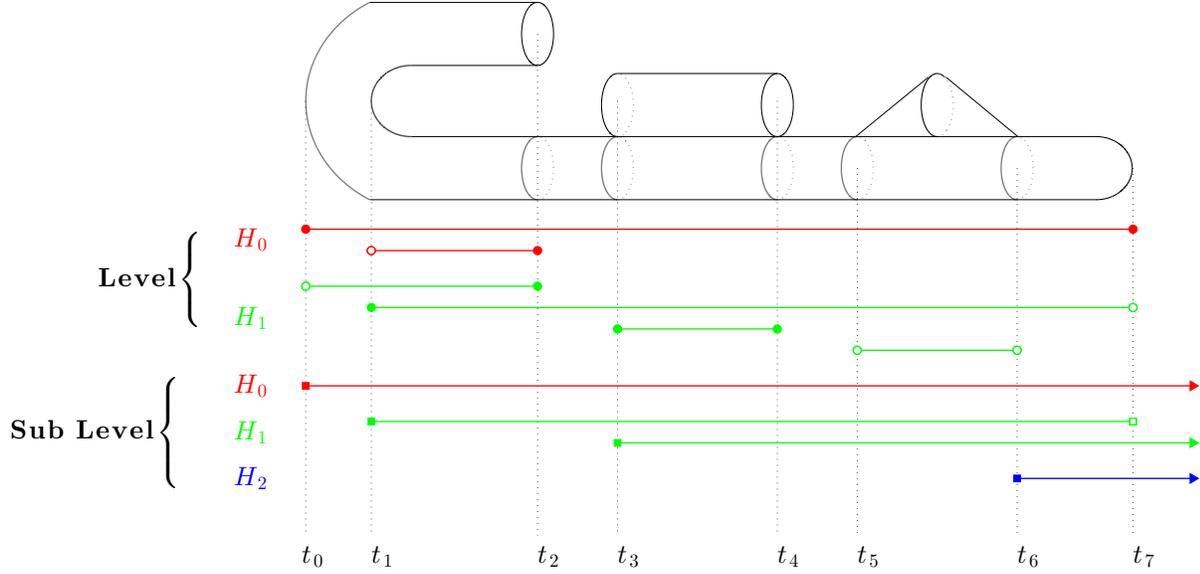}

\caption{Bar codes for level and sub-level persistence.} \label {figure421}

\end{figure}

\begin{lem} \label{levsublev}
Given a tame map $f: X\rightarrow \mathbb{R}$ with critical values $t_0<t_1<\cdots<t_N$.
We have
$$
\mu_r(t_i,t_j) = N_r[t_i, t_j)
$$
$$
\mu_r(t_i,\infty) = \sum_{l=i}^N N_r[t_i, t_l] + \sum_{l=0}^{i-1} N_{r-1}(t_l,t_i)
$$
for any critical values $t_i<t_j$.
\end{lem}

\begin{proof} 
Item 1 follows from formulas (\ref{E1}) and (\ref{E2}). 

Item 2 is more elaborate. One uses formula (\ref{E2}) which calculates $\mu_r(t_i, \infty)$ as $\mu_r(t_i, \infty)= \beta_r(t_i, \infty) - \beta_r(t_{i-1}, \infty)$.  A calculation of $\beta_r(t_i, \infty)$ can be recovered from Corollary 3.4 in \cite{BD} which implies that 
this number is exactly the number of $(r-1)$-bar codes of the form $(t_l, t_i), l= 0,1, \cdots, i-1$ plus the number of $r$-bar codes of the form $[a,b]$ with $a\leq t_i$.
Clearly $a, b$ should be critical values. A different derivation can be achieve independently of \cite {BD}.
\end{proof}

The bar codes for the level persistence can be also recovered from the bar codes for the sub-level persistence but from the bar codes of a collections of tame maps canonically associated to $f$.
This will be described in the next subsection.

For this purpose one uses an alternative but equivalent way to describe the level persistence based on a different collection of numbers,  referred below as relevant persistence numbers, $l_r, l_r^+, l^-_r, e_r, i_r.$ 

\begin{defn}
For a continuous map $f: X\rightarrow \mathbb{R}$ and $t''\leq t\leq t'$,
let $L_r(t):=H_r(X_t)$, $L_r^+(t;t'):=\ker(H_r(X_t)\rightarrow H_r(X_{t, t'}))$,
$L_r^-(t;t''):=\ker(H_r(X_t)\rightarrow H_r(X_{t'',t}))$
and $I_r(t, t'):= \text{img}(H_r(X_t)\rightarrow H_r(X_{t,t'})) \cap \text{img}(H_r(X_{t'})\rightarrow H_r(X_{t,t'})) $.

Define the \textbf{relevant level persistent numbers}
$$
l_r(t):= \dim L_r(t)
$$
$$
l_r^+(t;t'):=\dim L_r^+(t;t')
$$
$$
l_r^-(t;t''):=\dim L_r^-(t;t'')
$$
$$
e_r(t;t',t''):=\dim(L_r^+(t;t')\cap L_r^-(t;t''))
$$
$$
i_r(t,t'):=\dim(I_r(t,t'))
$$

\end{defn}

The relation between these collections of numbers is illustrated in the diagram below. 

$$
\xymatrix
{
	*+[F]\txt{$i_r(t, t')$} \ar@{=>}[rr]^{\hskip -1cm Thm\; 4.2} &
	 &
	 *+[F]\txt{$ l_r(t), l_r^+(t;t')$\\$l_r^-(t;t''), e_r(t;t',t'')$} \ar@{=>}[rr]^{\hskip 0.5cm Thm\; 4.3} &
	 &
	*+[F]\txt{$N_r([t, t'])$\\ $N_r((t, t'))$\\ $N_r((t, t'])$\\ $N_r([t, t'))$} \ar@/_-4pc/[ll]^{\hskip -.8cm Observation\; 4.1} \ar@/_-8pc/[llll]^{Observation\; 4.1}
}
$$

The first four have geometric meaning the last ones (the fifth) are more technical. However the first four $l_r, l_r^+, l^-_r, e_r$ can be derived from the last ones $i_r$ 

One can derive all the numbers $l_r, l_r^+, l^-_r, e_r, i_r$  from the number of bar codes $N_r(t_i, t_j)$, $N_r(t_i, t_j]$, $N_r[t_i, t_j)$, $N_r[t_i, t_j]$. 

\begin{obs}
For a tame map
we can derive relevant level persistent numbers
from the numbers $N's$ of bar codes for level persistence.
\end{obs}

\begin{proof}

For $t''\leq t\leq t'$
\begin{enumerate}
\item
$l_r(t)$ = number of intervals in $BL_r(f)$ which contain $t$;

\item
$i_r(t,t')$ = number of intervals in $BL_r(f)$ which contain $[t, t']$;

\item 
\begin{equation} \label{l+}
l_r^+(t; t') = \sum_{t_i\leq t< t_j\leq t'} N_r[t_i,t_j) + \sum_{t_i< t<t_j\leq t'} N_r(t_i,t_j);
\end{equation}
\item 
\begin{equation} \label{l-}
l_r^-(t; t'') = \sum_{t''\leq t_i< t\leq t_j} N_r(t_i,t_j] + \sum_{t''\leq t_i< t< t_j}  N_r(t_i,t_j);
\end{equation}
\item
\begin{equation} \label{e_r}
e_r(t;t',t'') = \sum_{ t''\leq t_i<t<t_j\leq t'} N_r(t_i,t_j).
\end{equation}
\end{enumerate}
\end{proof}

 
\begin{thm} \label{theorem326}

For a tame map the numbers $i_r(t,t')$
determine the numbers  $l_r(t)$, $l_r^+(t;t')$, $l_r^-(t;t'')$
 and  $e_r(t;t',t'').$ 
\end{thm}

\begin{proof}

$$
l_r(t) = i_r(t,t).
$$

$$
N_r(t_k, t_j) = i_r(t'',t') - i_r(t_k, t') - i_r(t'', t_j) + i_r(t_k, t_j)
$$
for any $t''\leq t'$ such that $t_k<t''<t_{k+1}$, $t_{j-1}<t'<t_j$.

$$
N_r[t_k, t_j] = i_r(t_k, t_j) - i_r(t'', t_j) - i_r(t_k, t') + i_r(t'', t')
$$
for any $t_{k-1}<t''<t_k$ and $t_j<t'<t_{j+1}$.

$$
N_r(t_k,t_j] = i_r(t'',t_j) - i_r(t_k,t_j) - i_r(t'',t') + i_r(t_k,t')
$$
for any $t_k<t''<t_{k+1}$ and $t_j<t'<t_{j+1}$.

$$
N_r[t_k,t_j) = i_r(t_k,t') - i_r(t_k,t_j) - i_r(t'',t') + i_r(t'',t_j)
$$
for any $t_{k-1}<t''<t_k$ and $t_{j-1}<t'<t_j$.

Plug in equation (\ref{l+}), (\ref{l-}) and (\ref{e_r})
we get $l_r^+(t;t')$, $l_r^-(t;t'')$ and  $e_r(t;t',t'')$.

\end{proof}

\begin{thm} \label{theorem327}
For a tame map
we can get the $r$-bar codes for level persistence from
relevant persistent numbers.
\end{thm}

\begin{proof}
First the numbers $ N_r(t_k, t_j)$ can be calculated by the formula.

\begin{equation}
N_r(t_k, t_j) = e_r(t; t_j,t_k) - e_r(t; t_j,t_{k+1})
- e_r(t; t_{j-1},t_k) + e_r(t; t_{j-1},t_{k+1}),
\end{equation}
for any $t_k<t<t_j$.

To determine the numbers $N_r(t_i,t_j]$, $N_r[t_i,t_j)$ and $N_r[t_i,t_j]$,
we introduce the following auxiliary numbers 
$$
\begin{array}{l}
n_r\{t_i,t_j\} := \text{the number of intervals in } BL_r(f) \text{ which intersect the levels } X_{t_i} \text{ and } X_{t_j};\\
n_r\{t_i,t_j) := \text{the number of intervals in } BL_r(f) \text{ which intersect the level } X_{t_i} \\
\text{ with open end at } t_j;\\
n_r\{t_i,t_j] := \text{ the number of intervals in } BL_r(f) \text{ which intersect the level } X_{t_i} \\
\text{ with closed end at } t_j;\\
n_r(t_i,t_j\} := \text{ the number of intervals in } BL_r(f) \text{ which intersect the level } X_{t_j} \\
\text{ with open end at } t_i;\\
n_r[t_i,t_j\} := \text{ the number of intervals in } BL_r(f) \text{ which intersect the level } X_{t_j} \\
\text{ with closed end at } t_i.
\end{array}
$$

The numbers $ n_r\{t_i,t_j),$  $n_r\{t_i,t_j),$   $n_r\{t_i,t_j\}$ and $n_r[t_i,t_j\}$ can be derived from the
 relevant persistent numbers as indicated below:

$$
\begin{array}{rcl}
n_r\{t_i,t_j) &=& l_r^+(t_i;t_j) - l_r^+(t_i;t_{j-1}) \\
n_r(t_i,t_j\} &=& l_r^-(t_j;t_i) - l_r^-(t_j;t_{i+1}) \\
n_r\{t_i, t_j\} &=& i_r(t_i,t_j) \\
n_r[t_i,t_j\} &=& n_r\{t_i,t_j\} - n_r\{t_{i-1},t_j\} - n_r(t_{i-1},t_j\}
\end{array}
$$

With their help one derive 

\begin{equation}
N_r(t_i,t_j] = n_r(t_i,t_j\} - n_r(t_i,t_{j+1}\} - N_r(t_i,t_{j+1})
\end{equation}
\begin{equation}
N_r[t_i,t_j) = n_r\{t_i,t_j) - n_r\{t_{i-1},t_j) - N_r(t_{i-1},t_j)
\end{equation}
\begin{equation}
N_r[t_i,t_j] = n_r[t_i,t_j\} - n_r[t_i,t_{j+1}\} - N_r[t_i,t_{j+1})
\end{equation}

\end{proof}

\vskip .5cm

In order to calculate the relevant persistence numbers $l_r(t)$, $l_r^+(t;t')$, $l_r^-(t;t'')$
and $e_r(t;t',t'')$ we induce so called
\textbf{positive and negative bar codes} in {subsection \ref{defpnbarcode}}.
Positive and negative bar codes are defined by sub-level
persistence, so they can be calculated with the
methods as described in {subsection \ref{secPHBFSCC}}.

We can also calculate $i_r(t,t')$ with similar method.
Details will be given in {subsection \ref{frm}}.

Once we get all relevant numbers,
we can  get the bar codes for the level persistence providing an alternative to the Carson deSilva algorithm cf. \cite {CS}  to calculate the level persistence bar codes as bar codes for Zigzag persistence.

\subsection{The General Framework for the Calculation of Bar Codes
for Level Persistence of Generic $\mathbb{R}$-valued Linear Maps
 on Simplicial Complexes} \label{frm}

\begin{defn}(Generic $\mathbb{R}$-valued Linear Maps on Simplicial Complexes)

Let $X$ be a simplicial complex and $|X|$ its spatial realization. An \emph{$\mathbb{R}$-valued linear map}
on $X$ is a continuous map $f:|X|\rightarrow \mathbb{R}$ whose restriction to
each simplex is linear. Let $X_0$ be the set of vertices
of $X$. $f$ is called \emph{generic} if $f:X_0\rightarrow \mathbb{R}$ is injective.
(See page 2, \cite{BDD})

\end{defn}

\begin{thm}

Generic $\mathbb{R}$-valued linear maps on simplicial complexes are tame. \footnote{Any simplicial maps are tame, the arguments although similar are more elaborate to write down in the general case.}

\end{thm}

\begin{proof}

Let $f$ be a generic $\mathbb{R}$-valued linear map on a
simplicial complex $X$ with $N+1$ vertices.
Let $t_0<\cdots<t_N$ be the the values of $f$ on the vertices of $X$.
Label vertices of $X$  such that $f(x_i)=t_i$, $0\leq i\leq N$.

(i) For $s<t$ with $t_{i-1}<s, t<t_i$ one produces a canonical homeomorphism $\theta_{s,t}: X_{s,t}\to X_t\times [s,t]$ which intertwines the map $f$ with the second factor projection of $X_t\times [s,t]$ on $[s,t]$.

Given $\sigma $ with $\sigma\cap X_{s,t}\ne \phi$.
Each point $x\in X_{s,t}\cap \sigma $ lies on a unique straight line  segment $l(x)$ which
is the trajectory of $\nabla f|_\sigma$ w.r.t. the flat metric on $\sigma$.
Denote by $L(l)$(resp. $R(l)$)  the intersection of this line with $X_s$(resp. $X_t$) and by $p_{s,t}(x):=L(l(x))$(resp. $q_{s,t}(x):=R(l(x))$) the obviously continuous map $X_{s,t}\cap \sigma \to X_s \cap \sigma$(resp. $X_{s,t}\cap \sigma \to X_t \cap \sigma$).
Consider $\tilde p_{s,t}: X_{s,t}\cap \sigma \to X_s \cap \sigma \times [s,t]$(resp. $\tilde q_{s,t}: X_{s,t}\cap \sigma \to X_t \cap \sigma \times [s,t]$)
the map defined by $\tilde p_{s,t}(x)= (p_{s,t}(x), f(x))$(resp. $\tilde q_{s,t}(x)= (q_{s,t}(x), f(x))$).
Clearly $\tilde p_{s,t}$ and $\tilde q_{s,t}$ are homeomorphisms which intertwines the map $f$ with the second factor projection of $X_t\times [s,t]$ on $[s,t]$.

For any $\sigma_1$ and $\sigma_2$, $\tilde p_{s,t}$ and $\tilde q_{s,t}$ agree on the
intersection $\sigma_1\cap \sigma_2$.
So we can define $\theta_{s,t}|_\sigma $, the restriction of $\theta_{s,t}$ on each $\sigma$, as $\tilde q_{s,t}$. In this way, one produces a canonical homeomorphism $\theta_{s,t}: X_{s,t}\to X_t\times [s,t]$ which intertwines the map $f$ with the second factor projection of $X_t\times [s,t]$ on $[s,t]$.

Therefore, for any $t\neq t_0,t_1,\cdots,t_N$, there exists $\epsilon>0$ so that $f:f^{-1}(t-\epsilon,t+\epsilon)\rightarrow(t-\epsilon,t+\epsilon)$
and the second factor projection $X_t\times (t-\epsilon,t+\epsilon)\rightarrow(t-\epsilon,t+\epsilon)$ are fiber-wise homotopy equivalent.

(ii) For $s<t$ with $s=t_{i-1}, t<t_i$ and $\sigma $ with $\sigma\cap X_{t_{i-1},t}\ne \phi$,
 it is easy to see that $X_{t_{i-1},t}\cap \sigma $ is the mapping cylinder of a canonical map $X_t\cap \sigma \to X_{t_{i-1}}\cap \sigma$  which prove that $X_{t_{i-1},t}\cap \sigma$ retracts by deformation on $X_{t_{i-1}}\cap \sigma$.
For any $\sigma_1$ and $\sigma_2$, the above deformation retractions agree on the
intersection $\sigma_1\cap \sigma_2$.
So $X_{t_{i-1},t}$ retracts by deformation on $X_{t_{i-1}}$.

For $s<t$ with $t_{i-1}<s, t=t_i$, we can prove $X_{s,t_i}$ retracts by deformation on $X_{t_i}$
in the same way.

\end{proof}

Let $f$ be a generic $\mathbb{R}$-valued linear map on a
simplicial complex $X$ with $N+1$ vertices.
Let $t_0<\cdots<t_N$ be the the values of $f$ on the vertices of $X$.
Since $f$ is tame, we can define bar codes for level persistence of $f:|X|\rightarrow \mathbb{R}$.
By {Lemma \ref{lemma324}}, the bar codes for level persistence of $f$ are intervals of the form
$[t_i, t_j]$ with $t_i\leq t_j$ and  $(t_i, t_j)$, $(t_i, t_j]$, $[t_i, t_j)$ with $t_i< t_j$. 

\begin{lem}

For the map $f:|X|\rightarrow \mathbb{R}$ defined above, the bar codes of the form $(t_i, t_{i+1})$
do not exist.

\end{lem}

Let $f$ and $X$ be defined as above.
Let $M_X$ be the boundary matrix of $X$ which records the boundary map
$\partial:X\rightarrow X$.
Start with $M_X$ and the values on vertices $t_0<t_1<\cdots<t_N$,
we can calculate the bar codes for level persistence with two methods:

\textit{Method 1}, which involves only $i_r(t,t')$, and

\textit{Method 2}, which  involves $l_r^+(t;t')$,
$l_r^-(t;t'')$, $e_r(t;t'',t')$ and $i_r(t,t')$.

Finally one use the pervious results to determine the bar codes.
It appears that {Method 1} is simpler, however we will see later that {Method 2}
is about 8 times faster than {Method 1}.

\begin{defn}

Given any $t$, we get a polytopal complex
$$
X_t := \{\; |\sigma| \cap |X|_t \;\mid\; \sigma\in X \;\}.
$$

Given any $s<t$, we get aother polytopal complex
$$
X_{s,t} := X_s \cup X_t \cup \{\; |\sigma| \cap |X|_{s,t} \;\mid\; \sigma\in X \;\}.
$$

For $s=t$, observe that $X_{s,t}$ is $X_t$.

\end{defn}

\vskip .2in

{\large \textbf{Method 1}} (with coefficients in $\mathbb Z_2$)
\medskip 
Given $s<t$, order cells in $X_{s,t}$ in the following way:

(i) if $\tau \in X_s$, $\tau' \in X_t$ and $\tau'' \in X_{s,t}$,
then $\tau \prec \tau' \prec  \tau'' $;

(ii) if $\dim(\tau)<\dim(\tau')$, then $\tau \prec \tau'$.

The ordering of the vertices provides orientation on simplices
and then on the cells obtained as ``cuts''.
We can write a matrix $M_{X_{s,t}}$ for the boundary
map $\partial_{s,t}: X_{s,t} \rightarrow X_{s,t}$
with respect to the order of cells above.

From the reduced form $R(M_{X_{s,t}})$ calculated  by {Algorithm 2.1.},
we can read the $\dim(H_r(|X|_{s})\rightarrow H_r(|X|_{s,t}))$
and $\dim(H_r(|X|_s \sqcup |X|_t)\rightarrow H_r(|X|_{s,t}))$.

Let $n_1= \sharp X_s$, $n_2=\sharp X_t$ and $n = \sharp X_{s,t}$.
A zero column $j$, $1\leq j\leq n_1$, contributes
a generator in the image of $H_r(|X|_{s})\rightarrow H_r(|X|_{s,t})$,
if there is no column $k>j$ such that $low(k) = j$.
Therefore $\dim(H_r(|X|_{s})\rightarrow H_r(|X|_{s,t}))$ = number of such columns
$j$.

A zero column $j'$, $1\leq j'\leq n_1+n_2$, contributes
a generator in the image of $H_r(|X|_s \sqcup |X|_t)\rightarrow H_r(|X|_{s,t})$,
if there is no column $k>j'$ such that $low(k) = j'$.
Therefore $\dim(H_r(|X|_s \sqcup |X|_t)\rightarrow H_r(|X|_{s,t}))$ = number of such columns
$j'$.

Let's change the order of cells in $X_{s,t}$.

Given $s<t$, order cells in $X_{s,t}$ in the following way:

(i) if $\tau \in X_t$, $\tau' \in X_s$ and $\tau'' \in X_{s,t}$,
then $\tau \prec \tau' \prec  \tau'' $;

(ii) if $\dim(\tau)<\dim(\tau')$, then $\tau \prec \tau'$.

We can write a matrix $M'_{X_{s,t}}$ for the boundary
map $\partial_{s,t}: X_{s,t} \rightarrow X_{s,t}$
with respect to the order above.

From the reduced form $R(M'_{X_{s,t}})$ calculated by {Algorithm 2.1.},
we can read the $\dim(H_r(|X|_{t})\rightarrow H_r(|X|_{s,t}))$
and $\dim(H_r(|X|_s \sqcup |X|_t)\rightarrow H_r(|X|_{s,t})$ as describe above.

Since
$$
I_r(s, t) = \text{img}(H_r(|X|_s)\rightarrow H_r(|X|_{s,t})) \cap \text{img}(H_r(|X|_{t})\rightarrow H_r(|X|_{s,t}))
$$
we have
$$
\begin{array}{rcl}
i_r(s, t) & = & \dim(I_r(s, t)) \\
 	    & = & \dim(\text{img}(H_r(|X|_s)\rightarrow H_r(|X|_{s,t})) \cap \text{img}(H_r(|X|_{t})\rightarrow H_r(|X|_{s,t}))) \\
	    & = & \dim(\text{img}(H_r(|X|_s)\rightarrow H_r(|X|_{s,t}))) + \dim(\text{img}(H_r(|X|_{t})\rightarrow H_r(|X|_{s,t}))) \\
	    &	 & - \dim(\text{img}(H_r(|X|_s)\rightarrow H_r(|X|_{s,t})) + \text{img}(H_r(|X|_{t})\rightarrow H_r(|X|_{s,t}))).\\
\end{array}
$$

Since
$$
\begin{array}{rl}
   & \text{img}(H_r(|X|_s)\rightarrow H_r(|X|_{s,t})) + \text{img}(H_r(|X|_{t})\rightarrow H_r(|X|_{s,t})) \\
= & \text{img}((H_r(|X|_s) \oplus H_r(|X|_{t})) \rightarrow H_r(|X|_{s,t})) \\
= & \text{img}(H_r(|X|_s \sqcup |X|_{t}) \rightarrow H_r(|X|_{s,t})) \\
\end{array}
$$
we have
$$
\begin{array}{rcl}
i_r(s, t) & = & \dim(\text{img}(H_r(|X|_s)\rightarrow H_r(|X|_{s,t}))) + \dim(\text{img}(H_r(|X|_{t})\rightarrow H_r(|X|_{s,t}))) \\
	    &	 & - \dim(\text{img}(H_r(|X|_s \sqcup |X|_{t}) \rightarrow H_r(|X|_{s,t}))) \\
\end{array}
$$
for $s<t$.

For $s=t$, we have
$$
i_r(s,s) = \dim(H_r(|X|_s))
$$
and this can be read from the reduced form $R(M_{X_s})$ directly.

Once we get $i_r(s, t)$ for any $s\leq t$, we can get the
numbers of all bar codes of level persistence
according to {Theorem \ref{theorem326}}.

Observe that in view of the tameness of $f$ we only need finitely many $i_r(s,t)$.
Let $t_{-1} = t_0-1$, $t_{N+1} = t_N + 1$ and
$t_{i+1/2} = (t_i+t_{i+1})/2$ for $-1\leq i \leq N$.
It is enough to know $i_r(s,t)$ for $s,t \in \{t_{i/2} | -1\leq i \leq 2N+1\}$.

Therefore we need to calculate reduced forms of about $4N^2$ matrices
in order to get $r$-bar codes of level persistence.

\vskip .2in 

{\large \textbf{Method 2}} (with coefficients in $\mathbb Z_2$)

\medskip 
Method 2 is based on the formulas of {Theorem \ref{theorem327}}.

So called positive and negative bar codes are introduced in the next
subsection in order to calculate $l_r(t_i)$, $l^+_r(t_i;t_j)$, $l^-_r(t_i;t_j)$
and $e_r(t_i;t_k,t_j)$ for $t_k\leq t_i \leq t_j$ efficiently.
We need to calculate reduced forms of about $2N$ matrices in order to
get $l_r(t_i)$, $l^+_r(t_i;t_j)$, $l^-_r(t_i;t_j)$
and $e_r(t_i;t_k,t_j)$ for $t_k\leq t_i \leq t_j$.

We still need to calculate $i_r(t_i;t_j)$ for $t_j<t_j$.
We have
$$
\begin{array}{rcl}
i_r(t_i, t_j) & = & \dim(\text{img}(H_r(|X|_{t_i})\rightarrow H_r(|X|_{t_i,t_j}))) + \dim(\text{img}(H_r(|X|_{t_j})\rightarrow H_r(|X|_{t_i,t_j}))) \\
	    &	 & - \dim(\text{img}(H_r(|X|_{t_i} \sqcup |X|_{t_j}) \rightarrow H_r(|X|_{t_i,t_j}))).\\
\end{array}
$$
for $t_i<t_j$.

$\dim(\text{img}(H_r(|X|_{t_i} \sqcup |X|_{t_j}) \rightarrow H_r(|X|_{t_i,t_j})))$ can be calculated in the same way as described in {Method 1}.

$\dim(\text{img}(H_r(|X|_{t_i})\rightarrow H_r(|X|_{t_i,t_j}))) $ and $\dim(\text{img}(H_r(|X|_{t_j})\rightarrow H_r(|X|_{t_i,t_j}))) $
can be calculated at very little cost with the knowledge of $l_r(t_i)$, $l^+_r(t_i;t_j)$ and $l^-_r(t_i;t_j)$.

$$
\begin{array}{rl}
& \dim(\text{img}(H_r(|X|_{t_i})\rightarrow H_r(|X|_{t_i,t_j})))  \\
= & \dim(H_r(|X|_{t_i})) - \dim(\ker(H_r(|X|_{t_i})\rightarrow H_r(|X|_{t_i,t_j}))) \\
= & l_r(t_i) - l^+_r(t_i;t_j)
\end{array}
$$

$$
\begin{array}{rl}
& \dim(\text{img}(H_r(|X|_{t_j})\rightarrow H_r(|X|_{t_i,t_j})))  \\
= & \dim(H_r(|X|_{t_j})) - \dim(\ker(H_r(|X|_{t_j})\rightarrow H_r(|X|_{t_i,t_j}))) \\
= & l_r(t_j) - l^-_r(t_i;t_j)
\end{array}
$$

We need to calculate reduced forms of about $N^2 / 2$ matrices in order to
get $i_r(t_i;t_j)$ for $t_j<t_j$.

Therefore we need to calculate reduced forms of about $N^2/2 + 2N$ matrices
in order to get $r$-bar codes of level persistence.

\subsection{Definition of Positive and Negative Bar Codes of a Tame Map} \label{defpnbarcode}

Suppose $f:X\rightarrow \mathbb{R}$ is a tame map with
$t_1<t_2<\cdots<t_N$ the critical values.
Define $t_{m+1/2}=(t_m+t_{m+1})/2$, $1\leq m\leq N-1$.
Given $t_1\leq u<s<t\leq t_N$, there exist $k,i,j$
such that $t_{k-1}<u\leq t_k$, $t_i\leq s<t_{i+1}$, $t_j\leq t<t_{j+1}$, we have

\begin{prop}\label{prop13}
$$
\begin{array}{l}
H_r(X_s)\cong\left\{
\begin{array}{ll}
H_r(X_{t_i}),& s=t_i\\
H_r(X_{t_{i+1/2}}),& t_i<s<t_{i+1}
\end{array}
\right.
\\
\\
\begin{array}{l}
\;\;\;\;\ker(H_r(X_s)\rightarrow H_r(X_{s,t})) \\
\cong\left\{
\begin{array}{ll}
\ker(H_r(X_{t_i})\rightarrow H_r(X_{t_i,t_j})),& s=t_i\\
\ker(H_r(X_{t_{i+1/2}})\rightarrow H_r(X_{t_{i+1/2}},t_j)),& t_i<s<t_{i+1}\leq t \\
0,& t_i<s<t<t_{i+1}
\end{array}
\right.
\end{array}
\\
\\
\begin{array}{l}
\;\;\;\; \ker(H_r(X_s)\rightarrow H_r(X_{u,s}))\\
\cong\left\{
\begin{array}{ll}
\ker(H_r(X_{t_i})\rightarrow H_r(X_{t_k,t_i})),& s=t_i\\
\ker(H_r(X_{t_{i+1/2}})\rightarrow H_r(X_{t_k,t_{i+1/2}})),& u\leq t_i<s<t_{i+1}\\
0,& t_i<u<s<t_{i+1}
\end{array}
\right.
\end{array}
\\
\\
\begin{array}{l}
\;\;\;\; \ker(H_r(X_s)\rightarrow H_r(X_{s,t}))\cap \ker(H_r(X_s)\rightarrow H_r(X_{u,s}))\\
\cong\left\{
\begin{array}{ll}
\ker(H_r(X_{t_i})\rightarrow H_r(X_{t_i,t_j}))\cap\ker(H_r(X_{t_i})\rightarrow H_r(X_{t_k,t_i})),& s=t_i\\
\ker(H_r(X_{t_{i+1/2}})\rightarrow H_r(X_{t_{i+1/2}},t_j))\cap\ker(H_r(X_{t_{i+1/2}})\rightarrow H_r(X_{t_k,t_{i+1/2}})), & u\leq t_i<s<t_{i+1}\leq t\\
0,& otherwise
\end{array}
\right.
\end{array}
\end{array}
$$
\end{prop}

\begin{proof}
We will verify the case $u\leq t_i<s<t_{i+1}\leq t$,
other cases can be proved similarly.

The commutative diagram (all maps are inclusions)
$$
\xymatrix{
X_{t_k,s} \ar[d] & X_s \ar[l] \ar[r] \ar[ld] \ar[rd] & X_{s,t_j} \ar[d] \\
X_{u,s} & & X_{s,t}
}
$$
induces
$$
\xymatrix{
H_r(X_{t_k,s}) \ar[d] & H_r(X_s) \ar[l] \ar[r] \ar[ld] \ar[rd] & H_r(X_{s,t_j}) \ar[d] \\
H_r(X_{u,s}) & & H_r(X_{s,t})
}
$$

We have two isomorphisms in the diagram above,
since $X_{u,s}$ and $X_{s,t}$ retracts by deformation onto
$X_{t_k,s}$ and $X_{s,t_j}$ by {Lemma \ref{lemma321}}.
So we have
$$
\begin{array}{rl}
  & \ker(H_r(X_s)\rightarrow H_r(X_{s,t})) = \ker(H_r(X_s)\rightarrow H_r(X_{s,t_j}))\\
  & \ker(H_r(X_s)\rightarrow H_r(X_{u,s})) = \ker(H_r(X_s)\rightarrow H_r(X_{t_k,s}))\\
  and &  \\
   & \ker(H_r(X_s)\rightarrow H_r(X_{s,t}))\cap \ker(H_r(X_s)\rightarrow H_r(X_{u,s}))\\
  = & \ker(H_r(X_s)\rightarrow H_r(X_{s,t_j}))\cap  \ker(H_r(X_s)\rightarrow H_r(X_{t_k,s}))
\end{array}
\eqno(\ast)
$$

When $s=t_i$, this provides the verification of the isomorphism.

When $t_i<s<t_{i+1}$, wlog, suppose $t_{i+1/2}\leq s<t_{i+1}$.
By {Lemma \ref{lemma323}}, $X_{t_{i+1/2},s}$ retracts by deformation onto
$X_{t_{i+1/2}}$ and $X_s$, the commutative diagram below (all maps are inclusions)
$$
\xymatrix{
 & X_s \ar[ld] \ar[d] \ar[r] & X_{s,t_j} \ar[d]\\
X_{t_k,s} & X_{t_{i+1/2,s}} \ar[l] \ar[r] & X_{t_{i+1/2},t_j}\\
X_{t_k,t_{i+1/2}} \ar[u] & X_{t_{i+1/2}} \ar[l] \ar[u] \ar[ru]&
}
$$
induces
$$
\xymatrix{
 & H_r(X_s) \ar[ld] \ar[d]^\cong \ar[r] & H_r(X_{s,t_j}) \ar[d]^\cong\\
H_r(X_{t_k,s}) & H_r(X_{t_{i+1/2,s}}) \ar[l] \ar[r] & H_r(X_{t_{i+1/2},t_j})\\
H_r(X_{t_k,t_{i+1/2}}) \ar[u]^\cong & H_r(X_{t_{i+1/2}}) \ar[l] \ar[u]^\cong \ar[ru]&
}
$$

From the diagram above, we have
$$
\begin{array}{rl}
   & H_r(X_s)\cong H_r(X_{t_{i+1/2}}) \\
   & \ker(H_r(X_s)\rightarrow H_r(X_{s,t_j}))\cong \ker(H_r(X_{t_{i+1/2}})\rightarrow H_r(X_{t_{i+1/2},t_j}))\\
   & \ker(H_r(X_s)\rightarrow H_r(X_{t_k,s}))\cong \ker(H_r(X_{t_{i+1/2}})\rightarrow H_r(X_{t_k,t_{i+1/2}}))\\
   & \ker(H_r(X_s)\rightarrow H_r(X_{s,t_j}))\cap\ker(H_r(X_s)\rightarrow H_r(X_{t_k,s}))\\
\cong & \ker(H_r(X_{t_{i+1/2}})\rightarrow H_r(X_{t_{i+1/2},t_j}))\cap\ker(H_r(X_{t_{i+1/2}})\rightarrow H_r(X_{t_k,t_{i+1/2}}))
\end{array}
$$

Combined with $(\ast)$, the isomorphisms follow as stated.

\end{proof}

From {Proposition \ref{prop13}}, we can see that relevant level
persistent numbers $l_r(s)$, $l_r^+(s;t)$,
$l_r^-(s;u)$ and $e_r(s;t,u)$ for $t_1\leq u \leq s \leq t \leq t_N$ are determined by
$l_r(t_{(i+1)/2}) $, $l^+_r(t_{(i+1)/2}; t_j) $, $l^-_r(t_{(i+1)/2}; t_k) $
and $e_r(t_{(i+1)/2}; t_j, t_k)$
for $1\leq i\leq 2N - 1$, $1\leq k \leq (i+1)/2 \leq j \leq N$.

\begin{prop}\label{prop14}Let
$$
\xymatrix{
V_1 & V_2 \ar[l]_{f_1} & \cdots \ar[l]_{f_2} & V_i \ar[l]_{f_{i-1}} \ar[r]^{f_i} & V_{i+1} \ar[r]^{f_{i+1}} & \cdots \ar[r]^{f_{N-1}} & V_N
}
$$
be a diagram of vector spaces $V_1,\cdots,V_N$ over a field $\kappa$.

Denote
$$
g_k = \left \{
\begin{array}{ll}
f_k\circ\cdots\circ f_{i-2}\circ f_{i-1} & k<i \\
id & k=i \\
f_k\circ\cdots\circ f_{i+1}\circ f_i & k>i
\end{array}
\right.
$$
$$
W_k = img(g_k:V_i\rightarrow V_k)
$$

One can always choose a basis $\{v_1,\cdots,v_n\}$ of $V_i$
such that condition $(P_1)$ below is satisfied for $1\leq k \leq N$.

\textbf{Condition}$(P_1)$: Nonzero elements of $\{g_k(v_1),\cdots,g_k(v_n)\}$
is a basis of $W_k$.

\end{prop}

\begin{proof}
First we obtain a base $\{ y_1,\cdots, y_n\}$ of $V_i$ such that 
$P_1$ is satisfied for all $k$, $i\leq k \leq N$ and a base
$\{ z_1, \cdots, z_n \}$ of $V_i$ such that $P_1$ is satisfied for all $k$, $1\leq k \leq i$.

(a) Let $\{x_1,\cdots,x_n\}$ be a basis of $V_i$, then $(P_1)$ is true for $k=i$.
Suppose $(P_1)$ is true for $i\leq k\leq j$, we can modify
the basis so that $(P_1)$ is also true for $i\leq k\leq j+1$.

If $g_{j+1}(x_i)=0$ for $1\leq i\leq n$, then $(P_1)$ is satisfied for $i\leq k\leq j+1$;
otherwise, there exist $1\leq n_1<n_2<\cdots<n_r\leq n$ such that
$g_{j+1}(x_{n_1}),\cdots,g_{j+1}(x_{n_r})$ are nonzero and $g_{j+1}(x_l)=0$ for
$l\notin\{n_1,\cdots,n_r\}$.

If $r=\dim(W_{j+1})$, then $\{g_{j+1}(x_{n_1}),\cdots,g_{j+1}(x_{n_r})\}$ forms
a basis of $W_{j+1}$; otherwise, there exist $2\leq r_1\leq r$
such that $g_{j+1}(x_{n_1}),\cdots,g_{j+1}(x_{n_{r_1-1}})$ are linearly independent,
but $\displaystyle g_{j+1}(x_{n_{r_1}}) = \sum_{l=1}^{r_1-1} a_l g_{j+1}(x_{n_l})$ for some $a_l\in \kappa$.

Let $x'_{n_{r_1}} = x_{n_{r_1}} - \displaystyle \sum_{l=1}^{r_1-1} a_l x_{n_l}$.
Replace $x_{n_r}$ by $x'_{n_{r_1}}$, $\{x_1,\cdots,x'_{n_{r_1}},\cdots,x_n\}$
is still a basis of $V_i$.
Since $g_{j+1}(x_{n_1}),\cdots,g_{j+1}(x_{n_{r_1}})$ are nonzero,
$g_k(x_{n_1}),\cdots,g_k(x_{n_{r_1}})$ are nonzero for $i\leq k\leq j$.
Since $(P_1)$ is true for $\{x_1,\cdots,x_n\}$ and $i\leq k\leq j$,
$g_k(x_{n_1}),\cdots,g_k(x_{n_{r_1}})$ are linearly independent for $i\leq k\leq j$.
Hence $g_k(x'_{n_{r_1}})\neq0$ for $i\leq k\leq j$.
Replace $g_k(x_{n_{r_1}})$ by $g_k(x'_{n_{r_1}})$, $\{g_k(x_1),\cdots,g_k(x_{n_{r_1}}),\cdots,g_k(x_n)\}$
and $\{g_k(x_1),\cdots,g_k(x'_{n_{r_1}}),\cdots,g_k(x_n)\}$ have the same number of nonzero elements
for $i\leq k\leq j$, since $g_k(x_{n_{r_1}})$ and $g_k(x'_{n_{r_1}})$ are nonzero for $i\leq k\leq j$.
In addition, those nonzero elements can be linearly represented by each other.
So nonzero elements of $\{g_k(x_1),\cdots,g_k(x'_{n_{r_1}}),\cdots,g_k(x_n)\}$ is still a basis
of $W_k$ for $i\leq k\leq j$.
Hence $(P_1)$ is still true for $\{x_1,\cdots,x'_{n_{r_1}},\cdots,x_n\}$ and $i\leq k\leq j$,
but the number of nonzero elements of $\{g_{j+1}(x_1),\cdots,g_{j+1}(x'_{n_{r_1}}),\cdots,g_{j+1}(x_n)\}$
is one less than the number of nonzero elements of $\{g_{j+1}(x_1),\cdots,g_{j+1}(x'_{n_{r_1}}),\cdots,g_{j+1}(x_n)\}$
since $g_{j+1}(x'_{n_{r_1}}) = 0$ by definition and $g_{j+1}(x_{n_{r_1}}) \neq 0$.

By induction on the number of nonzero elements of $\{g_{j+1}(x_1),\cdots,g_{j+1}(x_{n_{r_1}}),\cdots,$ $g_{j+1}(x_n)\}$,
there exist a basis $\{x'_1,\cdots,x'_n\}$ of $V_i$ such that $(P_1)$ is true for $i\leq k\leq j+1$.

Then by induction on $j$, we get a basis $\{y_1,\cdots,y_n\}$ of $V_i$ such that $(P_1)$ is true for $i\leq k\leq N$.

Similarly, we can get a basis $\{z_1,\cdots,z_n\}$ of $V_i$ such that $(P_1)$ is true for $1\leq k\leq i$.
To finalize the proof we need a definition.

(b)
 \textbf{Definition: Negative Life}
Given $z_l\in V_i$ as in the basis above, if $g_k(z_l)\neq 0$ and
$g_{k-1}(z_l)=0$ for some $2\leq k\leq i$ or $g_k(z_l)\neq 0$ for $k=1$,
define the \emph{negative life} of $z_l$ to be $i-k+1$.

Order $z_1,\cdots,z_n$ according to their negative life increasingly.

\textbf{Claim:} For any given $1\leq l\leq n-1$, if we replace
$z_{l+1}$ by $z'_{l+1} = z_{l+1}-\displaystyle \sum_{m=1}^l b_m z_m$
for any given $b_m\in\kappa$, $\{z_1,\cdots,z'_{l+1},\cdots,z_n\}$
still satisfies $(P_1)$ for $1\leq k\leq i$.

Proof of the claim:

Since $z_1,\cdots,z_n$ are ordered according to their negative life increasingly,
the negative life of $z_m$ is less than the negative life of $z_{l+1}$ for
$1\leq m\leq l$. Suppose $g_{k_1}(z_{l+1})\neq 0$ and
$g_{k_1-1}(z_{l+1})=0$ for some $2\leq k_1\leq i$, then we must have
$g_{k}(z_m)=0$ for all $1\leq k\leq k_1-1$ and $1\leq m\leq l+1$.
Therefore $g_k(z_{l+1})=g_k(z'_{l+1})=0$ for $1\leq k\leq k_1-1$.
Hence $\{g_k(z_1),\cdots,g_k(z'_{l+1}),\cdots,g_k(z_n)\}$ is the same as
$\{g_k(z_1),\cdots,g_k(z_n)\}$ for $1\leq k\leq k_1-1$.
So $\{z_1,\cdots,z'_{l+1},\cdots,z_n\}$
satisfies $(P_1)$ for $1\leq k\leq k_1-1$.

Given $k_1\leq k'\leq i$, we have $g_{k'}(z_{l+1})\neq 0$.
There exist $1\leq l_1\leq l$, such that $g_{k'}(z_m)=0$ for $1\leq m \leq l_1$ and
 $g_{k'}(z_m)\neq 0$ for $l_1+1\leq m\leq l+1$.
 Hence $g_{k'}(z'_{l+1}) = g_{k'}(z_{l+1})-\displaystyle \sum_{m=l_1+1}^l b_m g_{k'}(z_m)$.
Therefore nonzero elements of $\{g_{k'}(z_1),\cdots,g_{k'}(z_n)\}$ and
$\{g_{k'}(z_1),\cdots,g_{k'}(z'_{l+_1}),\cdots,g_{k'}(z_n)\}$ can be linearly represented
by each other.
Since $z_1,\cdots,z_n$ satisfies $(P_1)$ for $1\leq k\leq i$, nonzero elements of $\{g_{k'}(z_1),\cdots,g_{k'}(z_n)\}$
form a basis of $W_{k'}$.
Then nonzero elements of $\{g_{k'}(z_1),\cdots,g_{k'}(z'_{l+_1}),\cdots,g_{k'}(z_n)\}$ also
form a basis of $W_{k'}$.
So $\{z_1,\cdots,z'_{l+1},\cdots,z_n\}$
satisfies $(P_1)$ for $k_1\leq k\leq i$.

The claim is proved.

Notice that in part(a), we replace
 $x_{n_{r_1}}$ by $x'_{n_{r_1}} = x_{n_{r_1}} - \displaystyle \sum_{l=1}^{r_1-1} a_l x_{n_l}$.
We only need this type of replacement to modify a basis $\{x_1,\cdots,x_n\}$ into a basis
$\{y_1,\cdots,y_n\}$ that satisfies $(P_1)$ for $i\leq k\leq N$.
And the replacement
 $x_{n_{r_1}}$ by $x'_{n_{r_1}} = x_{n_{r_1}} - \displaystyle \sum_{l=1}^{r_1-1} a_l x_{n_l}$
  is only a special case of the replacement in the above claim.
  Therefore, start from a basis $\{z_1,\cdots,z_n\}$ of $V_i$ such that $(P_1)$ is true for
 $1\leq k\leq i$ which is ordered according to their negative life increasingly, we can
 use the same process in part(a) to modify this basis until $(P_1)$ is satisfied for $1\leq k\leq N$.

\end{proof}

We use the concept ``positive bar code'' because it is consistent with the sub-level persistence
associated to the space $X_{[t_{(i+1)/2},\infty)}$ and the map $f|_{X_{[t_{(i+1)/2},\infty)}}$.
Similarly, we use the concept ``negative bar code'' because it is consistent with the sub-level persistence
associated to the space $X_{(-\infty,t_{(i+1)/2}]}$ and the map $f|_{X_{(-\infty,t_{(i+1)/2}]}}$.

Let $d_{(i+1)/2}$ denote $\dim(H_r(X_{t_{(i+1)/2}}))$. Choose a basis $\{c_1,\cdots,c_{d_{(i+1)/2}}\}$ of
$H_r(X_{t_{(i+1)/2}})$ such that for any $k<(i+1)/2$ and $j>(i+1)/2$, nonzero elements
of $\{\sigma_k(c_1),\cdots,\sigma_k(c_{d_{(i+1)/2}})\}$ is always a basis of
$\text{img}(\sigma_k:H_r(X_{t_{(i+1)/2}})\rightarrow H_r(X_{t_k,t_{(i+1)/2}}))$ and nonzero elements of
$\{\tau_j(c_1),\cdots,\tau_j(c_{d_{(i+1)/2}})\}$ is always a basis of
$\text{img}(\tau_j:H_r(X_{t_{(i+1)/2}})\rightarrow H_r(X_{t_{(i+1)/2},t_j}))$, where $\sigma_k$ and
$\tau_j$ are induced by inclusion. This is possible by {Proposition \ref{prop14}}.

\begin{defn}Positive and Negative Bar Codes

(i) The set of $r$-positive bar codes $B_r^+(f;t_{(i+1)/2})$ of $X_{t_{(i+1)/2},\infty}$
is a collection of $d_{(i+1)/2}$ intervals described as follows. For any generator
$c_a\in\{c_1,\cdots,c_{d_{(i+1)/2}}\}$, if $c_a$ dies upward at $t_{j'}(t_{j'}>t_{(i+1)/2})$,
then we have an interval $\langle t_{(i+1)/2},t_{j'})$ in $B_r^+(f;t_{(i+1)/2})$;
otherwise, if $c_a$ never dies upward, we have an interval
$\langle t_{(i+1)/2},\infty)$ in $B_r^+(f;t_{(i+1)/2})$.

(ii) The set of $r$-negative bar codes $B_r^-(f;t_{(i+1)/2})$ of $X_{-\infty,t_{(i+1)/2}}$
is a collection of $d_{(i+1)/2}$ intervals described as follows. For any generator
$c_b\in\{c_1,\cdots,c_{d_{(i+1)/2}}\}$, if $c_b$ dies downward at $t_{k'}(t_{k'}<t_{(i+1)/2})$,
then we have an interval $(t_{k'},t_{(i+1)/2}\rangle$ in $B_r^-(f;t_{(i+1)/2})$;
otherwise, if $c_b$ never dies downward, we have an interval
$(-\infty,t_{(i+1)/2}\rangle$ in $B_r^-(f;t_{(i+1)/2})$.

(iii) Each generator $c_k$ correspond to an interval $(a_k,t_{(i+1)/2}\rangle$ in the negative bar code $B^-_r(f,t_{(i+1)/2})$
and an interval $\langle t_{(i+1)/2},b_k)$ in the positive bar code $B^+_r(f,t_{(i+1)/2})$.
 Therefore each generator $c_k$ defines an interval $(a_k,b_k)$.
Define $B_r(f,t_{(i+1)/2})$ as the collection of the above intervals $\{(a_k,b_k)|1\leq k\leq d_{(i+1)/2}\}$.

\end{defn}

We have introduced these additional numbers given the fact they can be efficiently calculated as
indicated in {subsection \ref{cptpnbarcode}} and they are equivalent to the relevant persistent numbers
$l_r(t_{(i+1)/2}) $, $l^+_r(t_{(i+1)/2}; t_j) $, $l^-_r(t_{(i+1)/2}; t_k) $
and $e_r(t_{(i+1)/2}; t_j, t_k)$ as indicated below.

We can get the relevant persistent numbers $l_r(t_{(i+1)/2}) $, $l^+_r(t_{(i+1)/2}; t_j) $, $l^-_r(t_{(i+1)/2}; t_k) $
and $e_r(t_{(i+1)/2};$ $t_j, t_k)$ from the positive and negative bar codes:
\begin{equation} \label{341}
l_r(t_{(i+1)/2}) = \dim H_r(X_{t_{(i+1)/2}}) = |B^+_r(f,t_{(i+1)/2})| = |B^-_r(f,t_{(i+1)/2})|
\end{equation}
\begin{equation} \label{342}
l^+_r(t_{(i+1)/2}; t_j) = l_r(t_{(i+1)/2}) - \sharp\{\langle t_{(i+1)/2}, b)\in B^+_r(f, t_{(i+1)/2}) \mid b>t_j\}
\end{equation}
\begin{equation} \label{343}
l^-_r(t_{(i+1)/2}; t_k) = l_r(t_{(i+1)/2}) - \sharp\{(a,t_{(i+1)/2}\rangle\in B^-_r(f, t_{(i+1)/2}) \mid a<t_k\}
\end{equation}
\begin{equation} \label{344}
\begin{array}{rl}
   & e_r(t_{(i+1)/2}; t_j, t_k) \\
 = & l_r(t_{(i+1)/2}) - (l_r(t_{(i+1)/2}) - l^+_r(t_{(i+1)/2}; t_j)) - (l_r(t_{(i+1)/2}) - l^-_r(t_{(i+1)/2}; t_k)) \\
 & + \sharp\{(a,b)\in B_r(f,t_{(i+1)/2}) \mid a<t_k, b>t_j\} \\
 = & l^+_r(t_{(i+1)/2}; t_j) + l^-_r(t_{(i+1)/2}; t_k) - l_r(t_{(i+1)/2})
 + \sharp\{(a,b)\in B_r(f,t_{(i+1)/2}) \mid a<t_k, b>t_j\}
\end{array}
\end{equation}

We can also get the positive and negative bar codes from the relevant persistent numbers
$l_r(t_{(i+1)/2}) $, $l^+_r(t_{(i+1)/2}; t_j) $, $l^-_r(t_{(i+1)/2}; t_k) $
and $e_r(t_{(i+1)/2}; t_j, t_k)$:
$$
\sharp\{\langle t_{(i+1)/2},t_j)\in B^+_r(f,t_{(i+1)/2})\} = l^+_r(t_{(i+1)/2}; t_j) - l^+_r(t_{(i+1)/2}; t_{j-1})
$$
$$
\sharp\{(t_k, t_{(i+1)/2}\rangle \in B^-_r(f, t_{(i+1)/2})\} = l^-_r(t_{(i+1)/2}; t_k) - l^-_r(t_{(i+1)/2}; t_{k+1})
$$
$$
\begin{array}{rl}
\sharp \{ (t_k,t_j) \in B_r(f, t_{(i+1)/2}) \} = & e(t_{(i+1)/2};t_j,t_k) - e(t_{(i+1)/2};t_j,t_{k+1}) - e(t_{(i+1)/2}; t_{j-1}, t_k) \\
 & + e(t_{(i+1)/2}; t_{j-1}, t_{k+1})
 \end{array}
$$

Therefore, positive and negative bar codes and relevant persistent numbers
$l_r(t_{(i+1)/2}) $, $l^+_r(t_{(i+1)/2}; t_j) $, $l^-_r(t_{(i+1)/2}; t_k) $
and $e_r(t_{(i+1)/2}; t_j, t_k)$ are equivalent.

\subsection{Computing Positive and Negative Bar codes
of Generic $\mathbb{R}$-valued Linear Maps
 on Simplicial Complexes} \label{cptpnbarcode}

Let $X$ be a simplicial complex with $N$ vertices.
Let $f:|X|\rightarrow \mathbb{R}$ be a generic linear map
on $X$, with $t_1<\cdots<t_N$ being values on vertices $x_1,\cdots, x_N$ of $X$.
According to {Method 2} in {subsection \ref{frm}},
we need to calculate $l_r(t_i)$, $l^+_r(t_i;t_j)$, $l^-_r(t_i;t_j)$
and $e_r(t_i;t_k,t_j)$ for $1\leq k\leq i \leq j\leq N$.
So we only need to calculate $B_r^+(|X|;t_i)$ and $B_r^-(|X|;t_i)$ for $1\leq i \leq N$,
then plug in (\ref{341})-(\ref{344}).

We'll describe how to compute
positive bar codes in details and negative bar codes
briefly at the end.

In order to get $B_r^+(|X|;t_i)$, we have to consider the filtration below
$$
|X|_{t_i}\subseteq |X|_{t_i,t_{i+1}}\subseteq \cdots\subseteq |X|_{t_i,t_N}
$$
(see {Definition 4.1 and 4.2}, \cite{BDD}).

It is not economical to compute the positive bar codes from this
filtration directly. Instead, we'll consider an equivalent
filtration.

Suppose that the vertices are ordered from $1$ to $N$.

\begin{defn}\label{def21}
Given a $j$-simplex $\sigma=[x_{n_0},\cdots,x_{n_j}]$ of $X$,   $1\leq n_0<\cdots< n_j\leq N$,
define $t_{min}(|\sigma|)=t_{n_0}$, $t_{max}(|\sigma|)=t_{n_j}$.
For $t_{n_0} < s < t_{n_j}$, let $|\sigma|_s$ denote $|\sigma| \cap |X|_s$,  
$|\sigma|_{s,\infty}$ denote $|\sigma| \cap |X|_{s,\infty}$, $|\sigma|_{-\infty,s}$ denote $|\sigma| \cap |X|_{-\infty,s}$.

Extend $t_{min}$ and $t_{max}$ to $|\sigma|_s$, $|\sigma|_{s,\infty}$ and
$|\sigma|_{-\infty,s}$, we have
$$
\begin{array}{ll}
t_{min}(|\sigma|_s)=t_{max}(|\sigma|_s)=s,&\\
t_{min}(|\sigma|_{s,\infty})=s,&t_{max}(|\sigma|_{s,\infty})=t_{n_j},\\
t_{min}(|\sigma|_{-\infty,s})=t_{n_0},&t_{max}(|\sigma|_{-\infty,s})=s.
\end{array}
$$

\end{defn}

\textbf{Note:}
It is easy to see that all cells in $X_{t_i,\infty}$
are  $\{ |\sigma|  \mid    t_{min}(|\sigma|)\geq t_i \}$,
$\{ |\sigma|_{t_i} \mid      t_{min}(|\sigma|)<t_i<t_{max}(|\sigma|) \}$
or $\{ |\sigma|_{t_i,\infty} \mid      t_{min}(|\sigma|)<t_i<t_{max}(|\sigma|)\}$.

\begin{defn} \label{Yst}
Given $s, t$ such that $t_1\leq s<t\leq t_N$, define $Y_{s,t}$ to be a
cell complex consist of cells in $X_{s,\infty}$ that are
contained in the space $X_{-\infty,t}$.
In another words, the cells in $Y_{s,t}$ are the cells $\mathbf{c}$
in $X_{s,\infty}$ which satisfy $t_{max}(\mathbf{c})\leq t$.

\end{defn}

\begin{lem}\label{lemma21}
Let $\sigma=[x_0,\cdots,x_n]$ be an $n$-simplex and
$f:|\sigma|\rightarrow \mathbb{R}$ be a generic linear map on it.
Suppose
$f(x_0)<f(x_1)<\cdots<f(x_n)$,
let $t_j$ denote $f(x_j)$, $0\leq j\leq n$.
Given $t_0<s<t_n$, suppose $t_i\leq s<t_{i+1}$ for some $0\leq i\leq n-1$.
Let $|\sigma|_{-\infty,s}$ denote $f^{-1}(-\infty,s]$, then
$|\sigma|_{-\infty,s}$ retracts by deformation onto $|[x_0,\cdots,x_i]|$.

\end{lem}

\begin{proof}
Given $x\in|\sigma|_{-\infty,s}$, write $x$ as $\displaystyle \sum_{j=0}^n a_j x_j$,
where $0\leq a_j\leq 1$,$\displaystyle \sum_{j=0}^n a_j=1$.
We have $\displaystyle \sum_{j=0}^i a_j>0$, otherwise $a_j=0$,$0\leq j\leq i$,
$x=\displaystyle\sum_{j=i+1}^n a_j x_j$, $f(x)=\displaystyle\sum_{j=i+1}^n a_j t_j>s$, since
$t_j>s$ for $j\geq i+1$ and $\displaystyle \sum_{j=i+1}^n a_j=1$.
This is a contradiction with $x\in|\sigma|_{-\infty,s}$.

Define $g_0:|\sigma|_{-\infty,s}\rightarrow |\sigma|_{-\infty,s}$ to be the identity map.

Define
$$
\begin{array}{rcl}
g_1:|\sigma|_{-\infty,s} & \rightarrow & |\sigma|_{-\infty,s} \\
x=\displaystyle\sum_{j=0}^n a_j x_j & \mapsto & \frac{1}{\displaystyle \sum_{j=0}^i a_j}\displaystyle \sum_{j=0}^i a_j x_j \\
\end{array}
$$
$g_1$ is a well-defined continuous map and
$g_1|_{|[x_0,\cdots,x_i]|}$ is the identity.

Define
$$
\begin{array}{rcl}
G:|\sigma|_{-\infty,s}\times I & \rightarrow & |\sigma|_{-\infty,s} \\
(x,\tau) & \mapsto & (1-\tau)g_0(x)+\tau g_1(x)
\end{array}
$$
$G$ is a deformation retraction from $|\sigma|_{-\infty,s}$ onto $|[x_0,\cdots,x_i]|$
and it is canonical.

\end{proof}

\begin{figure}[h!]
\centering
  \includegraphics[scale=.5]{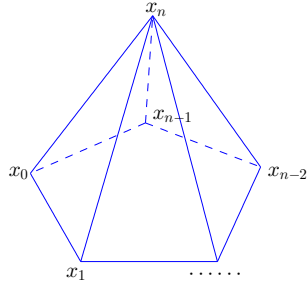}
  \caption{A cone with apex $x_n$ and base a convex cell with vertices $x_0,\cdots,x_{n-1}$.} \label {figure451}
\end{figure}

\begin{lem}\label{lemma22}
Let $\sigma$ be a cone with apex $x_n$
and base a convex cell
with vertices $x_0,\cdots,x_{n-1}$.
See {Figure \ref{figure451}} above.
Let $f$ be a linear map on $|\sigma|$ such that
$f(x_0)\leq f(x_1)\leq \cdots\leq f(x_{n-1})<f(x_n)$.
For $f(x_{n-1})\leq s < f(x_n)$, let $|\sigma|_{-\infty,s}$ denote $f^{-1}(-\infty,s]$,  $|\sigma|_{-\infty,s}$ retracts by
deformation onto the base, again by a canonical retraction.

\end{lem}

\begin{proof}
Since $\sigma$ is a cone, for any point $x\in |\sigma|(x\neq x_n)$,
there exists a unique $x'$ in the base such that $x=(1-a)x'+a x_n$,
for some $0\leq a \leq 1$.

Define $g_0:|\sigma|_{-\infty,s}\rightarrow |\sigma|_{-\infty,s}$ to be the identity map.

Define
$$
\begin{array}{rcl}
g_1: |\sigma|_{-\infty,s} & \rightarrow & |\sigma|_{-\infty,s}\\
x & \mapsto & x'
\end{array}
$$

$g_1$ is a well-defined continuous map and $g_1|_{the\; base}$
is the identity.

Define
$$
\begin{array}{rcl}
G: |\sigma|_{-\infty,s}\times I & \rightarrow & |\sigma|_{-\infty,s}\\
(x,\tau) & \mapsto & (1-\tau)g_0(x)+\tau g_1(x)
\end{array}
$$

G is a deformation retraction from $|\sigma|_{-\infty,s}$ onto the base.

\end{proof}

\begin{figure}[h!]
\centering
$
\begin{array}{c}
\includegraphics[scale=.5]{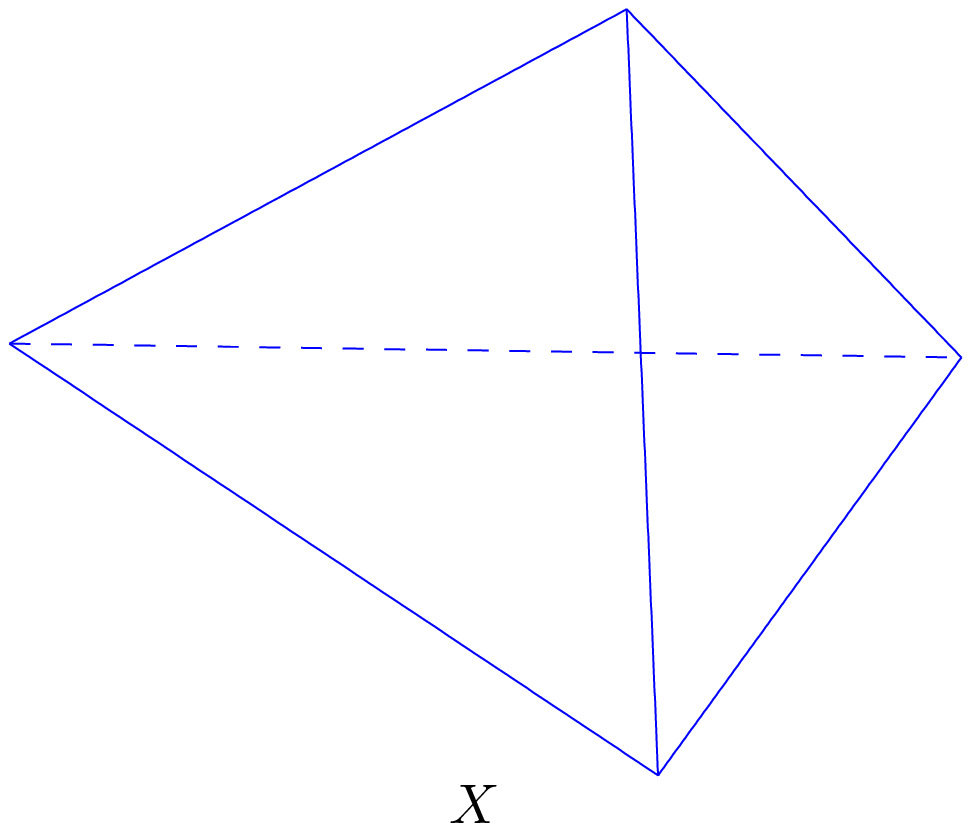}\\
\begin{array}{cc}
\includegraphics[scale=.5]{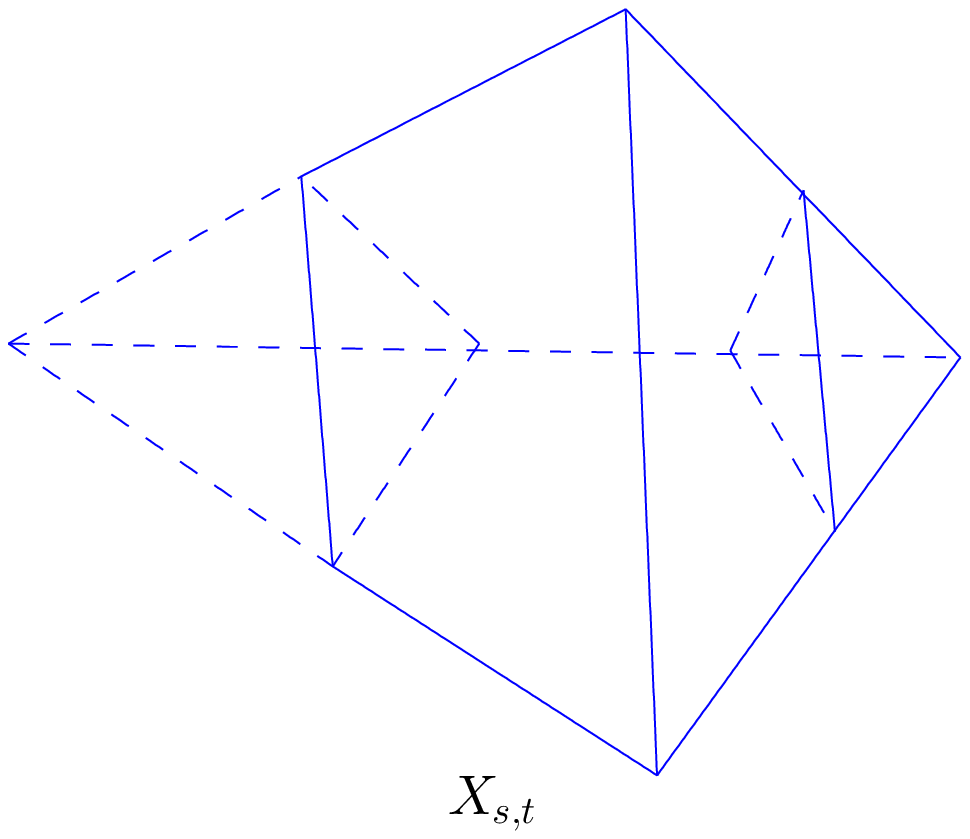} & \includegraphics[scale=.5]{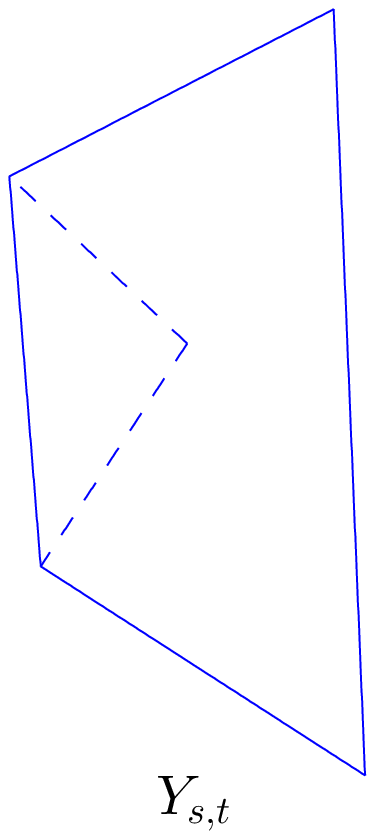}
\end{array}
\end{array}
$
  \caption{$X$, $X_{s,t}$ and $Y_{s,t}$ .} \label{figure452}
\end{figure}

\begin{lem}\label{lemma23}
Let
$X=[x_0,x_1,\cdots,x_n]$ be an $n$-simplex and
$f$ be a generic linear map on $|X|$.
Suppose $f(x_0)<f(x_1)<\cdots<f(x_n)$,
let $t_j$ denote $f(x_j)$, $0\leq j\leq n$.
Given $t_0<s<t<t_n$,
let $Y_{s,t}$ be as in {Definition \ref{Yst}},
then $X_{s,t}$ retracts by deformation onto $Y_{s,t}$.
See {Figure \ref{figure452}} above.
\end{lem}

\begin{proof}
Suppose $t_i<s\leq t_{i+1}$, $t_j\leq t<t_{j+1}$, $0\leq i\leq j\leq n-1$.
When $i=j$, $Y_{s,t}=X_s$ and $X_{s,t}$ retracts by deformation
onto $X_s$ by {Lemma\ref{lemma323}}.
When $i<j$, $X_{s,t}$ retracts by deformation onto $X_{s,t_j}$.
Since $Y_{s,t}=Y_{s,t_j}$, we only need to prove $X_{s,t_j}\searrow  Y_{s,t_j}$.

\textbf{Note:}``$\searrow$'' means ``retracts by deformation onto'' here and after.

$X_{s,\infty}$ is an $n$-dim convex cell, which can be viewed as
union of $n$-dim cones with apex $x_n$ and base $(n-1)$-dim
faces of $X_{s,\infty}$ which do not contain $x_n$.

By {Lemma \ref{lemma22}}, the intersection of each of the cones
and $X_{s,t_{n-1}}$ retracts by deformation to the base of
that cone.

Since the deformation retractions are compatible  on
the intersection of the cones, $X_{s,t_{n-1}}$ = union of
intersections of the cones and $X_s,t_{n-1}$ $\searrow$
union of bases of the cones = $(n-1)$-dim faces
of $X_{s,\infty}$ which do not contain $x_n=Y_{s,t_{n-1}}$.

Since $Y_{s,t_{n-1}}$ are union of $(n-1)$-dim faces of
$X_{s,\infty}$ which do not contain $x_n$,
$Y_{s,t_{n-1}}=Y_{s,t_{n-2}}\cup$ cells of $Y_{s,t_{n-1}}$ which contain $x_{n-1}$.

Let $\sigma_1,\sigma_2$ be $(n-1)$-dim cells of $Y_{s,t_{n-1}}$ which contain $x_{n-1}$,
$\sigma_1\cap\sigma_2\neq\phi$. $\sigma_1$ can be viewed as union of $(n-1)$-dim
cones with apex $x_{n-1}$.

By {Lemma \ref{lemma22}}, $\sigma_1\cap X_{s,t_{n-2}}$ $\searrow$ a cell sub complex of
$Y_{s,t_{n-2}}$. Similarly, $\sigma_2\cap X_{s,t_{n-2}}$ $\searrow$ a cell sub complex of $Y_{s,t_{n-2}}$.

The deformation retractions on $\sigma_1\cap X_{s,t_{n-2}}$ and
$\sigma_2\cap X_{s,t_{n-2}}$ induce the same deformation
retractions on $\sigma_1\cap\sigma_2\cap X_{s,t_{n-2}}$.
Therefore $Y_{s,t_{n-1}}\cap X_{s,t_{n-2}}\searrow Y_{s,t_{n-2}}$.

$X_{s,t_{n-1}}\searrow Y_{s,t_{n-1}}$ implies $X_{s,t_{n-2}}\searrow Y_{s,t_{n-1}}\cap X_{s,t_{n-2}}$.
So $X_{s,t_{n-2}}\searrow Y_{s,t_{n-2}}$.
Suppose $X_{s,t_k}\searrow Y_{s,t_k}$, we have $X_{s,t_{k-1}}\searrow Y_{s,t_k}\cap X_{s,t_{k-1}}$.

Let $\sigma$ be a cell of $Y_{s,t_k}$ which contains $x_k$.
$\sigma$ can be viewed as union of cones with
apex $x_k$ and base convex cells in $Y_{s,t_{k-1}}$.

By {Lemma \ref{lemma22}}, $\sigma\cap X_{s,t_{k-1}}$ retracts by
deformation onto a cell sub complex of $Y_{s,t_{k-1}}$.

The above deformation retractions are coherent on the
boundary of the cones, so they induce a
deformation retraction from $Y_{s,t_k}\cap X_{s,t_{k-1}}$ onto $Y_{s,t_{k-1}}$.

By induction $X_{s,t_j}\searrow Y_{s,t_j}$.

\end{proof}

\begin{prop}\label{prop24}
$X_{s,t}$ retracts by deformation onto $Y_{s,t}$.
\end{prop}

\begin{proof}
Four types of simplices of $X$ have
nonempty intersection with space $|X|_{s,t}$:

$(1)$ $\{\sigma\in X \mid t_{min}(\sigma)<s, s<t_{max}(\sigma)\leq t\}$

$(2)$ $\{\sigma\in X \mid t_{min}(\sigma)<s, t_{max}(\sigma)>t\}$

$(3)$ $\{\sigma\in X \mid t_{min}(\sigma)\geq s, t_{max}(\sigma)\leq t\}$

$(4)$ $\{\sigma\in X \mid s\leq t_{min}(\sigma)\leq t, t_{max}(\sigma)>t\}$

Therefore, there are four types of cells of $X_{s,t}$:

$(1)$ $\{\sigma_{s,\infty} \mid \sigma\in X, t_{min}(\sigma)<s, s<t_{max}(\sigma)\leq t\}$

$(2)$ $\{\sigma_{s,t} \mid \sigma\in X, t_{min}(\sigma)<s, t_{max}(\sigma)>t\}$

$(3)$ $\{\sigma\in X \mid t_{min}(\sigma)\geq s, t_{max}(\sigma)\leq t\}$

$(4)$ $\{\sigma_{-\infty,t} \mid \sigma\in X, s\leq t_{min}(\sigma)\leq t, t_{max}(\sigma)>t\}$

There are five types of cells of $X_{s,\infty}$ with respect to $t$ $(s<t)$:

$(1)$ $\{\sigma_{s,\infty} \mid \sigma\in X, t_{min}(\sigma)<s, s<t_{max}(\sigma)\leq t\}$

$(2)$ $\{\sigma_{s,\infty} \mid \sigma\in X, t_{min}(\sigma)<s, t_{max}(\sigma)>t\}$

$(3)$ $\{\sigma\in X \mid t_{min}(\sigma)\geq s, t_{max}(\sigma)\leq t\}$

$(4)$ $\{\sigma\in X \mid s\leq t_{min}(\sigma)\leq t, t_{max}(\sigma)>t\}$

$(5)$ $\{\sigma\in X \mid t_{min}(\sigma)>t\}$

There are four types of cell complexes of $Y_{s,t}$:

$(1)$ $\{\sigma_{s,\infty} \mid \sigma\in X, t_{min}(\sigma)<s, s<t_{max}(\sigma)\leq t\}$

$(2)$ $\{\text{subcomplexes of } \sigma_{s,\infty}\text{ which do not contain }x_{m+1},\cdots,x_n \mid
\sigma=[x_0,\cdots,x_n]\in X, t_{min}(\sigma)<s, t_{max}(\sigma)>t,
f([x_0,\cdots,x_i])\leq t, 0\leq i\leq m, f([x_0,\cdots,x_i])>t, m+1\leq i\leq n\}$

$(3)$ $\{\sigma\in X \mid t_{min}(\sigma)\geq s, t_{max}(\sigma)\leq t\}$

$(4)$ $\{[x_0,\cdots,x_m] \mid \sigma=[x_0,\cdots,x_n]\in X, s\leq t_{min}(\sigma)\leq t, t_{max}(\sigma)>t,
 f([x_0,\cdots,x_i])\leq t, 0\leq i\leq m,
f([x_0,\cdots,x_i])>t, m+1\leq i\leq n\}$

Notice that type $(1)$ and $(3)$ cells of $X_{s,t}$ and $Y_{s,t}$
are the same. By {Lemma \ref{lemma22}} type $(2)$ cells of $X_{s,t}$
retract by deformation onto type $(2)$ cell complexes
of $Y_{s,t}$. By {Lemma \ref{lemma21}} type $(4)$ cells of $X_{s,t}$
retract by deformation onto type $(4)$ simplices of $Y_{s,t}$.

\end{proof}

By {Proposition \ref{prop24}}, the filtration
$$
Y_{t_i}\subseteq Y_{t_i,t_{i+1}}\subseteq\cdots\subseteq Y_{t_i,t_N}
$$
provides the same bar codes as the filtration
$$
X_{t_i}\subseteq X_{t_i,t_{i+1}}\subseteq\cdots\subseteq X_{t_i,t_N}.
$$

The advantage to consider $Y_{t_i,t_j}$ instead of
$X_{t_i,t_j}$ is that $Y_{t_i,t_j}$ keeps more simplices from
$X$ and generates less ``new'' cells, which
makes representation of cells and construction
of boundary matrices more economical.

\textbf{Note:} From now on, we'll only consider homology
groups with $\mathbb{Z}_2$ coefficients. $H_r(X;\mathbb{Z}_2)$ will be
simply written as $H_r(X)$.

Let $X$ be a simplicial complex containing $N$ vertices
$x_1,x_2,\cdots,x_N$. Let $n$ denote $\dim X$, $m_j$ denote number of $j$-simplices of $X$,
$0\leq j\leq n$. Let $f$ be a generic linear map on $X$. For
simplicity, let $f(x_i)=i$, $1\leq i\leq N$. Represent $j$-simplex
of $X$ by $(a_0,a_1,\cdots,a_j)$, $1\leq a_0<a_1<\cdots<a_j\leq N$.

All information on the above simplicial complex $X$ can
be stored in $n+1$ matrices named $\text{Simplex}_0$, $\text{Simplex}_1$, 
$\cdots$, $\text{Simplex}_n$. The matrix $\text{Simplex}_j$ is an
$m_j\times (j+1)$ matrix, which stores $j$-simplices of $X$
as rows in lexicographic order.

\begin{figure}[h!]
\centering
 \includegraphics[height=50mm]{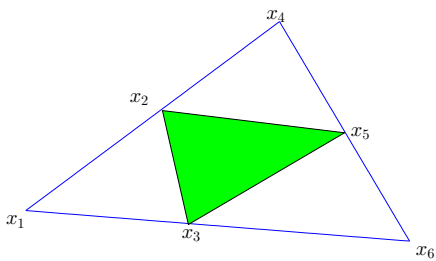}
\caption{Figure for Example \ref{example21}.}
\end{figure}

\begin{exmp}\label{example21}
\end{exmp}
\begin{tabular}{c c c}
  $\text{Simplex}_0=\left(
               \begin{array}{c}
                 1 \\
                 2 \\
                 3 \\
                 4 \\
                 5 \\
                 6 \\
               \end{array}
             \right)
  $  &
  $\text{Simplex}_1=\left(
                \begin{array}{c}
                     12 \\
                     13 \\
                     23 \\
                     24 \\
                     25 \\
                     35 \\
                     36 \\
                     45 \\
                     56 \\
                   \end{array}
                \right)
  $ &
  $\text{Simplex}_2=\left(
  \begin{array}{ccc}
  2 & 3 & 5
  \end{array}
  \right)
  $\\
\end{tabular}

\begin{defn} \textbf{Initial Order}

Order all simplices of $X$ first according to dimension
increasingly, then according to lexicographic order,
we call this order as \emph{initial order} of simplices
of $X$.

\end{defn}

\begin{exmp}
The initial order for {Example \ref{example21}} is
$1<2<3<4<5<6<12<13<23<24<25<35<36<45<56<235$.
\end{exmp}

Let $\sigma_1,\sigma_2,\cdots,\sigma_m$ be simplices of $X$ in initial order,
where $m=\displaystyle\sum_{j=0}^n m_j$. We have a boundary matrix of $\mathbb{Z}_2$
coefficients
$$
\partial:= \begin{tabular}{cc}
              & $\begin{array}{cccccc}
                   \sigma_1 & \sigma_2 & \cdots & \sigma_j & \cdots & \sigma_m
                 \end{array}$
              \\
             $\begin{array}{c}
               \sigma_1 \\
               \sigma_2 \\
               \vdots \\
               \sigma_i \\
               \vdots \\
               \sigma_m
             \end{array}$
              & $\left(
                  \begin{array}{cccccccccccc}
                     & & & & & & & & & & &\\
                     & & & & & & & & & & &\\
                     & & & & & & & & & & &\\
                     & & & & & & & & & & &\\[-.4cm]
                     & & & & & & \hskip 2mm \partial_{ij} & & & & &\\
                     & & & & & & & & & & &\\
                     & & & & & & & & & & &\\
                  \end{array}
                \right)$
               \\
           \end{tabular}
$$

$$
\partial_{ij}=\left\{
\begin{array}{ll}
1 & \text{if }\sigma_i\text{ is a codimension-1 face of }\sigma_j \\
0 & \text{otherwise}
\end{array}
\right.
$$

We'll compute positive and negative bar codes
from this boundary matrix $\partial$.

\subsection{Computing of Positive Bar codes $B_r^+(f;i)$}

We'll describe how to compute positive bar codes $B_r^+(f;i)$
of $X_{i,\infty}$ first.

There are five classes of cells in $X_{i,\infty}$:

$P_1=\{i\}$

$P_2=\{\sigma_i \mid \sigma\in X, t_{min}(\sigma)<i\text{ and }t_{max}(\sigma)>i\}$

$P_3=\{\sigma_{i,\infty} \mid \sigma\in X, t_{min}(\sigma)<i\text{ and }t_{max}(\sigma)>i\}$

$P_4=\{\sigma\in X \mid \dim(\sigma)>0\text{ and }t_{min}(\sigma)=i\}$

$P_5=\{\sigma\in X \mid t_{min}(\sigma)>i\}$

Order cells in $X_{i,\infty}$ first according to these
five classes, then according to initial order within
each class. We can easily construct a boundary
matrix $\partial_i^+$ from $\partial$ according to this order:
$$
\partial_i^+:=
\begin{tabular}{cc}
   & $\begin{array}{ccccccccccc}
        & P_1 & P_2 & P_3 && P_4 & P_5 &&
     \end{array}$
    \\
  $\begin{array}{c}
     P_1 \\
     P_2 \\
     P_3 \\
     P_4 \\
     P_5
   \end{array}
  $ & $\left(
\begin{tabular}{ccccc}
  \cline{2-2} \cline{4-4}
  & \multicolumn{1}{|c|}{$Q_1$} & & \multicolumn{1}{|c|}{$Q_3$} & \\
  \cline{2-2} \cline{4-4}
  \cline{2-2} \cline{3-3}
  & \multicolumn{1}{|c|}{$M_1$} & \multicolumn{1}{|c|}{$Q_2$} & & \\
  \cline{2-2} \cline{3-3}
  \cline{3-3}
  & & \multicolumn{1}{|c|}{} & &\\
  \cline{4-4}
  & & \multicolumn{1}{|c|}{$M_2$} & \multicolumn{1}{|c|}{} & \\
  \cline{5-5}
  & & \multicolumn{1}{|c|}{} & \multicolumn{1}{|c|}{$M_3$} & \multicolumn{1}{|c|}{$M_4$}\\
  \cline{3-5}
\end{tabular}
      \right)$
   \\
\end{tabular}
$$
where $M_i(1\leq i\leq 4)$ are correspondent submatrices of $\partial$,
elements in $Q_1=1$ iff $\dim\sigma=2$ and correspondent
column in $M_1$ contains only one ``1'', $Q_2$ is identity
matrix, elements in $Q_3=1$ iff $\dim\sigma=1$.

\begin{figure}[h!]
\centering
\includegraphics[height=50mm]{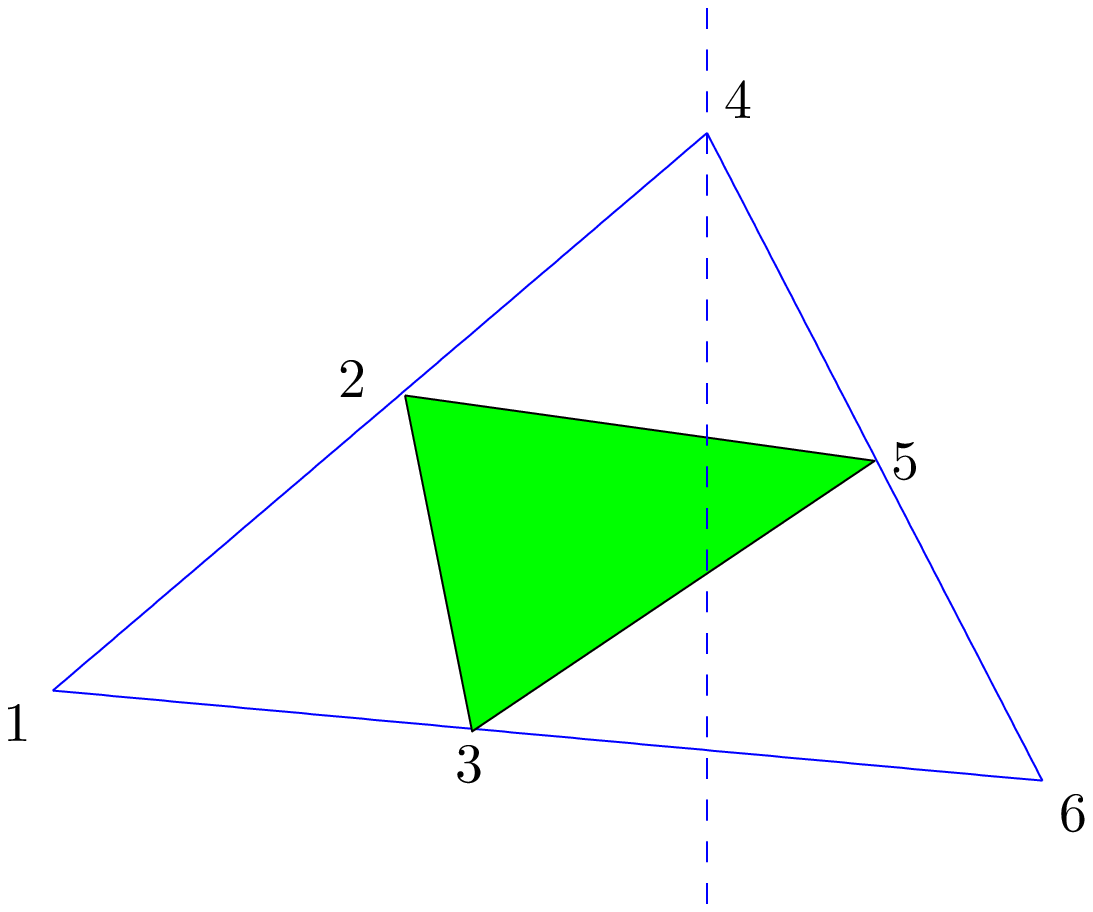}
\caption{Figure for Example \ref{example23}.}
\end{figure}

\begin{exmp}\label{example23}
\end{exmp}
Let $X$ be the same as in {Example \ref{example21}}
and $i=4$, we have

$P_1=\{4\}$

$P_2=\{25|_4,35|_4,36|_4,235|_4\}$

$P_3=\{25|_{4,\infty},35|_{4,\infty},36|_{4,\infty},235|_{4,\infty}\}$

$P_4=\{45\}$

$P_5=\{5,6,56\}$

%
%

\begin{center}
$$
\partial_4^+:=
\begin{tabular}{rl}
 & \hskip -.7cm
  $\begin{array}{c}
       P_1 \\[.18cm]
       4
     \end{array}

      \begin{array}{c}
        P_2 \\
        \overbrace{25|_4\;\;35|_4\;\;36|_4\;\;235|_4}
      \end{array}
     \begin{array}{c}
        P_3 \\
        \overbrace{25|_{4,\infty}\;\;35|_{4,\infty}\;\;36|_{4,\infty}\;\;235|_{4,\infty}}
      \end{array}
      \;
      \begin{array}{c}
        P_4 \\[.18cm]
        45
      \end{array}
      \;\;\;\;
      \begin{array}{c}
        P_5 \\
        \overbrace{5\;\;6\;\;56}
      \end{array}
$
 \\
$\begin{array}{lc}
    P_1 & \hskip -.9cm 4
    \\
    \\
    P_2 &  \hskip -.75cm \left\{\begin{array}{l}
\hskip -.2cm            25|_4 \\
\hskip -.2cm            35|_4 \\
\hskip -.2cm            36|_4 \\
\hskip -.2cm            235|_4
          \end{array}\right.
   \\
   \\
    P_3 & \hskip -.38cm \left\{\begin{array}{l}
\hskip -.2cm            25|_{4,\infty} \\
\hskip -.2cm            35|_{4,\infty} \\
\hskip -.2cm            36|_{4,\infty} \\
\hskip -.2cm            235|_{4,\infty}
          \end{array}\right.
  \\
  \\
  P_4 & \hskip -.9cm 45
  \\
  \\
  P_5 & \hskip -1.2cm \left\{\begin{array}{l}
\hskip -.2cm  5\\
\hskip -.2cm  6\\
\hskip -.2cm  56\\
  \end{array}\right.
  \\
\end{array}$
 & \hskip -.85cm
 $
 \left(
\begin{tabular}{ccccccccccccccccccccccc}
&&&&&&&&&&&&&&&&&&&&&\\[-.2cm]
\cline{3-9} \cline{20-20}
&&\multicolumn{1}{|c}{}&&&$Q_1$&&&\multicolumn{1}{c|}{}&&&&&&&&&&&\multicolumn{1}{|c|}{1}&&&\\
\cline{3-9} \cline{20-20}
&&&&&&&&&&&&&&&&&&&&&\\
\cline{3-9} \cline{11-19}
&&\multicolumn{1}{|c}{}&&&&&1&\multicolumn{1}{c|}{}&&\multicolumn{1}{|c}{1}&&&&&&&&\multicolumn{1}{c|}{}&&&&\\
&&\multicolumn{1}{|c}{}&&&&&1&\multicolumn{1}{c|}{}&&\multicolumn{1}{|c}{}&&1&&&&&&\multicolumn{1}{c|}{}&&&&\\
&&\multicolumn{1}{|c}{}&&&$M_1$&&&\multicolumn{1}{c|}{}&&\multicolumn{1}{|c}{$Q_2$}&&&&1&&&&\multicolumn{1}{c|}{}&&&&\\
&&\multicolumn{1}{|c}{}&&&&&&\multicolumn{1}{c|}{}&&\multicolumn{1}{|c}{}&&&&&&1&&\multicolumn{1}{c|}{}&&&&\\
\cline{3-9} \cline{11-19}
&&&&&&&&&&&&&&&&&&&&&&\\
\cline{11-19}
&&&&&&&&&&\multicolumn{1}{|c}{}&&&&&&1&&\multicolumn{1}{c|}{}&&&&\\
&&&&&&&&&&\multicolumn{1}{|c}{}&&&&&&1&&\multicolumn{1}{c|}{}&&&&\\
&&&&&&&&&&\multicolumn{1}{|c}{}&&&&&&&&\multicolumn{1}{c|}{}&&&&\\
&&&&&&&&&&\multicolumn{1}{|c}{}&&&&&&&&\multicolumn{1}{c|}{}&&&&\\
&&&&&&&&&&\multicolumn{1}{|c}{}&&&&$M_2$&&&&\multicolumn{1}{c|}{}&&&&\\
\cline{20-20}
&&&&&&&&&&\multicolumn{1}{|c}{}&&&&&&&&\multicolumn{1}{c|}{}&\multicolumn{1}{|c|}{}&&&\\
&&&&&&&&&&\multicolumn{1}{|c}{}&&&&&&&&\multicolumn{1}{c|}{}&\multicolumn{1}{|c|}{$M_3$}&&&\\
\cline{22-23}
&&&&&&&&&&\multicolumn{1}{|c}{1}&&1&&&&&&\multicolumn{1}{c|}{}&\multicolumn{1}{|c|}{1}&&\multicolumn{1}{|c}{}$M_4$&\multicolumn{1}{c|}{1}\\
&&&&&&&&&&\multicolumn{1}{|c}{}&&&&1&&&&\multicolumn{1}{c|}{}&\multicolumn{1}{|c|}{}&&\multicolumn{1}{|c}{}&\multicolumn{1}{c|}{1}\\
&&&&&&&&&&\multicolumn{1}{|c}{}&&&&&&&&\multicolumn{1}{c|}{}&\multicolumn{1}{|c|}{}&&\multicolumn{1}{|c}{}&\multicolumn{1}{c|}{}\\
\cline{11-19} \cline{20-20} \cline{22-23}
&&&&&&&&&&&&&&&&&&&&&&\\
\end{tabular}
  \right)
  $
  \\
\end{tabular}
$$
\end{center}

It's convenient to construct $\partial_4^+$ using the above order.
However, since this order is neither topologically consistent
nor filtration compatible with $Y_i\subseteq Y_{i,i+1}\subseteq \cdots\subseteq Y_{i,N}$
(see {Definition 4.2},\cite{BDD}), we have to reorder those cells
in $X_{i,\infty}$.

\begin{defn} \textbf{Positive Order}

Order cells in $X_{i,\infty}$ first according to $t_{max}$ increasingly,
then according to the initial order, we call this order
\emph{positive order}.

For example,$\tau_1$ and $\tau_2$ are two cells in $X_{i,\infty}$.
If $t_{max}(\tau_1)<t_{max}(\tau_2)$,
then $\tau_1<\tau_2$ in positive order.
If $t_{max}(\tau_1)=t_{max}(\tau_2)$,
according to the note after {Definition \ref{def21}},
$\tau_j(j=1,2)$ can be uniquely represented by $\sigma_j$,$\sigma_j|_i$ or $\sigma_j|_{i,\infty}$
for some $\sigma_j\in X$.
If $\sigma_1<\sigma_2$ in initial order,
we require $\tau_1<\tau_2$ in positive order.

\end{defn}

Notice that $P_1$ and $P_2$ do not change in positive
order. The cells in $P_3$, $P_4$ and $P_5$ permute. Denote
the cells in $P_3$, $P_4$ and $P_5$ in positive order by $\widetilde{P}_{3,4,5}$.
Denote the boundary matrix of $X_{i,\infty}$ by $\widetilde{\partial}_i^+$ when the
cells are in positive order.

\begin{exmp}\label{example24}
Consider {Example \ref{example23}}, reorder all cells in $X_{4,\infty}$
in positive order, we have
\end{exmp}

\begin{center}

$$
\widetilde{\partial_4}^+:=
\begin{tabular}{rl}
 & \hskip -.9cm \;\;$\begin{array}{c}
       P_1 \\[.19cm]
       4
     \end{array}
     \begin{array}{c}
        P_2 \\
        \overbrace{25|_4\;\;35|_4\;\;36|_4\;\;235|_4}
      \end{array}
     \begin{array}{c}
        \widetilde{P}_{3,4,5} \\
        \overbrace{5\;\;25|_{4,\infty}\;\;35|_{4,\infty}\;\;45\;\;235|_{4,\infty}\;\;6\;\;36|_{4,\infty}\;\;56}
      \end{array}
 $
 \\
$\begin{array}{lc}
   \\
    P_1 & \hskip -1.2cm 4   \\
    \\
    P_2 & \hskip -.9cm
    \left\{ \hskip -.2cm \begin{array}{l}
            25|_4 \\
            35|_4 \\
            36|_4 \\
            235|_4
          \end{array}\right.   \\
   \\
    \widetilde{P}_{3,4,5} & \hskip -.4cm
    \left\{ \hskip -.2cm \begin{array}{l}
            5\\
            25|_{4,\infty}\\
            35|_{4,\infty}\\
            45\\
            235|_{4,\infty}\\
            6\\
            36|_{4,\infty}\\
            56
          \end{array}\right.  \\
  \\
 \end{array}$
 & \hskip -.85cm
$\left(
\begin{tabular}{cccccccccccccccccccccccccc}
&&&&&&&&&&&&&&&&&&&&&&&&&\\[-.18cm]
&&&&&&&&&&&&&&&&&&1&&&&&&\\
&&&&&&&&&&&&&&&&&&&&&&&&&\\
&&&&&&&&&1&&&&\hskip -.3cm 1&&&&&&&&&&&&\\
&&&&&&&&&1&&&&&&&\hskip -.15cm 1&&&&&&&&&\\
&&&&&&&&&&&&&&&&&&&&&&&\hskip .2cm 1&&\\
&&&&&&&&&&&&&&&&&&&&1&&&&&\\
&&&&&&&&&&&&&&&&&&&&&&&&&\\[-.1cm]
&&&&&&&&&&&&&\hskip -.3cm 1&&&\hskip -.15cm 1&&1&&&&&&&1\\
&&&&&&&&&&&&&&&&&&&&1&&&&&\\
&&&&&&&&&&&&&&&&&&&&1&&&&&\\
&&&&&&&&&&&&&&&&&&&&&&&&&\\
&&&&&&&&&&&&&&&&&&&&&&&&&\\
&&&&&&&&&&&&&&&&&&&&&&&\hskip .2cm 1&&1\\
&&&&&&&&&&&&&&&&&&&&&&&&&\\
&&&&&&&&&&&&&&&&&&&&&&&&&\\
&&&&&&&&&&&&&&&&&&&&&&&&&\\[-.1cm]
\end{tabular}
\hskip .15cm \right)$
  \\
\end{tabular}
$$

\end{center}

From $\widetilde{\partial}_i^+$, we will get its reduced form $R(\partial_i^+)$ and read
the bar code directly from it.

Given an $m\times m$ matrix $\partial$ with $\mathbb{Z}_2$ coefficients. Let $low(j)$
be the row index of lowest $1$ in column $j$. If the
entire column is zero, then $low(j)$ is undefined.
We call $\partial$ \emph{reduced} if for any two non-zero columns
$j$ and $j_0$, we have $low(j)\neq low(j_0)$. See page 153,\cite{EH}.
The following algorithm reduces $\partial$ by adding
columns from left to right.

$$
\textbf{Algorithm 2.1.}
$$
\begin{center}
\begin{tabular}{|l|}
  \hline
  \textrm{for $j=1$ to $m-1$ do}\\
  \textrm{\;\;\;\;\;\;if column $j$ is nonzero do}\\
  \textrm{\;\;\;\;\;\;\;\;\;\;\;\;while there exists $j_0>j$ with $low(j_0)=low(j)$ do}\\
  \textrm{\;\;\;\;\;\;\;\;\;\;\;\;\;\;\;\;\;\;add column $j$ to column $j_0$}\\
  \textrm{\;\;\;\;\;\;\;\;\;\;\;\;endwhile}\\
  \textrm{\;\;\;\;\;\;\;\;\;\;\;\;let $\{i_1,\cdots,i_k\}=\{i|1\leq i\leq j-1, low(i)<low(j)\}$}\\
   \textrm{\;\;\;\;\;\;\;\;\;\;\;\;\;\;\;\;\;\;\;\;\;\;\;\;such that $low(i_1)>low(i_2)>\cdots>low(i_k)$}\\
  \textrm{\;\;\;\;\;\;\;\;\;\;\;\;for l=1 to k do}\\
  \textrm{\;\;\;\;\;\;\;\;\;\;\;\;\;\;\;\;\;\;while there exists $j_0>j$ with $low(j_0)=low(l)$ do}\\
  \textrm{\;\;\;\;\;\;\;\;\;\;\;\;\;\;\;\;\;\;\;\;\;\;\;\;add column $l$ to column $j_0$}\\
  \textrm{\;\;\;\;\;\;\;\;\;\;\;\;\;\;\;\;\;\;endwhile}\\
  \textrm{\;\;\;\;\;\;\;\;\;\;\;\;endfor}\\
  \textrm{\;\;\;\;\;\;endif}\\
  \textrm{endfor}\\
  \hline
\end{tabular}
\end{center}

\begin{exmp}\label{example25}
The reduced form $R(\widetilde{\partial}_4^+)$ of ${\partial}_4^+$ in {Example \ref{example24}} is:
\end{exmp}

\begin{center}

$$
R(\widetilde{\partial_4}^+):=
\begin{tabular}{rl}
 & \hskip -.9cm \;\;$\begin{array}{c}
       P_1 \\[.19cm]
       4
     \end{array}
     \begin{array}{c}
        P_2 \\
        \overbrace{25|_4\;\;35|_4\;\;36|_4\;\;235|_4}
      \end{array}
     \begin{array}{c}
        \widetilde{P}_{3,4,5} \\
        \overbrace{5\;\;25|_{4,\infty}\;\;35|_{4,\infty}\;\;45\;\;235|_{4,\infty}\;\;6\;\;36|_{4,\infty}\;\;56}
      \end{array}
 $
 \\
$\begin{array}{lc}
   \\
    P_1 & \hskip -1.2cm 4   \\
    \\
    P_2 & \hskip -.9cm
    \left\{ \hskip -.2cm \begin{array}{l}
            25|_4 \\
            35|_4 \\
            36|_4 \\
            235|_4
          \end{array}\right.   \\
   \\
    \widetilde{P}_{3,4,5} & \hskip -.4cm
    \left\{ \hskip -.2cm \begin{array}{l}
            5\\
            25|_{4,\infty}\\
            35|_{4,\infty}\\
            45\\
            235|_{4,\infty}\\
            6\\
            36|_{4,\infty}\\
            56
          \end{array}\right.  \\
  \\
 \end{array}$
 & \hskip -.85cm
$\left(
\begin{tabular}{cccccccccccccccccccccccccc}
&&&&&&&&&&&&&&&&&&&&&&&&&\\[-.18cm]
&&&&&&&&&&&&&&&&&&1&&&&&&\\
&&&&&&&&&&&&&&&&&&&&&&&&&\\
&&&&&&&&&1&&&&\hskip -.3cm 1&&&&&1&&&&&&&1\\
&&&&&&&&&1&&&&&&&\hskip -.15cm &&&&&&&&&\\
&&&&&&&&&&&&&&&&&&&&&&&\hskip .2cm 1&&1\\
&&&&&&&&&&&&&&&&&&&&1&&&&&\\
&&&&&&&&&&&&&&&&&&&&&&&&&\\[-.1cm]
&&&&&&&&&&&&&\hskip -.3cm 1&&&\hskip -.15cm &&&&&&&&&\\
&&&&&&&&&&&&&&&&&&&&1&&&&&\\
&&&&&&&&&&&&&&&&&&&&1&&&&&\\
&&&&&&&&&&&&&&&&&&&&&&&&&\\
&&&&&&&&&&&&&&&&&&&&&&&&&\\
&&&&&&&&&&&&&&&&&&&&&&&\hskip .2cm 1&&\\
&&&&&&&&&&&&&&&&&&&&&&&&&\\
&&&&&&&&&&&&&&&&&&&&&&&&&\\
&&&&&&&&&&&&&&&&&&&&&&&&&\\[-.1cm]
\end{tabular}
\hskip .15cm \right)$
  \\
\end{tabular}
$$

\end{center}

We can read bar codes directly from the reduced form $R(\widetilde{\partial}_i^+)$.
For the first $|P_1|+|P_2|$ columns of $R(\widetilde{\partial}_i^+)$, if a column $j$
is zero, check if any column $j_0$ on the right satisfies
$low(j_0)=j$. If there doesn't exist such a column $j_0$,
then we have an interval $[i,\infty)$ in the bar code $B_r^+(f;i)$,
where $r=\dim(\sigma_j)$, $\sigma_j$ is the cell correspond to column $j$.
If there exists such a column $j_0$ and $j_0>|P_1|+|P_2|$,
then we have an interval $[i,j')$ in the bar code $B_r^+(f;i)$
where $r=\dim(\sigma_j)$, $j'=t_{max}(\sigma_{j_0})$, $\sigma_j$ and $\sigma_{j_0}$ are
the cells correspond to column $j$ and $j_0$ respectively.

\begin{exmp}
We can read bar code from $R(\widetilde{\partial}_4^+)$ in {Example \ref{example25}}.
$$
B_0^+(f;4)=\{[4,\infty),[4,5),[4,6)\}
$$
\end{exmp}

\subsection{Computing of Negative Bar codes $B_r^-(f;i)$}

About computing of negative bar codes $B_r^-(f;i)$,
one proceed in a similar manner. Precisely, first
replace the filtration $X_i\subseteq X_{i-1,i}\subseteq\cdots\subseteq X_{1,i}$
by $X_i\subseteq Z_{i-1,i}\subseteq\cdots\subseteq Z_{1,i}$ where $Z_{k,i}$ is the
set of all cells in $X_{-\infty,i}$ that are
contained in $X_{k,\infty}$. Second, construct the boundary
matrix $\partial_i^-$. There are five classes of cells in $X_{-\infty,i}$:

$N_1=P_1=\{i\}$, $N_2=P_2$, $N_3=P_3$

$N_4=\{\sigma\in X \mid \dim(\sigma)>0\text{ and }t_{max}(\sigma)=i\}$

$N_5=\{\sigma\in X \mid t_{max}(\sigma)<i\}$

If we order all cells in $X_{-\infty,i}$ first according to these
five classes then according to the initial order, we
can easily construct $\partial_i^-$ from the submatrices of $\partial$.
Third, order cells in $X_{-\infty, i}$ first according to $t_{min}$
decreasingly then according to the initial order(so called
\emph{negative order}), the boundary matrix of $X_{-\infty,i}$ becomes
$\widetilde{\partial}_i^-$. Fourth, apply {Algorithm 2.1} to $\widetilde{\partial}_i^-$ and get reduced
form $R(\widetilde{\partial}_i^-)$. Finally, read the negative bar code $B_r^-(f;i)$
from the reduced form $R(\widetilde{\partial}_i^-)$.

Notice: The first $|P_1|+|P_2|$ columns of $\partial_i^+$ and $\partial_i^-$ are
essentially the same, since $P_1\cup P_2$ and $N_1\cup N_2$ contain
the same cells in the same order. It's the same
situation for $\widetilde{\partial}_i^+$ and $\widetilde{\partial}_i^-$, since reordering cells
in positive and negative order will not affect the
first $|P_1|+|P_2|$ cells. The same thing happened for the
reduced form $R(\widetilde{\partial}_i^+)$ and $R(\widetilde{\partial}_i^-)$, since we apply the
same algorithm to them and we only add columns
from left to right in that algorithm. We have
the same amount of zero columns in $R(\widetilde{\partial}_i^+)$ and
$R(\widetilde{\partial}_i^-)$, which are in the same position and correspond
to the same cycles of $X_i$. Therefore, the zero
columns in $R(\widetilde{\partial}_i^+)$ and $R(\widetilde{\partial}_i^-)$ which represent generators
of $H_{\ast}(X_i)$ are also in one to one correspondence. We
can pair up intervals in positive and negative bar codes
according to generators.

\subsection{Numerical Experiments}

All the bar code in this section are generated by the Matlab code given in the {Appendix}.
All intervals in positive and negative bar codes paired up in the given order.

\begin{figure}[h!]
\centering
 \includegraphics[height=45mm]{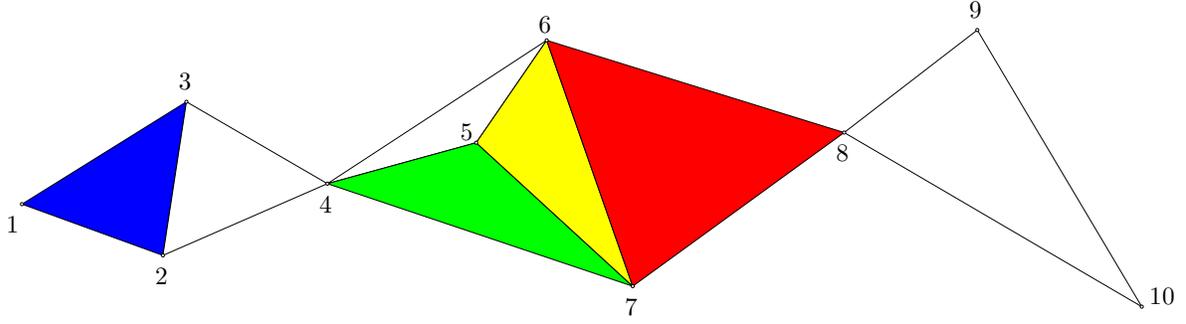}
\caption{$f^1$: a linear map on a $2$-simplicial complex.} \label{figure461}
\end{figure}

\begin{exmp}
Bar codes of $f^1$ in Figure \ref{figure461}.
\vskip -3mm
$$
\begin{array}{ll}
  B_0^+(f^1;1) = \{\;[1,\infty)\;\} & B_0^-(f^1;1) = \{\;(-\infty,1]\;\}\\
  B_0^+(f^1;2) = \{\;[2,\infty)\;\} & B_0^-(f^1;2) = \{\;(-\infty,2]\;\}\\
  B_0^+(f^1;3) = \{\;[3,\infty),\;[3,4)\;\} & B_0^-(f^1;3) = \{\;(-\infty,3],\;(2,3]\;\}\\
  B_0^+(f^1;4) = \{\;[4,\infty)\;\} & B_0^-(f^1;4) = \{\;(-\infty,4]\;\}\\
  B_0^+(f^1;5) = \{\;[5,\infty),\;[5,6)\;\} & B_0^-(f^1;5) = \{\;(-\infty,5],\;(4,5]\;\}\\
  B_0^+(f^1;6) = \{\;[6,\infty)\;\} & B_0^-(f^1;6) = \{\;(-\infty,6]\;\}\\
  B_0^+(f^1;7) = \{\;[7,\infty)\;\} & B_0^-(f^1;7) = \{\;(-\infty,7]\;\}\\
  B_0^+(f^1;8) = \{\;[8,\infty)\;\} & B_0^-(f^1;8) = \{\;(-\infty,8]\;\}\\
  B_0^+(f^1;9) = \{\;[9,\infty),\;[9,10)\;\} & B_0^-(f^1;9) = \{\;(-\infty,9],\;8,9]\;\}\\
  B_0^+(f^1;10) = \{\;[10,\infty)\;\} & B_0^-(f^1;10) = \{\;(-\infty,10]\;\}
\end{array}
$$
\end{exmp}

\begin{figure}[h!]
\centering
 \includegraphics[height=45mm]{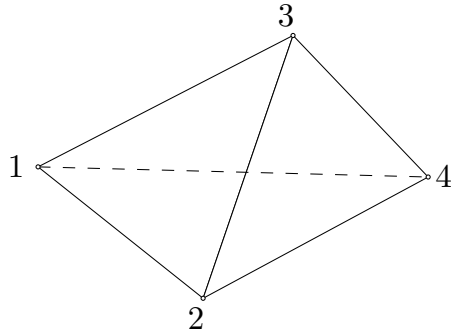}
\caption{$f^2$: a linear map on the surface of a tetrahedron.} \label{figure342}
\end{figure}

\begin{exmp}
Bar codes of $f^2$ in Figure \ref{figure342}.
\vskip -3mm
$$
\begin{array}{ll}
  B_0^+(f^2;1) = \{\;[1,\infty)\;\} & B_0^-(f^2;1) = \{\;(-\infty,1]\;\}\\
  B_0^+(f^2;2) = \{\;[2,\infty)\;\} & B_0^-(f^2;2) = \{\;(-\infty,2]\;\}\\
  B_0^+(f^2;3) = \{\;[3,\infty)\;\} & B_0^-(f^2;3) = \{\;(-\infty,3]\;\}\\
  B_0^+(f^2;4) = \{\;[4,\infty)\;\} & B_0^-(f^2;4) = \{\;(-\infty,4]\;\}\\
  B_1^+(f^2;2) = \{\;[2,4)\;\} & B_1^-(f^2;2) = \{\;(1,2]\;\}\\
  B_1^+(f^2;3) = \{\;[3,4)\;\} & B_1^-(f^2;3) = \{\;(1,3]\;\}\\
\end{array}
$$
\end{exmp}

\begin{figure}[h!]
\centering
 \includegraphics[height=45mm]{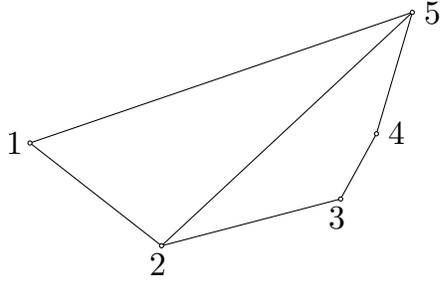}
\caption{$f^3$: a linear map on a planar graph.} \label{figure463}
\end{figure}

\begin{exmp}
Bar codes of $f^3$ in Figure \ref{figure463}.
\vskip -3mm
$$
\begin{array}{ll}
  B_0^+(f^3;1) = \{\;[1,\infty)\;\} & B_0^-(f^3;1) = \{\;(-\infty,1]\;\}\\
  B_0^+(f^3;2) = \{\;[2,\infty),\;[2,5)\;\} & B_0^-(f^3;2) = \{\;(-\infty,2],\;(1,2]\;\}\\
  B_0^+(f^3;3) = \{\;[3,\infty),\;[3,5),\;[3,5)\;\} & B_0^-(f^3;3) = \{\;(-\infty,3],\;(1,3],\;(2,3]\;\}\\
  B_0^+(f^3;4) = \{\;[4,\infty),\;[4,5),\;[4,5)\;\} & B_0^-(f^3;4) = \{\;(-\infty,4],\;(1,4],\;(2,4]\;\}\\
  B_0^+(f^3;5) = \{\;[5,\infty)\;\} & B_0^-(f^3;5) = \{\;(-\infty,5]\;\}\\
\end{array}
$$
\end{exmp}

\begin{figure}[h!]
\centering
 \includegraphics[height=45mm]{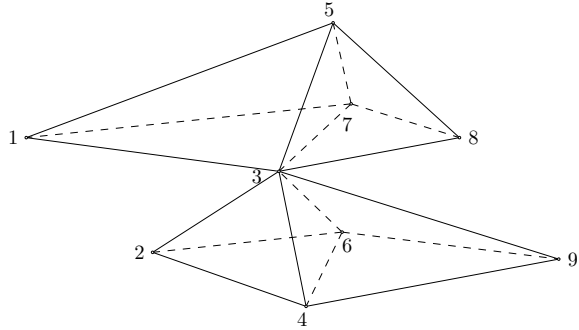}
\caption{$f^4$: a linear map on the union of the surfaces of two tetrahedra.} \label{figure464}
\end{figure}

\begin{exmp}
Bar codes of $f^4$ in Figure \ref{figure464}.
\vskip -3mm
$$
\begin{array}{ll}
  B_0^+(f^4;1) = \{\;[1,\infty)\;\} & B_0^-(f^4;1) = \{\;(-\infty,1]\;\}\\
  B_0^+(f^4;2) = \{\;[2,\infty),\;[2,3)\;\} & B_0^-(f^4;2) = \{\;(-\infty,2],\;(-\infty,2]\;\}\\
  B_0^+(f^4;3) = \{\;[3,\infty)\;\} & B_0^-(f^4;3) = \{\;(-\infty,3]\;\}\\
  B_0^+(f^4;4) = \{\;[4,\infty),\;[4,\infty)\;\} & B_0^-(f^4;4) = \{\;(-\infty,4],\;(3,4]\;\}\\
  B_0^+(f^4;5) = \{\;[5,\infty),\;[5,\infty)\;\} & B_0^-(f^4;5) = \{\;(-\infty,5],\;(3,5]\;\}\\
  B_0^+(f^4;6) = \{\;[6,\infty),\;[6,\infty)\;\} & B_0^-(f^4;6) = \{\;(-\infty,6],\;(3,6]\;\}\\
  B_0^+(f^4;7) = \{\;[7,\infty),\;[7,\infty)\;\} & B_0^-(f^4;7) = \{\;(-\infty,7],\;(3,7]\;\}\\
  B_0^+(f^4;8) = \{\;[8,\infty),\;[8,\infty)\;\} & B_0^-(f^4;8) = \{\;(-\infty,8],\;(3,8]\;\}\\
  B_0^+(f^4;9) = \{\;[9,\infty)\;\} & B_0^-(f^4;9) = \{\;(-\infty,9]\;\}
  \\
  B_1^+(f^4;2) = \{\;[2,8)\;\} & B_1^-(f^4;2) = \{\;(1,2]\;\}\\
  B_1^+(f^4;3) = \{\;[3,8),\;[3,9)\;\} & B_1^-(f^4;3) = \{\;(1,3],\;(2,3]\;\}\\
  B_1^+(f^4;4) = \{\;[4,9),\;[4,8)\;\} & B_1^-(f^4;4) = \{\;(2,4],\;(1,4]\;\}\\
  B_1^+(f^4;5) = \{\;[5,8),\;[5,9)\;\} & B_1^-(f^4;5) = \{\;(1,5],\;(2,5]\;\}\\
  B_1^+(f^4;6) = \{\;[6,9),\;[6,8)\;\} & B_1^-(f^4;6) = \{\;(2,6],\;(1,6]\;\}\\
  B_1^+(f^4;7) = \{\;[7,9),\;[7,8)\;\} & B_1^-(f^4;7) = \{\;(2,7],\;(1,7]\;\}\\
  B_1^+(f^4;8) = \{\;[8,9)\;\} & B_1^-(f^4;8) = \{\;(2,8]\;\}\\
\end{array}
$$
\end{exmp}

\end{document}